%% file: main.tex
\documentclass[%
    reprint, %
    superscriptaddress,
    amsmath, amssymb, %
    aps,    %
    prx,    %
    10pt,   %
    floatfix,
]{revtex4-2}

\input{preamble}

\begin{document}
\setcounter{secnumdepth}{0} %

\title{Quantum Memory and Autonomous Computation in Two Dimensions}

\author{Gesa D\"unnweber}
\email{gesa.duennweber@mpq.mpg.de}
\affiliation{Max Planck Institute of Quantum Optics, Hans Kopfermann Str. 1, 85748 Garching, Germany}
\affiliation{Munich Center for Quantum Science and Technology, Schellingstr. 4, 80799 Munich, Germany}
\affiliation{Technical University of Munich, School of Natural Sciences, Physics Department, 85748 Garching, Germany}

\author{Georgios Styliaris}
\affiliation{Max Planck Institute of Quantum Optics, Hans Kopfermann Str. 1, 85748 Garching, Germany}
\affiliation{Munich Center for Quantum Science and Technology, Schellingstr. 4, 80799 Munich, Germany}

\author{Rahul Trivedi}
\email{rahul.trivedi@mpq.mpg.de}
\affiliation{Max Planck Institute of Quantum Optics, Hans Kopfermann Str. 1, 85748 Garching, Germany}
\affiliation{Munich Center for Quantum Science and Technology, Schellingstr. 4, 80799 Munich, Germany}

\begin{abstract}
Quantum error correction (QEC) normally requires active maintenance using measurements, classical processing and externally controlled correction schedules. 
Passive QEC, by contrast, was previously established only in unphysical spatial dimensions. 
Here, we give an explicit scheme for autonomous quantum error correction and computation in two dimensions, formulated as a dissipative quantum cellular automaton with a fixed, local and translation-invariant update rule.
The construction uses hierarchical, self-simulating control elements based on ideas from the seminal classical results of \gacs (1986, 1989) together with a measurement-free concatenated quantum code.
We prove the existence of a nonzero noise threshold under a local noise model.
Below this threshold, logical errors on encoded initial states are suppressed exponentially with increasing system size and the memory lifetime diverges in the thermodynamic limit.
As an extension, we outline how this mechanism can also be implemented in continuous time by a time-independent, translation-invariant local Lindbladian.
Quantum circuits can moreover be encoded in the automaton's initial state and executed fault-tolerantly with polylogarithmic overhead under the same circuit-independent rule. Our scheme thus realizes a self-correcting quantum computer capable of universal computation.
\end{abstract}

\maketitle

Quantum information is fragile. 
In all currently known hardware architectures, the physical states of a quantum computer are rapidly corrupted by noise.
Quantum error correction (QEC) counteracts this fragility by encoding logical states across many physical qubits and repeatedly detecting and reversing errors. Contemporary fault-tolerant schemes use \textit{actively} controlled correction cycles with syndrome extraction and classical processing leading to a substantial resource overhead~\cite{fowler2012, bravyi2024, google2025}.\linebreak A longstanding question is whether the protection and processing of quantum information can instead be achieved \textit{passively} by embedding error correction directly into the autonomous dynamics of the quantum hardware itself. 
From the viewpoint of condensed matter physics, 
this asks whether a local, translation-invariant quantum system can support robust memory and stable non-equilibrium order without external signals. 

In classical information processing, the analogous problem led to seminal results by Toom and \gacs which showed that reliable storage and computation can emerge from noisy local dynamics without external control~\cite{toom1980, gacs1986, gacs1989}. 
The quantum case is subtler~\cite{brown2016, bombin2013}. 
Passive (\textit{self-correcting}) quantum memories are known in four spatial dimensions (for example the four-dimensional toric code under thermal noise~\cite{dennis2002, alicki2010}), 
but are ruled out in two dimensions by no-go results for the case of local stabilizer Hamiltonian models~\cite{bravyi2009, temme2017}.
In lower %
dimensions, prior works have considered a variety of codes with numerical indications for passive QEC capabilities~\cite{haah2011, bravyi2013, herold2015, lin2024, roberts2025}. 
Analytical approaches have been rather limited. Known results crucially rely either on noiseless classical processing assumed to be available at each site~\cite{harrington2004,lake2025, selub2026}, or on explicitly hard-coded interactions whose space-time dependence %
takes the place of an external controller~\cite{balasubramanian2026a}.
Thus, whether scalable, truly autonomous stabilization of quantum information is achievable in low spatial dimensions has remained open.

Here, we answer this question constructively. We present a two-dimensional quantum system with 
fixed, local and translation-invariant interactions 
that autonomously stabilizes encoded quantum information under local noise.
Our result does not contradict the known no-go restrictions~\cite{bravyi2009} because it does not fall under the conventional assumption of a commuting-stabilizer Hamiltonian coupled to thermal noise (typically modeled, in the Markovian limit, by Davies generators). %
Instead, our mechanism uses error-correcting dynamics that are realized by engineered \textit{dissipation} in which fixed local couplings to an environment %
continuously drive the system towards the prescribed code space~\cite{verstraete2009, pastawski2011}.

\begin{figure*}[tb]
    \centering
    \includegraphics[width=0.95\textwidth]{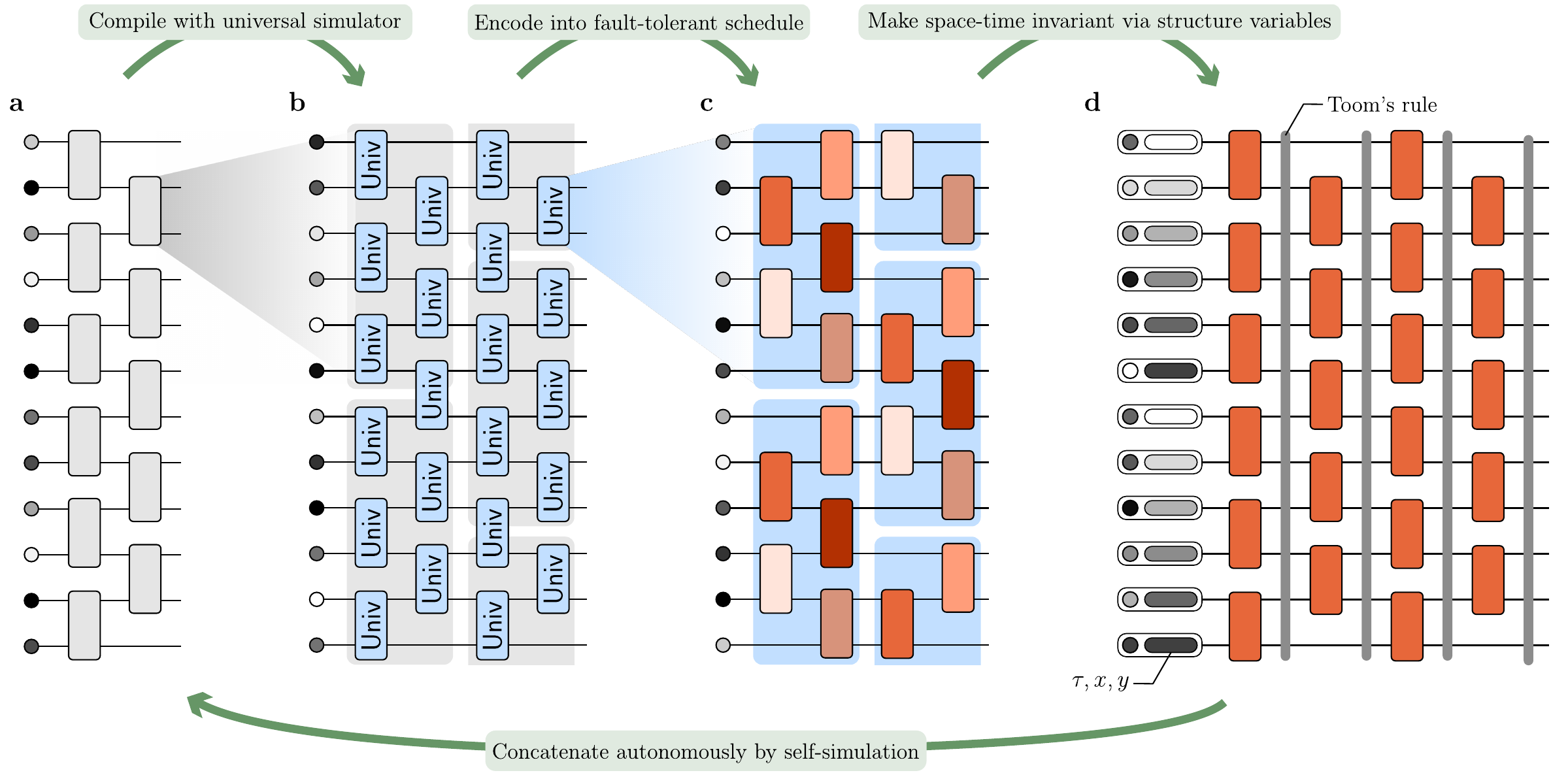}
    \caption{\textbf{Construction of the self-correcting quantum cellular automaton.}
    We build a system that robustly simulates its own interaction rule, thereby enabling autonomous stabilization of states encoded in a corresponding self-similar hierarchical code.
    To construct this, consider an arbitrary quantum cellular automaton (\textbf{a}) which is simulated by a fixed universal update channel $\Univ$ (\textbf{b}). Encoding with a fault-tolerant quantum error-correcting code produces a constant-size schedule (\textbf{c}) for universal simulation within the one-level encoded space. To recover space-time invariance, each site stores its believed position in local structure registers that are corrected by classical Toom's rule (\textbf{d}). For sufficiently large local dimension, the resulting uniform channel can be chosen as the QCA being simulated, establishing a self-simulating fixed point.
    The initial state on the top encoding level dictates the target quantum circuit executed within the logical space.
    \vspace{-1 ex}
    } 
    \label{fig:construction-sketch}
\end{figure*}

Our construction is inspired by \gacs's self-simulating classical cellular automata~\cite{gacs1986, gacs1989}. A cellular automaton is a discrete-time dynamical system specified by a fixed local transition rule that is applied uniformly in space and time. We build a dissipative quantum cellular automaton (QCA)~\cite{schumacher2004, guedes2024, arrighi2019}, in which the local update applies a quantum channel, and design it to correct itself across a hierarchy of scales.
This hierarchy supports a concatenated quantum code~\cite{aharonov1997, gottesman2000} whose error syndromes are accessed and processed autonomously. Local gadgets implement error correction at the lowest level, while higher levels automatically enact coarse-grained stabilization through the same fixed physical interactions. The resulting code is inherently non-local, with logical degrees of freedom supported on recursive structures, yet the enforcing interactions remain strictly local (and uniform).
We analytically establish robustness by proving the existence of a nonzero noise threshold below which the memory lifetime diverges in the thermodynamic limit. 
Moreover, the architecture is expressive enough for quantum circuits to be embedded in the initial state and autonomously executed within the stabilized tower of simulations, yielding self-correcting \textit{computation} as well as memory.

\vspace{-2 ex}
\notoc\section{Main theorem and properties} \label{sec:statement}

The results can be stated as follows:
There exists a two-dimensional quantum system that, on suitably encoded initial states, performs self-correcting, fault-tolerant universal quantum computation.
The interactions are local, translation-invariant and independent of the desired accuracy and the physical system size. %

The dynamics can be realized %
by a dissipative quantum cellular automaton with a fixed local update rule applied in discrete time-steps, and the protection holds under a noise model that allows every local operation to fail with strength $\eta$ per space-time location in an arbitrary way (including correlated and non-Markovian faults mediated by a common environment~\cite{aliferis2006}).
We prove existence of a nonzero threshold value of $\eta$ below which the induced logical noise can be made arbitrarily small by increasing the system size. In particular, the memory lifetime diverges in the thermodynamic limit.

For computation,
any desired translation-invariant quantum circuit of depth $D$ on $N$ qubits (assumed to be given over a fixed universal gate/channel set) can be compiled into the initial state and executed fault-tolerantly with accuracy $\delta$ using $\bigO(N \polylog( N D/ \delta))$ physical qudits and total physical runtime $\bigO(D \polylog( N D/ \delta))$, all under the same circuit-independent interactions.
The required encoded initial state can be prepared with a circuit of depth $\bigO(\polylog(ND/\delta))$ and the final logical output is recovered by decoding the highest level of the hierarchy.
To simulate a spatially varying circuit, one first maps it into a translation-invariant circuit with a spatially varying initial state and then simulates this under the stated polylogarithmic overhead.

We now provide the key ideas of our scheme. The formal definition of the local rule and the threshold proof can be found in the Supplementary Information.

\vspace{4 ex}
\notoc\section{Hierarchical construction}

The construction can be pictured as a recursively layered hierarchical mechanism, analogous to the concatenated codes used in active QEC~\cite{aharonov1997, gottesman2000}. The core difficulty is to correct not only the single-level quantum data but also the space-time organization of the fault-tolerant protocol itself. This organization is repaired by fixed local dynamics that simultaneously protect (and process) the encoded data. 
The central idea is to build a quantum system that performs \textit{self-simulation} (as in \gacs's classical constructions~\cite{gacs1986, gacs1989}), meaning that one may adopt a coarse-grained view in which blocks of qudits encode higher-level qudits that in turn evolve under simulated interactions mirroring the original, physical interactions. 
Once such self-simulation is robustly implemented, the concatenation hierarchy autonomously builds itself up to the largest scale permitted by the finite-sized system. 

We define the transition rule of our self-correcting QCA step-by-step by going through one loop of self-simulation (Fig.~\ref{fig:construction-sketch}).
Given an arbitrary quantum cellular automaton, we construct a noise-resistant simulation of it, thereby obtaining a new QCA. We then verify that the self-simulation can be closed, i.e. that one may build a system where equating the original rule with the final simulating interactions is well-defined.

\vspace{4 ex}
\notoc\subsection{Programming a fixed local rule}

We begin with uncorrected simulation. Following the strategy laid out in Ref.~\cite{arrighi2008}, we build a \textit{universal} quantum cellular automaton, $\Univ$, which is \textit{classically programmable}, i.e. which can simulate any prescribed local dynamics dictated by a `program' written into the initial state as a product state, and for which this program always remains unentangled with the quantum data.
The transition rule of $\Univ$ is fixed, once and for all, so only the initial program tape depends on the target automaton. A block of physical cells represents one simulated, `higher-level' cell on which the desired evolution is enacted via a sequence of lower-level steps (one macro-step). In our setting, we start with a given QCA with a potentially large local dimension. Its complexity is shifted entirely into the spatial and temporal overhead of the simulation, using physical qudits with constant local dimension fixed by $\Univ$. 
At this stage the construction is only a simulator and errors in the simulated data are simply inherited by the next scale.

\vspace{4 ex}
\notoc\subsection{Embedding fault-tolerant correction}

To make the simulation robust, we replace each update of $\Univ$ by a \textit{concatenated error-correction} gadget. It is well-known that there exist quantum error-correcting codes, equipped with a universal set of fault-tolerant gates, which can be recursively iterated (concatenated) to achieve better and better correction capabilities~\cite{aharonov1997}. We use a measurement-free, nearest-neighbor-local scheme~\cite{aharonov1997, gottesman2000} and, by compiling the fixed channel $\Univ$ with it, construct a space-time-dependent schedule that implements fault-tolerant universal interactions in the encoded space (incurring only a constant additional space-time overhead).
However, this schedule differs from site to site within each code-block and, since we will concatenate the code to reduce the effective logical noise, this would result in a spatially varying scheme.
To recover space-time invariance, we introduce local \textit{structure registers} $(\tau,x,y)$ that record the locally believed space-time position at each site. This yields a local channel which performs individual steps of the encoded $\Univ$ schedule according to the structural state. When the structure registers are correct (modulo the fixed single-level block size and cycle period), this exactly reproduces the intended encoded update on the data registers. As before, the logical circuit to be executed is stored only in the initial state (now in the form of an encoded program tape), whereas the local interaction itself remains invariant.

\vspace{4 ex}
\notoc\subsection{Protecting the control layer}

\begin{figure}[tb]
    \centering
    \includegraphics[width=0.95\linewidth]{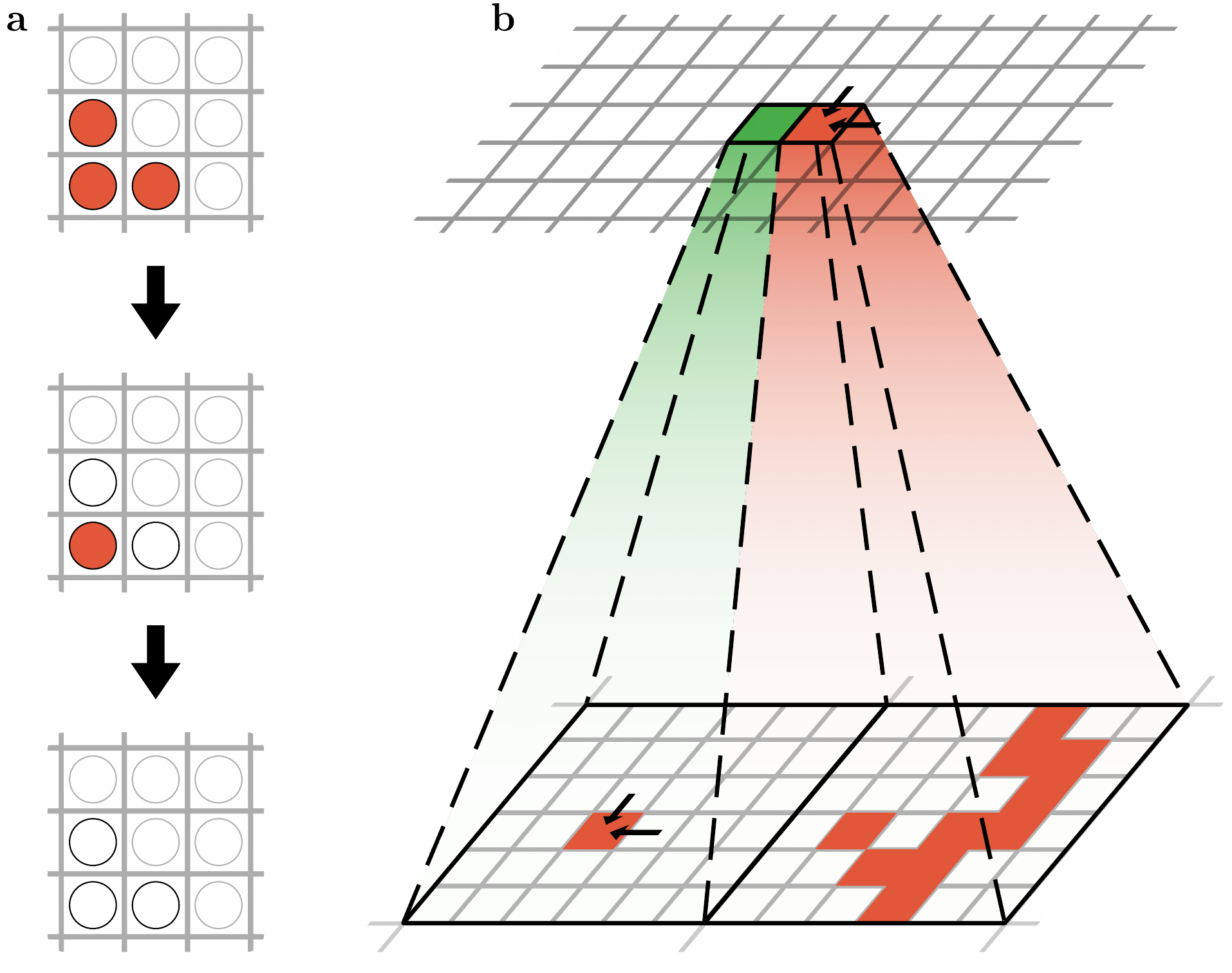}
    \caption{\textbf{Analytical noise reduction.}
    \textbf{a.} Ideal evolution under classical Toom's rule deterministically heals damage contained within any finite triangular region.
    \textbf{b.} Each level of the construction corrects local errors via Toom's rule (structure layer) and a concatenated quantum error-correcting code (data layer) while simultaneously performing a simulation of the quantum cellular automaton acting on the level above.
    The states of simulated higher-level cells are spread out over blocks of qudits on the level below and appear faulty whenever the errors within one such block can no longer be corrected by the code.
    We prove that the noise model is effectively renormalized at each level of self-simulation so that the logical noise of the top level is suppressed with increasing system size.
    \vspace{-3 ex}
    } 
    \label{fig:layer-sketch}
\end{figure}

The structure and program registers form a control layer that coordinates the correction schedule. But these control registers, too, are noisy; so the controller must also be error-corrected.
Importantly, within any given level of simulation, the control layer remains essentially classical, meaning that the registers can be dephased into the computational basis without disturbing the encoded quantum data. We may therefore protect them with a \textit{Toom-type rule}~\cite{toom1980} that is geometrically local in two dimensions. 
At every step, each site updates its structure registers by a majority vote among its own values and those of its local north and east neighbors, with modular shifts that compensate for the ideal offsets in space and time.
Finite islands of errors are hence eroded from their north-east boundary inwards (Fig.~\ref{fig:layer-sketch}.a).
In parallel, an analogous block-wise majority update restores the product-state program symbols by copying each within-block symbol to temporary registers, shifting across neighboring blocks and refreshing from the corresponding north and east copies. 
Thereby the space-time pattern that determines the simulator's actions is itself maintained by local, translation-invariant dynamics. This method admits a geometric analysis~\cite{gacs2021, berman1988} by containing the erroneous sites inside north-east-triangular `damaged regions' that shrink deterministically in every fault-free step.

\vspace{4 ex}
\notoc\subsection{Closing the self-simulation loop}

Thus far we have constructed a fault-tolerant universal simulator for an arbitrary QCA. 
The final step is to set the simulated interactions equal to the constructed simulator itself (Fig.~\ref{fig:construction-sketch}). Then each block encodes a higher-level cell governed by the very same local dynamics, so the hierarchy required for concatenated error correction is regenerated autonomously. 
A priori, it is non-trivial that this is well-defined, since the simulated rule determines the required block size and cycle time, and these in turn inflate the local dimension and rule complexity of the simulator which must store space-time coordinates, program symbols and workspace symbols at every site, thus introducing some self-consistency constraints.
The point of our construction is that these dependencies are only polylogarithmic. 
Writing $M$ for the linear block size and $T$ for the time-duration of one macro-step, self-simulation reduces to consistency conditions of the form
$$ M \geq \Theta(\polylog(M,T)), \qquad T \geq \Theta(\poly (M) ). $$
In the Supplementary Information we show that these inequalities close at finite constant values of $M$ and $T$. The corresponding QCA is a fixed point of the coarse-graining map and self-corrects upon initialization in an appropriately encoded state. %

\vspace{4 ex}
\notoc\section{Threshold proof}

The threshold proof combines the ideas of concatenated fault-tolerance developed for actively controlled (space-time-varying) quantum circuits~\cite{aliferis2006} with an analysis of the Toom-type rule protecting the control layer.
One macro-step on an $M\times M \times T$ space-time block is viewed as a single higher-level location. If faults inside that region are few, the block implements the intended operation; otherwise it is declared `bad' and counted as an effective simulated fault. Defects in the control registers are confined to sparse damaged regions whose triangular envelope shrinks under Toom-correction whenever no new faults occur~\cite{gacs1989}. As a result, the induced effective noise weakens under coarse-graining from one level to the next (Fig.~\ref{fig:layer-sketch}.b). Iterating this renormalization yields a nonzero threshold below which the logical noise strength decreases with the system's scale and the memory lifetime diverges in the thermodynamic limit.

\vspace{4 ex}
\notoc\section{Continuous-Time Implementation}

The self-correcting dynamics can also be implemented in continuous time. In this version of our scheme, the updates are realized by a \textit{translation-invariant, time-independent local Lindbladian} with engineered dissipative jump operators. (Alternatively, the dynamics can be approximated using a time-independent local Hamiltonian acting on the system and local  ancillas, together with dissipative resets of the ancillas~\cite{kashyap2025}.)

To remove the global clock, we replace the synchronized updates by local dissipative jumps that advance the simulation subject to an asynchronous `marching-soldier' constraint~\cite{wang1991, gacs1986, nehaniv2004} which prohibits neighboring cells' update schedules from drifting by more than one sub-step. At each level, this restriction is enforced only within finite synchronization neighborhoods %
that simulate separate higher-level operations. 
The correction jumps preserve causality by storing a local buffer of the classical control registers, so each structural update reads concurrent neighbor states. The quantum data operations are written as a brickwork-geometry pattern of disjoint local channels, invariant under reordering up to $\pm 1$ local asynchrony without requiring extra data copies. 

This asynchronous protocol can be controlled by local, finite-valued clocks using only constant additional onsite dimension. %
Faults in the clock registers are locally detectable and can be repaired so that, outside damaged regions, each synchronization neighborhood executes the correct higher-level update (with at most constant slowdown compared to the synchronized scheme). 
Under the local unitary jump-noise model specified in the Supplementary Information, we argue that trajectories of the corresponding Lindbladian reduce to the asynchronous discrete-time automaton with a renormalized local fault strength. Applying the threshold arguments to this setting then indicates that the logical errors are suppressed with increasing system size. We leave a formal proof of the resulting extension of the threshold theorem to continuous time for future work.

\vspace{4 ex}
\notoc\section{Discussion}

We have shown that autonomous fault-tolerant quantum computation is possible in two dimensions with fixed, local and translation-invariant dynamics. 
In that sense,
our construction gives the autonomous counterpart of the concatenated-code threshold theorems for actively controlled architectures:
Scalable fault-tolerance does not fundamentally require measurements, classical feedback, spatial inhomogeneity or externally imposed timing.
More broadly, our findings indicate new physical phenomena that can appear in non-equilibrium many-body systems. We have shown that structured quantum information processing can itself become a stable dynamical property of a two-dimensional system. 
The special case where the encoded computation is a plain counter yields long-lived temporal order that may be termed a stable dissipative quantum time crystal~\cite{yao2020}.

The
hierarchical self-simulation framework is more modular than the described implementation might suggest. The proof uses only that the data layer admits a nearest-neighbor, measurement-free concatenated gadget set with constant ancilla overhead and sufficiently large fixed distance. The same organization can therefore protect a broad class of concatenated encoding schemes. Because the recursive Toom-correction reproduces the required hierarchy along each row and column of the two-dimensional lattice, this can be used to protect either a single genuinely two-dimensional code or many one-dimensional codes in parallel; and universal nearest-neighbor gates can be applied between the separate logical systems of adjacent rows in the latter case.
Beyond concatenated codes, the same architecture could also be combined with other quantum code families. In particular, one might use the self-correcting structure layer to eliminate explicit spatial dependence and global synchronization requirements in local-automaton decoders for topological codes~\cite{balasubramanian2026a}, thereby enabling fully translation-invariant correction of the toric code in two dimensions.

We finish by discussing the limitations of our work and pointing to related open directions.

First, it would be very interesting to see whether one can obtain a similar scheme even in one dimension. Upon embedding a quantum error-correcting code within the structure of a state stored in \gacs's one-dimensional automaton \cite{gacs1986}, the proof should, in principle, resemble the two-dimensional case, but the construction would be considerably more complicated since the hierarchical structure cannot be error-corrected so easily by comparison with the local neighborhood. 
Success would yield a counterexample to dissipative analogues of the rapid thermalization results in Ref.~\cite{bergamaschi2025a}.
On the classical side, one may use a classical repetition code in the data layer of our scheme to obtain a simplified version of \gacs's two-dimensional automaton. Carrying our ideas over to one dimension might similarly result in a streamlined proof of the classical one-dimensional case.

Second, in the presence of noise, residual logical faults at any finite system size eventually (after an exponentially long time) mix between the encoded sectors, so the steady state is completely mixed in the logical degrees of freedom. 
Intuitively, one might consider a history-state embedding~\cite{verstraete2009} that applies the correction jumps both forwards and backwards in time at the highest logical level (with weak bias towards the initial state) and which thus produces a unique attractive fixed point whose local statistics can be post-selected to extract the outcome of an encoded computation. 
If that picture can be made rigorous, it would yield a system for which logical expectation values on the fixed point are stable (in the sense of Ref.~\cite{trivedi2024}) and would thereby imply BQP-completeness for the class of stable, size-independent Lindbladian fixed-point problems.

In the memory setting, the perturbed continuous-time protocol produces a local, translation-invariant Lindbladian whose convergence to equilibrium must, granting the outlined threshold proof, be slow, since the encoded memory lifetime grows with the system size.
Our analysis indicates robustness under sufficiently weak local unitary noise jumps. We conjecture that this extends to arbitrary low-weight local perturbations, so that the system would provide a concrete example pertaining to the question of whether slow-mixing Lindbladians can remain slow-mixing under perturbations that would, in isolation, generate rapidly mixing dynamics.

Moreover, our findings also hold implications for the computational complexity of tensor network states. 
Computing local expectation values of two-dimensional injective projected entangled-pair states (PEPS) is known to be hard (post-BQP-complete) for sufficiently small injectivity parameter~\cite{schuch2007, malz2025, harley2025}. Because our construction embeds arbitrary fault-tolerant computation into a translation-invariant local rule, it enables an extension of these hardness results to translation-invariant PEPS in three dimensions, when a spatially varying boundary is used to project onto the required initial state. A next step would be to remove even that boundary inhomogeneity, for example by using a long history-state sequence (in the spirit of Ref.~\cite{gottesman2009}) that counts up to the desired non-uniform input state before post-selection.

Two further open challenges are self-organization and encoding rate. First, can the same fixed interactions also fault-tolerantly prepare the hierarchical initial state? If so, one would need to supply only the circuit description or stored state in translation-invariant product form, for example repeated locally across the lattice, and the dynamics would autonomously grow the full hierarchical pattern. For this, the main obstacle is ensuring that this growth proceeds on a slower time-scale than the error correction inside already undamaged regions, so that the act of building the hierarchy does not itself introduce and spread too many errors.
Second, an extension to constant encoding rate may be available by replacing the fixed base code with a tower of Hamming-type codes whose block size increases with level, following the ideas of Refs.~\cite{yamasaki2024, gidney2025}. The recursive structure should provide enough local information about a cell's position within the hierarchy to select the appropriate level-dependent correction operations.

Finally, natural next steps are to generalize the noise model and to optimize the efficiency of the scheme.
For the continuous-time construction, the Supplementary Information provides an outline of a threshold proof by adapting the formalism used in the synchronous case; a more detailed proof is left for future work. 
Extending the analysis to more general noise, in particular to non-unitary perturbations on the Lindbladian (e.g. via a fault-path expansion directly in continuous time) or non-local thermal noise, would substantially strengthen the result.
Likewise, obtaining explicit threshold estimates, reducing the local dimension and simplifying the control mechanism are desirable steps towards practical realizability.

\vspace{2 ex}
\textit{Note added:} \quad After the first version of this paper appeared on the arXiv, Balasubramanian, Davydova and Lin posted a preprint that independently constructs a passive quantum memory in three dimensions~\cite{balasubramanian2026}. Their work studies an equilibrium Hamiltonian setting, giving a geometrically local, spatially non-uniform stabilizer Hamiltonian with exponential memory time under weak-coupling Markovian thermalization at sufficiently low temperature. This is complementary to the present construction which uses spatially uniform non-equilibrium dynamics in two dimensions, and which extends further to fault-tolerant computation. An interesting open question is whether the Hamiltonian approach can be made robust beyond the assumption of Markovian thermal noise, or extended to autonomous computation.

\vspace{4 ex}
\notoc\section{Acknowledgments}

We are grateful to S.~Balasubramanian, J.I.~Cirac, M.~Davydova, J.~Eisert, P.~\gacs, J.~Haah, E.~Lake, C.~Pattison and J.~Preskill for discussions and feedback. 
Financial support:
This research is part of the Munich Quantum Valley, which is supported by the Bavarian State Government with funds from the Hightech Agenda Bayern Plus. 
R.T. acknowledges funding from the European Union's Horizon Europe research and innovation program under grant agreement no. 101221560 (ToNQS).
G.S. acknowledges funding from the Deutsche Forschungsgemeinschaft (DFG, German Research Foundation) under Germany’s Excellence Strategy EXC2111–390814868.

\vspace{4 ex}
\notoc\section{Author contributions}
G.D. conceived the self-correction scheme. G.D., G.S. and R.T. refined the construction. G.D. and R.T. developed the formal analysis. R.T. supervised the project. G.D. wrote the manuscript, with substantial input from R.T. during revisions.

\vspace{4 ex}
\notoc\section{Competing interests}
The authors declare no competing interests.

\FloatBarrier %

\let\oldaddcontentsline\addcontentsline
\renewcommand{\addcontentsline}[3]{}
\bibliography{bibliography}
\let\addcontentsline\oldaddcontentsline

\let\oldaddcontentsline\addcontentsline
\renewcommand{\addcontentsline}[3]{}
\let\addcontentsline\oldaddcontentsline

\clearpage

\input{SM.tex}

\end{document}

%% file: preamble.tex
\usepackage[english]{babel}
\usepackage[utf8]{inputenc}
\usepackage[T1]{fontenc}
\usepackage{csquotes} %

\usepackage{ragged2e}   %
\AtBeginDocument{} %
\usepackage{enumitem} %
\usepackage{needspace}

\usepackage{graphicx}
\usepackage{placeins}   %
\usepackage{tabularx}   %
\usepackage{multirow}   %
\usepackage{adjustbox}  %
\usepackage[normalem]{ulem} %

\usepackage{mathtools} %
\usepackage{physics}   %
\usepackage{bm}        %
\usepackage{bbm}       %
\usepackage{dsfont}    %

\usepackage{amsthm}
\newtheoremstyle{indented-definition}{.5\topsep}{.5\topsep}{\addtolength{\leftskip}{2.5em}}{-2.5em}{}{}{ }{}
\theoremstyle{plain} %
\newtheorem{definition}{Definition}%
\theoremstyle{plain}
\newtheorem{theorem}{Theorem}%
\newtheorem{lemma}{Lemma}%
\newtheorem{fact}{Fact}%

\usepackage{xspace} %
\usepackage[dvipsnames]{xcolor} %

\newcommand{\e}{\mathrm{e}}

\newcommand{\nocontentsline}[3]{}
\newcommand{\notoc}[2]{\bgroup\let\addcontentsline=\nocontentsline#1{#2}\egroup}

\newcommand{\poly}{\operatorname{poly}}
\newcommand{\bigO}{\mathcal{O}}
\newcommand{\polylog}{\operatorname{polylog}}

\newcommand{\gacs}{Gács\xspace}

\DeclareMathOperator{\Maj}{Maj}

\DeclareMathOperator{\Toom}{\mathsf{Toom}}
\DeclareMathOperator{\R}{\mathsf{R}}
\DeclareMathOperator{\SimR}{ \mathsf{R_2} } %
\DeclareMathOperator{\AsyncR}{\mathsf{R_{async}}}
\DeclareMathOperator{\Univ}{\mathsf{Univ}}
\DeclareMathOperator{\Concat}{\mathsf{Concat}}

\DeclareMathOperator{\Dam}{Dam} %

\newcommand{\vol}{\operatorname{vol}} %

\newcommand{\tFaults}{\lfloor \tEC / s_0 \rfloor} %

\newcommand{\Id}{\mathds{1}}
\newcommand{\Penc}{\Phi_{\mathrm{enc}}} %
\DeclareMathOperator{\Enc}{ {\Phi}_{\R_2 \to \Univ} } 
\DeclareMathOperator{\Dec}{ {\Phi}_{\Univ \to \R_2} }

\newcommand{\s}{\mathbf{s}\xspace} %
\newcommand{\x}{x\xspace} %
\newcommand{\y}{y\xspace} %

\newcommand{\tEC}{ t_{\mathrm{EC}} }
\newcommand{\tbad}{ t_{\mathrm{bad}} }

\input{exrec_diagrams} %

\usepackage{hyperref}
\hypersetup{
    colorlinks   = true,            %
    linkcolor    = {blue!60!black}, %
    citecolor    = {green!60!black},%
    urlcolor     = {blue!60!black}  %
}

%% file: exrec_diagrams.tex
\usepackage{tikz}

\newcommand{\diagram}[1]{ %
    \begin{tikzpicture}[thick, baseline=-0.5ex] %
        \draw(0,0) -- (0.5,0);
        #1
    \end{tikzpicture}
}
\newcommand{\Diagram}[1]{ %
    \begin{tikzpicture}[thick, baseline=-0.5ex] %
        \draw[double] (0,0) -- (0.5,0);
        #1
    \end{tikzpicture}
}
\newcommand{\TQDiagram}[1]{ %
    \begin{tikzpicture}[thick, baseline=-0.5ex] %
        \draw[double] (0,0.625) -- (0.5,0.625);
        #1
    \end{tikzpicture}
}

\newcommand{\Wire}{%
    \pgfgetlastxy{\lastx}{\lasty}%
    \begin{scope}[shift={(\lastx,\lasty)}]
        \draw[double] (0,0) -- (0.5,0);
    \end{scope}%
}

\newcommand{\filter}[1]{%
    \pgfgetlastxy{\lastx}{\lasty}%
    \begin{scope}[shift={(\lastx,\lasty)}]
        \draw (0,-0.5) rectangle (0.25,0.5);
        \node[above right, inner sep=1pt] at (0.25,0.5) {$\scriptstyle #1$};
        \draw[double] (0.25,0) -- (0.75,0);
    \end{scope}%
}

\newcommand{\EC}[1]{%
    \pgfgetlastxy{\lastx}{\lasty}%
    \begin{scope}[shift={(\lastx,\lasty)}]%
        \draw (0,-0.5) rectangle (1,0.5) node[midway] {EC};
        \node[above right, inner sep=1pt] at (1,0.5) {$\scriptstyle #1$};
        \draw[double] (1,0) -- (1.5,0);
    \end{scope}%
}

\newcommand{\Rectangle}[1]{%
    \pgfgetlastxy{\lastx}{\lasty}%
    \begin{scope}[shift={(\lastx,\lasty)}]%
        \draw (0,-0.5) rectangle (1,0.5) node[midway] { };
        \node[above right, inner sep=1pt] at (1,0.5) {$\scriptstyle #1$};
        \draw[double] (1,0) -- (1.5,0);
    \end{scope}%
}

\newcommand{\rectangle}[1]{%
    \pgfgetlastxy{\lastx}{\lasty}%
    \begin{scope}[shift={(\lastx,\lasty)}]%
        \draw (0,-0.5) rectangle (1,0.5) node[midway] { };
        \node[above right, inner sep=1pt] at (1,0.5) {$\scriptstyle #1$};
        \draw (1,0) -- (1.5,0);
    \end{scope}%
}

\newcommand{\gate}[1]{%
    \pgfgetlastxy{\lastx}{\lasty}%
    \begin{scope}[shift={(\lastx,\lasty)}]%
        \draw[rounded corners] (0,-0.5) rectangle (1,0.5) node[midway] {$\mathcal{G}$};
        \node[above right, inner sep=1pt] at (1,0.5) {$\scriptstyle #1$};
        \draw (1,0) -- (1.5,0);
    \end{scope}%
}
\newcommand{\Gate}[1]{%
    \pgfgetlastxy{\lastx}{\lasty}%
    \begin{scope}[shift={(\lastx,\lasty)}]%
        \draw[rounded corners] (0,-0.5) rectangle (1,0.5) node[midway] {$\mathcal{G}$};
        \node[above right, inner sep=1pt] at (1,0.5) {$\scriptstyle #1$};
        \draw[double] (1,0) -- (1.5,0);
    \end{scope}%
}

\newcommand{\TQgate}[1]{%
    \pgfgetlastxy{\lastx}{\lasty}%
    \begin{scope}[shift={(\lastx,\lasty)}]%
        \draw[rounded corners] (0,-1.75) rectangle (1,0.5) node[midway] {$\mathcal{G}$};
        \node[above right, inner sep=1pt] at (1,0.5) {$\scriptstyle #1$};
        \draw (1,0) -- (1.5,0);
    \end{scope}%
}
\newcommand{\TQGate}[1]{%
    \pgfgetlastxy{\lastx}{\lasty}%
    \begin{scope}[shift={(\lastx,\lasty)}]%
        \draw[rounded corners] (0,-1.75) rectangle (1,0.5) node[midway] {$\mathcal{G}$};
        \node[above right, inner sep=1pt] at (1,0.5) {$\scriptstyle #1$};
        \draw[double] (1,0) -- (1.5,0);
    \end{scope}%
}

\newcommand{\decoder}{%
    \pgfgetlastxy{\lastx}{\lasty}%
    \begin{scope}[shift={(\lastx,\lasty)}]%
        \draw (0,-0.5) -- (0, 0.5) -- (0.5, 0) -- cycle;
        \draw (0.5,0) -- (1,0);
    \end{scope}%
}
\newcommand{\unitarydecoder}{%
    \pgfgetlastxy{\lastx}{\lasty}%
    \begin{scope}[shift={(\lastx,\lasty)}]%
        \draw (0,-0.5) -- (0, 0.5) -- (0.5, 0) -- cycle;
        \draw (0.25,-0.25) -- (0.25,-0.5) -- (1,-0.5);
        \draw (0.5,0) -- (1,0);
    \end{scope}%
}
\newcommand{\unitarydecoderinverse}{%
    \pgfgetlastxy{\lastx}{\lasty}%
    \begin{scope}[shift={(\lastx,\lasty)}]%
        \draw (0.5,-0.5) -- (0.5, 0.5) -- (0, 0) -- cycle;
        \draw (0.25,-0.25) -- (0.25,-0.5) -- (0,-0.5);
        \draw[double] (0.5,0) -- (1,0);
    \end{scope}%
}

\newcommand{\preparation}[1]{%
    \pgfgetlastxy{\lastx}{\lasty}%
    \begin{scope}[shift={(\lastx,\lasty)}]%
        \draw (-0.5,-0.5) arc (270:90:0.5) -- cycle;
        \node[above right, inner sep=1pt] at (-0.5, 0.5) {$\scriptstyle #1$};
        \pgfpathmoveto{\pgfpoint{0.5}{0}}
    \end{scope}%
}

\newcommand{\measurement}[1]{%
    \pgfgetlastxy{\lastx}{\lasty}%
    \begin{scope}[shift={(\lastx,\lasty)}]%
        \draw (0,0.5) arc (90:-90:0.5) -- cycle;
        \node[above right, inner sep=1pt] at (0.25,0.5) {$\scriptstyle #1$};
    \end{scope}%
}

%% file: SM.tex
\onecolumngrid
\setcounter{secnumdepth}{3} %
\begin{center}
{\large \textbf{Supplementary Information for} \par \textbf{"Quantum Memory and Autonomous Computation in Two Dimensions"} }
\end{center}
\vspace{2 ex}
\twocolumngrid
\tableofcontents

\vspace{4 ex}

This Supplementary Information functions as a self-contained discussion of our self-correction scheme and provides a rigorous proof of the threshold theorem.

Sec.~\ref{sec:overview} gives an informal overview.
Sec.~\ref{sec:setting} defines quantum cellular automata (QCAs) and the noise model.
Sec.~\ref{sec:FT-encodings} reviews the necessary background on concatenated fault-tolerance.
Sec.~\ref{sec:2D-construction} constructs the self-correcting QCA.
Sec.~\ref{sec:threshold-proof} proves the threshold theorem for discrete-time evolution under this QCA. 
Sec.~\ref{sec:computation-setup} explains initialization and read-out in that setting.
Finally, Sec.~\ref{sec:async-scheme} translates the same scheme into continuous time.

\section{Overview} \label{sec:overview}

We begin with an overview of the main intuitions and relevant ingredients.
Firstly, we note that there exists a \textit{classical} transition function which performs self-correction when applied globally at discrete time-steps on a two-dimensional lattice.
This is given by the (remarkably simple) rule~\cite{toom1980}:
\begin{definition}[Toom's rule]
    A classical trajectory $\s(t,i,j)$ on the periodic square lattice $\Lambda = \mathbb{Z}_L^2$ follows \emph{Toom's rule} if and only if
    \begin{align*}
        \s(t+1,i,j) =&{} \Toom(\s(t,\cdot),i,j)  \\
        :=&{} \Maj(\s(t,i,j) , \s(t,i+1,j) , \s(t,i,j+1))
    \end{align*}
    at all time-steps $t \in \mathbb{N}_0$ and all sites $(i,j)\in \Lambda$. The function $\Maj(a,b,c)$ represents the majority vote: if at least two of its arguments coincide, then it returns their shared value, otherwise it returns the first argument. At each step, Toom's rule thus computes the majority among the north and east neighbors and the central site.

    When we speak of applying Toom's rule to a quantum register, we mean the channel obtained by first dephasing into the computational basis and then applying a classical Toom-update to the resulting diagonal state.
\end{definition}
Toom's rule stores a single bit encoded in the classical repetition code for an extended amount of time.
The intuition is that, in two or higher dimensions, taking the majority vote of a local neighborhood erodes islands of errors from their corners inwards.
Moreover, even in one dimension it is possible to reliably store information within \gacs's classical cellular automaton~\cite{gacs1986, gray2001}. This requires a complex update rule which ensures that the entire state is hierarchically structured and that regions of errors within it shrink whenever their internal structure has a smaller scale than that of the surrounding state.
For the sake of simplicity, we proceed with the two-dimensional case.

As a second ingredient to our scheme, we recall existing \textit{active} quantum error correction schemes with concatenated codes.
These have been used to prove the threshold theorem in Refs.~\cite{aharonov1997, gottesman2000}.
For our purposes, the essential fact (formalized as Lem.~\ref{lem:concat-code} in Sec.~\ref{sec:FT-encodings} below) is that, for any fixed number of errors $\tEC$, there exists a constant-size quantum code together with an encoding scheme for a complete gadget set, which is nearest-neighbor local (in one or two dimensions), measurement-free, and $\tEC$-fault-tolerant (meaning that it corrects up to $\tEC \in \mathbb{N}$ errors per encoding block and that the encoding scheme can be iterated recursively, i.e. concatenated; see Sec.~\ref{sec:FT-encodings}).
Importantly, although such schemes require no intermediate classical processing, it is still necessary to store and control the space-time-varying structure of the protocol.
To obtain a truly \textit{passive/self-correcting} quantum error correction method, we supply the concatenated code with a clock mechanism that achieves this control in a time-independent manner whilst preserving the stability to errors.

The main contribution of our construction hence lies in finding a way to self-correct the required space-time information that is hard-coded or externally-enforced in active correction schemes.
We achieve this by realizing recursively layered Toom-like correction in corresponding structural layers across many scales.
Essentially, what we build is a self-correcting automaton that is capable of universal computation, drawing inspiration from the classical automaton of Ref.~\cite{gacs1989} but combining the universal control at each level with a concatenated quantum error-correcting code for an overall self-correcting quantum system.

To make this concrete, we utilize a universal quantum cellular automaton to replace the temporal control in quantum computations with a fixed, local transition rule. It is known that there exist QCAs which achieve universality by shifting the computational complexity entirely to the initial state~\cite{schumacher2004, vollbrecht2006, shepherd2006, arrighi2008, vollbrecht2008, kohler2022}.
We let $\Univ$ denote the one-step update channel of such a universal automaton (made more explicit in Sec.~\ref{sec:2D-construction} below).
Universality means that for any finite quantum circuit (or local transition rule) there exists a product-state `program' such that repeated application of $\Univ$ on a state initialized to hold this program implements the desired dynamics on a data-representing subsystem of the state (up to a fixed encoding).
For example, there exists an initial state that represents the program with which $\Univ$ implements Toom's rule on a designated register.

However, Toom's rule alone has limited error-correcting properties, allowing for the storage of merely a single bit. The recursive application on several levels improves on this but brings the additional challenge that this hierarchical updating must itself emerge from a spatially uniform local rule. To resolve this, note that a system which is capable of universal computation is by definition capable of simulating any prescribed quantum dynamics. We will see that the program which controls such simulation can be engineered such that the system effectively simulates the execution of its own update rule on states encoded in the simulated level above. 
This \textit{self-simulation} is the crucial idea underlying \gacs's classical results~\cite{gacs1986, gacs1989} and our novel quantum construction.
It takes some care to achieve self-simulation without disturbing lower levels and while using only a constant local state-space, but, once the self-simulation is robustly implemented, the error-correcting properties recursively amplify themselves up to the highest level that fits in the system.
To impose robustness we use that the hierarchical pattern of the states stored by this self-simulation is exactly the same as the structure needed in concatenated quantum error correction schemes.

Altogether, we construct a fixed set of interactions that simultaneously performs two tasks: Firstly, these interactions correct a constant number of errors locally within a quantum error-correcting code chosen according to Lem.~\ref{lem:concat-code} and within a classical error-correcting code that is a variant of Toom's rule. These codes are located in separate subspaces and we consider the full system to obey a combined code. And secondly, the interactions enforce a simulation of an interacting quantum system whose qudits are represented by blocks of physical qudits. The interactions within this simulated system are applied on a slower scale than the physical interactions but, crucially, they are exactly the same as those of the system itself, i.e. the system exhibits self-simulation.
Together, these properties ensure that appropriate states can be stored in the system for an extensive amount of time. One may include further program instructions to perform gates of a chosen circuit and prepare an initial state which carries those instructions at the top encoding level of the self-simulating hierarchy.
This way, the largest blocks in the system undergo an evolution that fault-tolerantly simulates the desired circuit in the highest level of encoding.

Throughout this work, our notation as well as our proofs are heavily inspired by \gacs's classical construction in Ref.~\cite{gacs1989}. Readers familiar with that reference should note that we differ in the representation of control information. Aiming for ease of understanding, instead of distributing a site's structural state over a `colony of sites', we will store the full local structural information at each cell. The increased overhead from self-simulating the resulting large but constant onsite state-space will contribute to the spatial extent of the simulating blocks.

\section{Setting} \label{sec:setting}

Let us start by describing the framework within which our construction and proofs take place.

\subsection{Quantum cellular automata} \label{sec:QCA-def}

Most of this work is formulated in the setting of quantum cellular automata (QCAs), by which we mean quantum systems evolving in discrete time-steps at which fixed local quantum channels are applied in alternating brickwork layers to a line or grid of qudits. The term `cells' refers to the local, finite-dimensional subsystems placed at individual sites on the lattice.
Operations are applied throughout the system at integer times and we consider local noise which can produce faulty operations by acting either during or in between these time-steps.

\begin{definition}[Brickwork QCA] \label{def:brickwork-QCA}
A 2D quantum cellular automaton with update channel $\R$ (a channel acting on four qudits) is 
the following geometrically local brickwork circuit:
\begin{itemize}[label=-]
    \item On odd layers, apply $\R$ on the $2\times 2$ blocks of sites $(2n + 1, 2m + 1), (2n + 1, 2m + 2), (2n + 2, 2m + 1), (2n + 2, 2m + 2)$ for $n, m \in \{0, 1, 2, \dots \}$;
    \item On even layers, apply $\R$ on the $2\times 2$ blocks of sites $(2n, 2m), (2n + 1, 2m), (2n, 2m + 1), (2n + 1, 2m + 1)$ for $n, m \in \{0, 1, 2, \dots \}$.
\end{itemize}
We refer to the automaton and its update channel (also called its update or transition rule) by the same symbol~$\R$.
\end{definition}
\emph{Remark.} Although QCAs may be defined directly on infinite lattices, we will not deal with such mathematical subtlety here and instead always consider finite grids. 
Unless open boundaries are stated explicitly, we instantiate the QCA on an $L \times L$ periodic grid, take $L$ even and consider site indices modulo $L$; this gives periodic boundary conditions and two complete brickwork tilings. A `step' of the automaton refers to a single update layer, which applies either an odd or an even layer of brickwork channels.
When an automaton is instantiated on an $L \times L$ grid we denote the channel of a single global update step by $\R_L$.

To precisely formulate simulation overheads later on, we introduce a set of universal channels: %
\begin{align}\label{eq:universal-gateset}
\mathcal{G} = \{\textnormal{RESET}, \text{SWAP}, \mathcal{U}_1, \mathcal{U}_2, \dots, \mathcal{U}_G \}
\end{align}
where $\textnormal{RESET}$ is the single-qudit reset channel 
$
    \textnormal{RESET}(\rho) = 
    \text{Tr}(\rho) \bigl\lvert \vec{0} \bigr\rangle
\!\bigl\langle \vec{0} \bigr\rvert %
$, 
$\textnormal{SWAP}$ is the two-qudit swap gate,
and $\mathcal{U}_1, \mathcal{U}_2, \dots, \mathcal{U}_G$ form a universal set of unitary gates. 
We fix $\mathcal{U}_1, \dots, \mathcal{U}_G$ to be the universal gate set $\mathcal{G}_1$ from the extended version of Ref.~\cite{aharonov1997}, which also underlies the fault-tolerant gadgets in Lem.~\ref{lem:concat-code} below.
(More explicitly, embedding the qudits into prime dimension $p$, using arithmetic modulo $p$ and taking $u \in \mathbb{F}_p$ and $\omega = \e^{2 \pi i /p}$, the set $\mathcal{G}_1$ consists of the increment gates $\ket{a} \mapsto \ket{a + u}$, the sum gates $\ket{a,b} \mapsto \ket{a, a \pm b}$, the multiplication gates $\ket{a} \mapsto \ket{u a}$ ($u \neq 0$), the phase gates $\ket{a} \mapsto \omega^{ua} \ket{a}$, the controlled-phase gates $\ket{a, b} \mapsto \omega^{u a b} \ket{a,b}$, the Fourier transforms $\ket{a} \mapsto p^{-1/2} \sum_{b \in \mathbb{F}_p} \omega^{u a b} \ket{b}$ ($u \neq 0$), and the generalized Toffoli gate $\ket{a,b,c} \mapsto \ket{a,b, c + ab}$. The remaining operations of $\mathcal{G}_1$ listed in Ref.~\cite{aharonov1997}, namely swapping, adding and discarding qudits, are covered by SWAP and RESET operations using some ancilla qudits.)
When we speak of $\mathcal{G}$ as a universal set in two dimensions, this refers to brickwork circuits which may at each step apply channels from $\mathcal{G}$ to arbitrary placements (including diagonals) within the local $2\times 2$ plaquettes.

From the Solovay-Kitaev theorem and Stinespring dilation, any channel on finitely many qudits can be approximated using operations from $\mathcal{G}$ (and possibly some ancillary qudits). Moreover, certain channels can be implemented exactly. 
If a channel $\mathcal{K}$ on a $k \times k$ grid of qudits has an exact implementation as a sequence of channels from $\mathcal{G}$, we denote by $\vol(\mathcal{K})$ the minimum number of elementary channel locations in such an implementation. This is the total volume %
of the compiled schedule. Without loss of generality, $\vol(\mathcal{K}) \geq k$ (otherwise pad the sequence with one layer of identity channels).
Note that because $\text{SWAP} \in \mathcal{G}$, any non-geometrically-local implementation can be made geometrically local.

\subsection{Noise model and fault-path expansion} \label{sec:noise-model}

We now describe the noise model used in our theorem. 
Because the QCA is updated synchronously at integer times, its ideal evolution can be viewed as a translation-invariant circuit where each QCA step applies a constant-depth pattern of local operations to the lattice. A \textit{location} is the local operation applied in the QCA step at a specified space-time position.
We count one location as being one $4$-site channel acting on a $2 \times 2$ brickwork plaquette, and one fault as one incorrectly applied such location.
(So a single fault can corrupt all four qudits in the support of one operation.)
We model noise as acting on each location by coupling the system to an arbitrary environment and replacing the ideal operation by a joint system-bath unitary.
This yields a \textit{fault-path expansion} where the full noisy evolution is characterized by the subset of space-time locations at which a fault occurs.

Formally, label the ideal channels applied at the $T$ locations of a circuit under consideration as $\mathcal{K}_1, \mathcal{K}_2, \dots, \mathcal{K}_T$. We assume that the system interacts with a (possibly global) environment $E$, and in the presence of this environment, the implemented noisy channel on the cells can be expressed as 
\begin{align*}
    \Phi^\mathrm{noisy}(\cdot) 
&= \sum_{F \subseteq \{1, 2, \dots, T\}} \Phi^\mathrm{noisy}_F(\cdot) 
\\ &= \sum_{F \subseteq \{1, 2, \dots, T\}} \text{Tr}_E[\Phi_F( (\cdot)\otimes \rho_E(0))],
\end{align*}
where the sum is over all the fault paths $F$ which are subsets of the locations $\{1, 2, \dots, T\}$.
$\Phi_F$ is a system-environment super-operator (not necessarily a channel) that can be viewed as a 2D brickwork of super-operators in which at all locations $i \not\in F$ it applies the noiseless channel $\mathcal{K}_i$, and at all locations $i \in F$ it applies an arbitrary super-operator between the four qudits involved at that location and the environment.

We assume that there exists an $\eta < 1$ such that the physical noise is \textit{local with strength} $\eta$:
\begin{definition}[Local noise] \label{def:local-noise}
    For a set $F$ of space-time locations, let
    $$
    E(F) := \sum_{F' :\, F \subseteq F'} \Phi^\mathrm{noisy}_{F'}
    $$
    denote the sum over all fault paths in which faults occur at all locations in $F$.
    The noise is local with strength $\eta$ if, for every $F$,
    \[
    \norm{E(F)}_\diamond \leq \eta^{|F|} .
    \]
\end{definition}
This model allows arbitrary correlations across space and time, mediated by the bath.
It is automatically satisfied when the system-bath interaction is weak and local, meaning that each fault couples only to the qudits acted on at that location~\cite{aliferis2006}.
The physical system evolves under synchronized applications of the QCA at discrete time-steps and experiences local noise, i.e. we allow correlated, possibly non-Markovian noise as in Refs.~\cite{aliferis2006, terhal2005, aharonov2006}.
For intuition one may picture the special case of independent stochastic faults with local error probability $p$, which satisfies Def.~\ref{def:local-noise} with strength $\eta = p$. %

\section{Background on fault-tolerant encodings} \label{sec:FT-encodings}

We now review the fault-tolerance machinery later used by the construction: the graphical gadget conditions for fault-tolerance and the extended-rectangle formalism (Sec.~\ref{sec:exRec-review}), concatenated quantum error-correcting codes (Sec.~\ref{sec:code-existence}), and the classical repetition code (Sec.~\ref{sec:rep-code}).
Nothing here is specific to self-correction, and readers familiar with concatenated fault-tolerance may skim this Section.

\subsection{Review of the extended-rectangle formalism} \label{sec:exRec-review}

We first introduce the gadget notation and the conditions for fault-tolerance, and then review how they yield a renormalization of the noise strength in concatenated encoding schemes, following Refs.~\cite{aliferis2006, gottesman2024}.

\subsubsection{Conditions for fault-tolerance} \label{sec:FT-conditions}

Let $\mathcal{C}$ be a quantum error-correcting code which can correct up to $\tEC$ errors per block.
A fault-tolerant encoding $\mathcal{E}$ is given by several \textit{gadgets} that systematically replace parts of the circuit with their error-corrected variants.
We divide a given quantum circuit $Q$ into separate \textit{locations} at which a local action (such as state preparation, application of a gate, measurement, or storage) is performed on the qudits.
The encoded circuit $\mathcal{E}(Q)$ is defined by replacing each qudit with an encoded qudit, replacing each location with a gadget that performs the respective action on the encoded qudits, and applying error-correction operations (EC gadget) after each step.

We use a graphical representation to specify sufficient conditions for fault-tolerance. Encoded blocks of the code are indicated by double-lines whereas single lines represent single unencoded qudits.
Let $\Diagram{\filter{r}}$ denote the projector onto the code's subspace
of states with at most $r$ errors on the block.
We write $\Diagram{\decoder}$ for the ideal decoder of $\mathcal{C}$. We represent the gadgets graphically as
\begin{align*}
    \Diagram{\preparation{r}} \quad &= \text{state preparation gadget with $r$ faults}, \\
    \Diagram{\Gate{r}} &= \text{gate gadget with $r$ faults}, \\
    \Diagram{\EC{r}} &= \text{error-correction gadget with $r$ faults}, \\
    \Diagram{\measurement{r}} \qquad &= \text{measurement gadget with $r$ faults.}
\end{align*}

For the purpose of Def.~\ref{def:FT-conditions}, a fault is 
an incorrectly applied location of the gadget schedule.
As noted in Sec.~\ref{sec:noise-model}, in our convention each location is a $4$-site channel on a $2 \times 2$ plaquette, so a single fault can leave errors on up to four physical qudits.
(For gadget schedules written in one- and two-qudit gates, as in Lem.~\ref{lem:concat-code} below cited from the literature, a fault is correspondingly a single faulty gate.)
We account for the resulting error counts via the \textit{spread} $s_0$ of an encoding scheme, which bounds the maximal number of errors individual faults can cause on any individual code block.

We use the above notation to define a set of conditions which a given fault-tolerant encoding scheme (i.e. a set of gadgets for encoded state preparation, gate application, measurement and error correction) may possess.

\begin{definition} [Conditions for fault-tolerance] \label{def:FT-conditions}
    A code $\mathcal{C}$ with associated encoding scheme $\mathcal{E}$ is said to be $\tEC$-fault-tolerant with spread $s_0$ if its gadgets satisfy the following conditions:

    \begin{enumerate}
        \item[(1)] The state-preparation gadget satisfies
        $$\Diagram{\preparation{r} \decoder } \quad = \quad \diagram{ \preparation{} } 
        $$
        and
        $$\Diagram{\preparation{r}} \quad = \quad \Diagram{ \preparation{r} \filter{s_0 r} } 
        $$
        for all $s_0 r \leq \tEC$.

        \item[(2)] Each single-qudit gate gadget satisfies 
        \begin{equation*}
            \Diagram{\filter{r} \Gate{s} \decoder } \quad = \quad \Diagram{ \filter{r} \decoder \gate{} } 
        \end{equation*}
        and
        \begin{equation*}
            \Diagram{ \filter{r} \Gate{s} } \quad = \quad \Diagram{ \filter{r} \Gate{s} \filter{r+s_0 s} } 
        \end{equation*}
        for all $r + s_0  s \leq \tEC$.
        
        Each two-qudit gate gadget satisfies 
        \begin{equation*}
            \TQDiagram{
                \filter{r_1} \TQGate{s} \decoder 
                \draw[double] (0,-0.625) -- (0.5,-0.625);
                \filter{r_2}
                \draw[double] (2.25,-0.625) -- (2.75,-0.625);
                \decoder
            } 
        \quad = \quad 
            \TQDiagram{ \filter{r_1} \decoder \TQgate{} 
                \draw[double] (0,-0.625) -- (0.5,-0.625);
                \filter{r_2}
                \decoder
                \draw (3.25,-0.625) -- (3.75,-0.625);
            } 
        \end{equation*}
        and
        \begin{equation*}
            \TQDiagram{\filter{r_1} \TQGate{s}
                \draw[double] (0,-0.625) -- (0.5,-0.625);
                \filter{r_2}
                \draw[double] (2.25,-0.625) -- (2.75,-0.625);
            } 
        \quad = \quad 
            \TQDiagram{ \filter{r_1} \TQGate{s} \filter{r_1 + r_2 + s_0 s} 
            \draw[double] (0,-0.625) -- (0.5,-0.625);
            \filter{r_2}
            \draw[double] (2.25,-0.625) -- (2.75,-0.625);
            \filter{r_1 + r_2 + s_0 s}
            } 
        \end{equation*}
        for all $r_1 + r_2 + s_0 s \leq \tEC$.
        
        \item[(3)] 
        The measurement gadget satisfies
        \begin{equation*}
            \Diagram{\filter{r} \measurement{s} } \quad = \quad 
            \Diagram{\filter{r} \decoder \measurement{} } 
        \end{equation*}
        for all $r + s_0 s \leq \tEC$.
        
        \item[(4)] 
        The error-correction gadget satisfies the correctness and propagation properties,
        \begin{equation*} %
            \Diagram{\filter{r} \EC{s} \decoder } \quad = \quad \Diagram{ \filter{r} \decoder } 
        \end{equation*}
        for all $ r + s_0 s \leq \tEC$,
        and
        \begin{equation*} %
            \Diagram{ \EC{s} } \quad = \quad \Diagram{ \EC{s} \filter{s_0 s} } 
        \end{equation*}
        for all $ s_0 s \leq \tEC$.
        Note that the latter property is required for any input state regardless of the number of errors.
    \end{enumerate}
\end{definition}

When a scheme is said to be $\tEC$-fault-tolerant without explicitly stating the spread, this refers to the case with $s_0 = 1$.
In general, an encoding scheme that is $\tEC$-fault-tolerant with spread $s_0$ has fault budget $\tFaults$, i.e. each gadget in the scheme is guaranteed to tolerate up to $\tFaults$ faults under the conditions above. Hence the spread merely rescales the required code distance by a constant factor.

\subsubsection{Noise renormalization in the extended rectangle formalism} \label{sec:exRec-noise-renormalization}

The error-correction capabilities of codes and encoding schemes satisfying the above conditions can be established by tracking the effective logical noise through successive concatenation levels using the method of extended rectangles. Because our later proofs will follow a similar line of reasoning, we here provide a short review of this formalism.

We analyze the effective noise strength of an encoded circuit under a $\tEC$-fault-tolerant concatenated encoding scheme. 
This relies on grouping the gadgets into \textit{extended rectangles}:

\begin{definition} [Extended rectangles] \label{def:exRec} %
    For every location in a logical circuit $Q$, the corresponding \emph{rectangle} in the encoded circuit $\mathcal{E}(Q)$ consists of that specific gadget followed by an error-correction gadget (EC gadget). An \emph{extended rectangle}~(exRec) is defined as a rectangle along with the EC gadgets immediately preceding it.
\end{definition}

\begin{figure}[tb]
    \centering
    \includegraphics[width=0.8\linewidth]{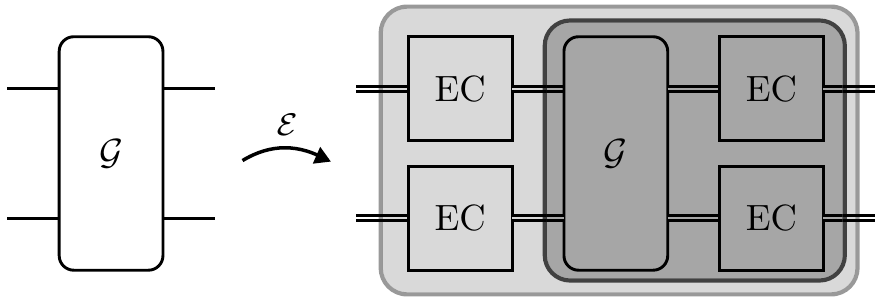}
    \caption{\textbf{Extended-rectangle formalism.}
    We group fault-tolerant gadgets into rectangles and extended rectangles to characterize noise at the encoded level. The Figure illustrates Def.~\ref{def:exRec} for the example of a two-qudit gate. The rectangle (dark gray) consists of the encoded gate gadget and trailing EC gadgets, and the extended rectangle (light gray) additionally includes the leading EC gadgets.
    } 
    \label{fig:exRec-defn-example}
\end{figure}

An example is shown in Fig.~\ref{fig:exRec-defn-example}.
The leading EC gadgets included in each exRec ensure that the input to encoded operations resides in a subspace with bounded errors while the trailing EC gadgets ensure that the output is error-corrected before passing to the next step.
We show that, going from one concatenation level to the next recursively, the effective faults experienced by the higher-level operations fall again under the same local noise model but with renormalized strength.
The higher-level effective faults are formalized by the notion of good and bad exRecs.
We define:
\begin{definition} [Goodness and Correctness] \label{def:goodness}
    For integer $s \geq 0$, an exRec is $s$-\emph{good} iff it contains at most $s$ faults among its constituents. Otherwise, it is $s$-\emph{bad}. Unless explicitly stated otherwise, good/bad means $\tbad$\nobreakdash-good/$\tbad$-bad.

    An exRec is \emph{correct} iff it performs the intended ideal operations on the encoded qudits, i.e. iff an ideal decoder can be commuted through the rectangle:
    \begin{equation*}
        \Diagram{ \Rectangle{} \decoder } \quad = \quad \Diagram{ \decoder \rectangle{} }\quad .
    \end{equation*}
\end{definition}

\emph{Remark.}
To anticipate these taking distinct values in our final proof (Sec.~\ref{sec:proof-notation}), we introduce here separate notation for the fault-tolerant code's correctable number of errors $\tEC$ and for the threshold number $\tbad$ of constituent faults above which an exRec is considered bad.

One obtains the following Lemma by checking, for each possible gadget, that the decoder can be commuted through the respective rectangles using the correctness and propagation properties (see Ref.~\cite{gottesman2024} for details).
\begin{lemma} \label{lem:goodness-correctness} %
    If the gadgets in an encoded circuit all adhere to the fault-tolerance conditions from Def.~\ref{def:FT-conditions} with spread $s_0$, then $\tFaults$\nobreakdash-goodness of exRecs implies correctness of the contained rectangle.
\end{lemma}

It follows that if all exRecs in an encoded circuit are $\tFaults$-good (or $\tbad$-good with $\tbad \leq \tFaults$), then that circuit produces the same output as the ideal logical circuit.
To analyze fault-path instances in which some exRecs are bad, let us imagine that the decoder stores the syndromes of all measured check operators in a separate wire. This defines a unitary decoder (*-decoder)
$$\Diagram{\unitarydecoder} \quad. \vspace{3 ex}$$
We extend the  definitions of the fault-tolerance conditions and of good and correct exRecs to include the *\nobreakdash-decoder.

Since this *-decoder is unitary, it has an inverse with which it can be inserted as identity at each step of the circuit:
$$\Diagram{\unitarydecoder \unitarydecoderinverse} \quad = \quad \Diagram{\Wire} \quad .$$
Thus we can formally move the *\nobreakdash-decoder from output to input through an entire circuit for an arbitrary fault-path instance that may contain both good and bad exRecs.
In each good exRec, the corresponding rectangle implements the ideal logical location (up to an operation acting only on the syndrome registers carried by the *-decoder), whereas each bad exRec is treated as a single effective fault at the encoded level.

This machinery yields the standard threshold theorem for concatenated encoding schemes:
\begin{lemma}[Threshold theorem; proof in \cite{aliferis2006}, proof sketch below] \label{lem:threshold-thm}
    Consider a quantum computer subject to local noise (Def.~\ref{def:local-noise}).
    Let $\mathcal{C}$ be a quantum error-correcting code with concatenated encoding scheme $\mathcal{E}$ that is $\tEC$\nobreakdash-fault-tolerant with spread $s_0$ and $\tEC \geq s_0$ (Def.~\ref{def:FT-conditions}).

    Then there exists a constant threshold $\eta_\mathrm{th} > 0$ such that any ideal circuit of depth $D$ on $N$ qubits can be simulated to within logical accuracy $\delta$ by the level-$k$ encoded circuit with $k = \bigO( \log \log(N D / \delta) )$, provided the local physical noise strength is $\eta < \eta_\mathrm{th}$.
\end{lemma}

To prove Lem.~\ref{lem:threshold-thm}, we sketch how exRecs rescale the effective local noise strength under each level of encoding. We follow Ref.~\cite{aliferis2006}.
For a given fault path, let $\mathcal{N}_0$ be the set of faulty physical locations and, for $l \geq 1$, let $\mathcal{N}_l$ be the set of exRecs of the $l$-th concatenation (forming locations in the level-$l$ encoded circuit) which are bad, i.e. of whose constituent locations on the $(l{-}1)$-th level more than $\tbad$ belong to $\mathcal{N}_{l-1}$. We consider these bad locations as effective faults and define the coarse-grained noise strength $\eta_l$ as the local noise strength induced on level $l$ in the sense of Def.~\ref{def:local-noise}.

The proof proceeds by induction. Consider the encoding from one level to the next.
For any set $F$ of $|F|$ locations inside an exRec, local noise bounds the fault-path sum $E(F)$ by $\norm{E(F)}_\diamond \leq \eta^{|F|}$.
Take $1 \leq \tbad \leq \tFaults$.
Counting all relevant fault paths with the inclusion-exclusion principle, we bound the fault-path contributions in an exRec with more than $\tbad$ faults (i.e., the bad part of the exRec) as
$$\Bigg\|\sum_{F: |F| \geq \tbad + 1} \Phi^\mathrm{noisy}_F \Bigg\|_\diamond \leq A\,\eta^{\tbad+1},$$
where $A$ is a combinatorial constant depending only on the maximal number of locations within an exRec and not on the simulated circuit.
Thus, a level-1 exRec simulates the ideal logical location with an effective encoded noise strength that is upper-bounded by $\eta_1 \leq A \eta^{\tbad+1}$, and likewise $\eta_{l+1} \leq A \eta_l^{\tbad+1}$ at higher levels.
It follows that repeated application of the fault-tolerant encoding scheme (concatenation) for $k$ levels decreases the logical noise strength doubly exponentially to
$$ \eta_\text{log} := \eta_k \leq \eta_{\text{th}} (\eta / \eta_\text{th})^{(\tbad+1)^k}$$
whenever the physical noise strength $\eta_0 = \eta$ is below a constant threshold $\eta_\text{th} = A^{-1/\tbad}$.
For an ideal circuit of depth $D$ on $N$ qubits, the norm of the sum of all bad fault-path contributions is bounded by $\bigO( N D \, \eta_\mathrm{log} )$, which upper-bounds the simulation error (in trace distance, hence also in total variation distance on measurement outcomes). Choosing
$ k = \bigO\!\left(\log\log(ND / \delta) \right) $
achieves logical accuracy $\delta$ with polylogarithmic overhead, proving Lem.~\ref{lem:threshold-thm}.

Later on (in Sec.~\ref{sec:threshold-proof}), we will make use of the exRec framework to analyze our translation-invariant, self-simulating QCA (constructed in Sec.~\ref{sec:2D-construction} below). In that setting, neighboring exRecs share boundary operations (resulting, for example, from the structure correction in Sec.~\ref{sec:structure-registers}) at \textit{all} spatial boundaries, even if the exRecs belong to two distinct encoded gates. A single fault on a block-boundary-crossing idling location may thereby affect the encoded state in both involved exRecs.
This is not accounted for in the standard exRec formalism of Refs.~\cite{aliferis2006, gottesman2024}, wherein exRecs overlap only temporally (resulting from shared EC gadgets) and in which the precise boundaries of exRecs are assigned recursively for each given fault path by truncation (with the rule that whenever a candidate good exRec is followed by a bad exRec, the shared locations are truncated from the former and counted only in the latter), so that any faulty operation contributes to at most one bad exRec. We omit the truncation and instead let exRecs \textit{overlap}, saying that a lower-level location \textit{participates} in every exRec whose blocks it acts on, and counting it towards the badness of each of these exRecs in Def.~\ref{def:goodness}. Under local operations each location participates in at most $R$ distinct exRecs, where $R$ is a constant depending only on the spatial locality/geometry of the gadget schedule.

The worst-case overhead from boundary faults is then a constant factor in the renormalization exponent, and can be compensated when the concatenated code has a sufficiently high distance:
\begin{lemma}[Noise renormalization with overlaps] \label{lem:threshold-thm-w-overlap}
    Let $\mathcal{C}$ be an error-correcting code with concatenated encoding scheme $\mathcal{E}$ that is $\tEC$-fault-tolerant with spread $s_0$ and $\tEC \geq s_0$ (Def.~\ref{def:FT-conditions}).
    Suppose the scheme's gadgets are local with at most $R$ encoded blocks interacting at each individual level-$l$ operation, so that each level-$l$ location participates in at most $R$ level-$(l{+}1)$ exRecs,
    and assume the physical computer experiences local noise (Def.~\ref{def:local-noise}) of strength $\eta_0 = \eta$.
    Fix $\tbad \leq \tFaults$ and declare an exRec bad when it is $\tbad$-bad.

    Then,
    considering bad exRecs as effective faults on the higher levels of encoding, the coarse-grained noise at each level $l$ is again local with strength $\eta_l$, 
    and there exists a constant $A$ (depending only on the maximal number of locations participating in one exRec) such that $ \eta_{l+1} \leq A \eta_l^{(\tbad +1)/ R}.$
\end{lemma}
To see this, note that for any set $F$ of level-$(l{+}1)$ locations, any fault path in which every location in $F$ is bad must contain at least $\lceil |F| (\tbad + 1)/R \rceil$ distinct level-$l$ effective faults (because each bad exRec contains at least $\tbad +1$ faulty level-$l$ locations and any single such location can contribute to the badness of at most $R$ exRecs). 
The combinatorial counting from the proof of Lem.~\ref{lem:threshold-thm} hence applies verbatim with $\eta_l^{1/R}$ in place of $\eta_l$, yielding the same noise renormalization statement as for the non-overlapping exRec convention, but with reduced renormalization exponent $(\tbad + 1)/R$ instead of $\tbad +1$.

\subsection{Concatenated quantum codes} \label{sec:code-existence}

There exist known quantum error-correcting codes which satisfy the conditions from Sec.~\ref{sec:FT-conditions}:

\begin{lemma}[Fault-tolerant quantum codes] \label{lem:concat-code}
    For any given $\tEC \geq 1$ and any constant qudit dimension $d_{\mathrm{data}}$, there exist a constant $M_\mathrm{code}$ and a quantum error-correcting code $\mathcal{C}$ mapping one $d_{\mathrm{data}}$-level qudit to an $M_\mathrm{code} \times M_\mathrm{code}$ block of $d_{\mathrm{data}}$-level qudits, together with a complete, nearest-neighbor-local gadget set (state preparation, a universal gate set, error correction, and measurement) that is $\tEC$-fault-tolerant in the sense of Def.~\ref{def:FT-conditions}.
    These gadgets define a constructive encoding map $\mathcal{E}$ that transforms an ideal logical circuit $Q$ into a one-level encoded circuit $\mathcal{E}(Q)$.
    The level-$l$ concatenated circuit $\mathcal{E}^l(Q)$ %
    retains 2D nearest-neighbor locality and can be written as a space-time-varying channel that requires no intermediate measurements or classical control.
\end{lemma}

\begin{proof}
    The extended version of Ref.~\cite{aharonov1997} proves the statement for qudits of prime dimension $p$ and supplies exactly the required properties.
    (Specifically, Thm.~4 therein gives for any desired $\tEC$ a quantum polynomial code of length $m =6\tEC + 1$, for any prime $p > m$, that corrects $\tEC$ errors; Thm.~5 therein provides fault-tolerant, measurement-free gadgets of spread $1$ for the universal gate set $\mathcal{G}_1$ as well as for error correction, state-preparation and decoding; and Thm.~15 therein makes the encoded circuit nearest-neighbor-local on a 2D grid at constant overhead per level.)
    For arbitrary qudit dimension $d_\mathrm{data}$, one may embed each $d_\mathrm{data}$-level qudit into a qudit of the smallest prime dimension $p \geq \max(d_\mathrm{data}, 6\tEC + 2)$ and, if needed, represent each $p$-level code qudit within $\lceil \log_{d_\mathrm{data}} p \rceil$ $d_\mathrm{data}$-level qudits (since $d_\mathrm{data}$ and $\tEC$ are constants this yields only an additional constant overhead in the encoding block size and gadget duration).

    For readers familiar with standard concatenated quantum error-correcting codes, we briefly outline how such encoding schemes are made nearest-neighbor-local and written as a measurement-free channel (cf.~\cite{aharonov1997, gottesman2000}).
    Given a non-local scheme, one may use SWAP networks to realize each constant-size base-code gadget in a geometrically local manner. 
    When locality is enforced in this manner already at the first concatenation level, and then preserved recursively, this changes the overhead relative to non-local schemes only by a constant factor.
    Further, the gadgets are realized in a measurement-free way via stabilizer applications to local ancillas that act as control qudits in subsequent operations.
    Each gadget uses only finitely many ancilla qudits per encoded block; hence the ancillas can be counted as part of the constant block size of the base code.
\end{proof}

Note that we may explicitly include dissipative RESET among the fault-tolerant logical channels of the scheme in Lem.~\ref{lem:concat-code}, so in particular the encoded schedule is exactly implementable with channels from $\mathcal{G} = \mathcal{G}_1 \cup \{\text{RESET}\}$ (Eqn.~\ref{eq:universal-gateset}). The RESET can be implemented by discarding (i.e. measuring and forgetting) the input block and then applying the fault-tolerant encoded state-preparation gadget. Thus, the same gadget framework accommodates dissipative logical operations as well as logical gates.

\subsection{Classical repetition code} \label{sec:rep-code}

Our construction additionally stores classical program information (used for the recursive simulation in Sec.~\ref{sec:2D-construction}) in local registers. This needs to be handled separately from the quantum information encoded in the concatenated code above, in order to ensure that the encoded program state remains compatible with classical Toom-correction. 
We use that the program registers remain in computational-basis product states throughout the ideal evolution, so they can be encoded in a classical repetition code. Let us recall here the standard construction. A symbol from some finite alphabet $\Sigma$ is stored as $n_\mathrm{rep}$ identical copies. We arrange these copies in an $M_\mathrm{code} \times M_\mathrm{code}$ block, taking $n_\mathrm{rep} = M_\mathrm{code}^2$, to match the block geometry of Lem.~\ref{lem:concat-code}.
\begin{lemma}[Repetition code] \label{lem:rep-code}
    Consider the $n_\mathrm{rep}$-copy repetition code over a finite alphabet $\Sigma$ and let $\tEC < n_\mathrm{rep}/2$. Then:
    \begin{enumerate}
        \item[(a)] There is a geometrically local error-correction gadget (block-wise majority vote) that is tolerant to $\tEC$ faults, i.e. that satisfies the correctness and propagation properties of Def.~\ref{def:FT-conditions}.
        \item[(b)] For any classical function $f: \Sigma^j \to \Sigma^j$, there is a geometrically local, $\tEC$-fault-tolerant gadget to compute $f$ transversally on $j$ code blocks. 
    \end{enumerate}
\end{lemma}
\begin{proof}
    (a) The majority vote over the $n_\mathrm{rep}$ copies in a block is a classical function with an exact geometrically local implementation from the channel set $\mathcal{G}$ of Eqn.~\ref{eq:universal-gateset}, using freshly reset ancillas and SWAP networks. If at most $\tEC < n_\mathrm{rep}/2$ copies are erroneous and the gadget is fault-free, every output copy carries the ideal symbol. Since the circuit processes the copies in disjoint local routings, $s$ faults leave at most $s$ erroneous output copies regardless of the input.

    (b) The gadget evaluates $f$ transversally as follows: For each $i \in \{1, \dots, n_\mathrm{rep}\}$, a disjoint constant-size circuit reads the $i$-th copies of the $j$ input blocks, routed together by SWAP networks ($n_\mathrm{rep}$ and $j$ are constants, so this preserves nearest-neighbor locality at constant depth), and writes $f$ of them to the $i$-th copies of the output blocks. A single fault affects at most one copy per output block, so $r$ faults within the gadget together with at most $s$ erroneous input copies per block (which corrupt at most $s$ of the disjoint circuits), give at most $r + s$ erroneous copies per output block. Thus, for $r + s \leq \tEC$, blockwise majority vote restores the ideal encoded output.
\end{proof}
\phantom{.}

\section{Discrete time: Self-correcting QCA} \label{sec:2D-construction}

We now lay out the specifics of our self-correcting QCA construction. 
Throughout this Section, $\SimR$ denotes the (for now arbitrary) QCA rule to be simulated, only in the final step will we set the simulated rule equal to the simulator itself.

As depicted in Fig.~\ref{fig:construction-sketch} of the main text, we first fix a universal QCA $\Univ$ that simulates any given automaton via a classical program stored in its initial state (Sec.~\ref{sec:Univ}), and protect it with the fault-tolerant codes described above (Sec.~\ref{sec:combined-encoding}). The resulting encoded update schedule is space-time varying, so we restore translation-invariance by storing believed space-time coordinates in local structure registers, which are themselves corrected by Toom's rule (Sec.~\ref{sec:structure-registers}). 
We then assemble these components into the \textit{error-correcting simulator} $\R$, a translation-invariant QCA that fault-tolerantly simulates $\SimR$ (Sec.~\ref{sec:EC-simulator}).
Finally, we close the self-simulation loop by setting $\SimR = \R$, verifying that the recursive parameter constraints admit a self-consistent solution (Sec.~\ref{sec:self-simulation}).

\subsection{The universal automaton} \label{sec:Univ}

As indicated in the overview, our simulation framework starts with the transition function $\Univ$ of a universal quantum cellular automaton.

When we speak of one QCA simulating another QCA, we mean that there exist block-wise encoding and decoding maps under which the evolution of the former QCA after decoding reproduces the evolution of the latter QCA. 
In our specific $\Univ$ constructed below, the state of the higher-level simulated QCA is stored spread out over several sites in a local quantum data-track register, and there is another local register, the program register, which determines the QCA being simulated via the initial configuration. 
We will later need the property that the attempted gate sequence of the simulation depends only on the program registers and not on the quantum data-track, so an input data error may change the simulated input state but it does not affect the simulated update rule.
Further, we require that the ideal program state of $\Univ$ is a computational-basis product state and remains so throughout the ideal evolution (so we can think of the program registers as classical), and that the program state can be taken periodic (with period twice the encoding block size, see below) on periodic boundary conditions, and we need the encoding block size and time overhead of simulation with $\Univ$ to be linear in the sequential circuit complexity (i.e. the total gate volume of a compiled implementation) of the simulated update rule.

We phrase this more formally as:

\begin{widetext}
\begin{lemma}[Universal automaton]\label{lem:Univ-existence}
Let $\mathcal{G}$ be a finite set of channels on $(\mathbb{C}^{d_\textnormal{data}})^{\otimes 4}$ where $d_\textnormal{data} \geq 2$
and $\mathcal{G}$ contains the SWAP-gates along horizontal, vertical, and diagonal directions of the $2\times 2$ block.
Then there exists a 2D QCA $\Univ$ with local Hilbert space $\mathbb{C}^{d_{\Univ}} = \mathbb{C}^{d_\textnormal{prog}} \otimes  \mathbb{C}^{d_\textnormal{data}} $, where each site contains a data-track register of dimension $d_\textnormal{data}$ and a program register of dimension $d_\textnormal{prog}$,
such that:
\begin{enumerate}
\item[(a)] 
    Its four-site update channel has the form
    \[
    \Univ(\rho) = \sum_{\vec{\gamma} \in \{1, \dots, d_\textnormal{prog}\}^4} 
    \ket{\phi(\vec{\gamma})} \!\bra{\phi(\vec{\gamma})}\otimes \mathcal{F}_{\vec \gamma}\left(\bra{\vec{\gamma}} \rho \ket{\vec{\gamma}}\right),
    \]
    where $\phi: \{1, \dots, d_\textnormal{prog}\}^4 \to  \{1, \dots, d_\textnormal{prog}\}^4$ is a classical map governing the evolution of the program track and $\mathcal{F}_{\vec \gamma} \in \mathcal{G}$ is a program-dependent channel acting on the data track.

\item[(b)] 
    Let $\SimR$ be a 2D QCA with local Hilbert space $\mathbb{C}^{d_{\SimR}}$ where $d_{\SimR} = d_\textnormal{data}^{k^2}$ for some $k$. Suppose that its update rule viewed as a channel on $2k \times 2k$ $d_\textnormal{data}$-dimensional qudits is exactly implementable from the set $\mathcal{G}$ with volume $\vol(\SimR)$. 
    Then, %
    for every sufficiently large even spatial block length 
    $M_\mathrm{sim} \geq \Theta(\vol(\SimR))$,
    there is a corresponding temporal block length $T_\mathrm{sim} = \bigO(M_\mathrm{sim})$, an encoding 
    $$
        \Enc: \textnormal{B}\left( \mathbb{C}^{d_{\SimR}} \right) \to 
        \textnormal{B}\left( (\mathbb{C}^{d_\textnormal{data}})^{\otimes M_\mathrm{sim}^2} \right)
    $$ 
    (which maps one site in $\SimR$ to the data-track registers of a $M_\mathrm{sim} \times M_\mathrm{sim}$ block of sites in $\Univ$), and a decoding 
    $\Dec$ with $\Dec \circ \Enc = \Id$.
    With these, for every even system size $L$, there exists a computational-basis program configuration 
    $\ket{\textnormal{prog}_{\SimR, L}}$ on $(M_\mathrm{sim} L)^2$ program registers, which has period $2 M_\mathrm{sim}$ in both spatial directions,
    such that for all $t \in \mathbb{N}$,
    $$
        ({\SimR})^t(\rho) =
            \Dec^{\otimes L^2} 
            \circ \operatorname{Tr}_\textnormal{prog} 
            \circ \left(\Univ_{M_\mathrm{sim} L}\right)^{T_\mathrm{sim} t} 
            \left(
                \ket{\textnormal{prog}_{\SimR,L}}\!\bra{\textnormal{prog}_{\SimR, L}}
                \otimes \Enc^{\otimes L^2}(\rho)
            \right),
    $$
    where $\operatorname{Tr}_\textnormal{prog}$ traces out the program registers on all sites. At every intermediate time the ideal program remains a computational-basis product state with the same spatial period.
\end{enumerate}
\end{lemma}
\end{widetext}

\begin{proof} 
    We adapt the one-dimensional universal automaton of Ref.~\cite{arrighi2008} to two dimensions, constructing $\Univ$ as a brickwork QCA in the sense of Def.~\ref{def:brickwork-QCA}. 

    Let us first describe the local state-space. Each site carries a quantum data-track register of dimension $d_\text{data}$, and a classical program register of dimension $d_\text{prog}$ which holds the tuple 
    $$ \gamma = (p, m, e, \mu_H, \mu_V). $$
    Here, $p \in A \cup \{\emptyset\}$ is the local symbol of the \textit{program tape} over a finite alphabet $A$ (with $\emptyset$ denoting an empty tape cell), $m \in \mathbb{Z}_4$ is the local \textit{mode} steering the data movement, $e \in \{0,1\}$ is an occupancy flag that flags \textit{empty} data locations, and $\mu_H, \mu_V \in \{0, \text{Act.}, \text{Pass.}\}$ flag the presence of active/passive horizontally or vertically moving \textit{mode-change} markers. The resulting $d_\text{prog} = 72 (|A|+1)$ is constant.
    
    Note that, comparing with the three local bands of Ref.~\cite{arrighi2008} (data, program, mode), we have made two modifications: First, the marker flags $\mu_H, \mu_V$ are used to generate the two-dimensional mode pattern below, whereas in one dimension their role was played by a designated tape symbol (`change color' in Ref.~\cite{arrighi2008}). Second, Ref.~\cite{arrighi2008} encodes whether a site is empty in a distinguished basis state of the quantum data band and uses this in the routing, so a data fault could corrupt the routing pattern and thereby affect subsequent gates.
    To prevent this, we instead store the occupancy classically in the flag $e$ and let the data motion access this flag to determine whether to move onto a site. 
    With this, the attempted schedule of operations is determined by the classical program track alone, and a mismatch between $e$ and the data state merely counts as a data fault.

    Intuitively, the universal update can be understood as follows. 
    The data-track registers store the simulated system with simulated cells being spread out over blocks of $d_\text{data}$\nobreakdash-dimensional qudits. The program tape prescribes the schedule of local operations to apply. Gates are enacted when two program symbols collide at non-empty data locations within a stationary window designated by the mode field. The location of two colliding program signals determines the location where an elementary channel from the fixed universal set $\mathcal{G}$ gets applied, whereas the value of the colliding signals determines which elementary channel gets applied.    
    The mode pattern shown in Fig.~\ref{fig:Univ-pattern}, together with the $e$- and $\mu_{H/V}$-flags, coordinates the motion of the data such that the relevant data qudits are routed into position to form a simulated two-dimensional brickwork schedule. 

    \begin{figure}[tb]
        \centering
        \includegraphics[width=0.22\textwidth]{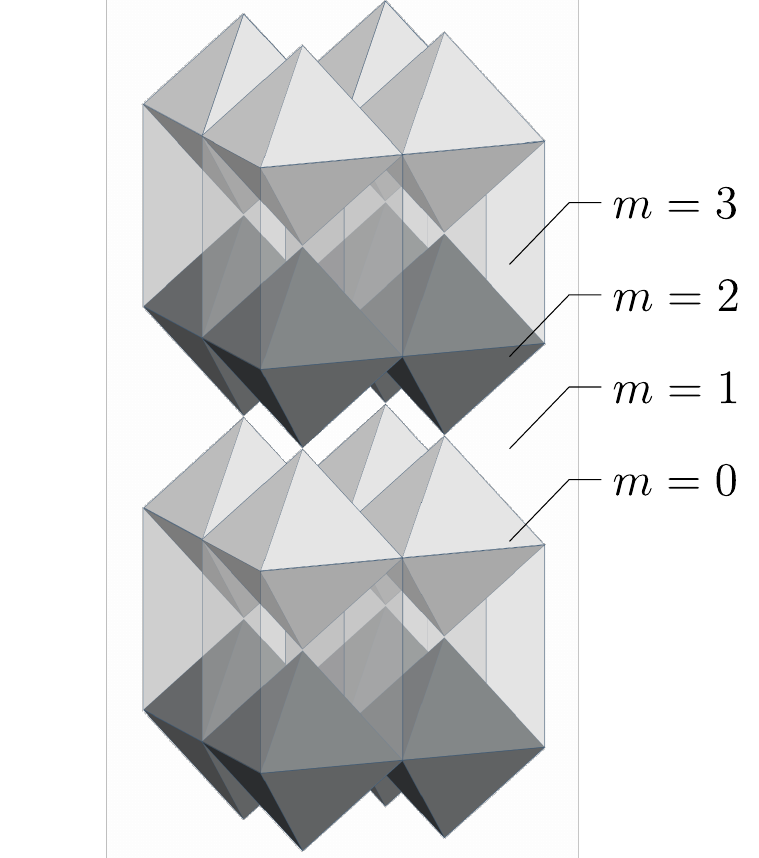}
        \caption{\textbf{Routing pattern for universal simulation of two-dimensional QCAs.}
        We adapt the universal simulator from Ref.~\cite{arrighi2008} to two dimensions.
        Shown is the ideal space-time pattern of the mode field, with time pointing upwards. The sweeping mode-change marker fronts produce a periodic pattern in the mode field, with the octahedral volumes corresponding to constant modes $m = 0$ (light gray octahedra), $m = 2$ (dark gray octahedra), and the transparent space between them corresponding to $m = 1$ (empty regions) or $m = 3$ (lightly shaded regions). The quantum data rests inside the stationary-mode volumes and flows along the slanted edges of incoming $m = 0$ transport regions. 
        \vspace{-2ex}
        }
        \label{fig:Univ-pattern}
    \end{figure}

    \begin{figure*}[p]
        \centering
        \includegraphics[width=0.95\linewidth]{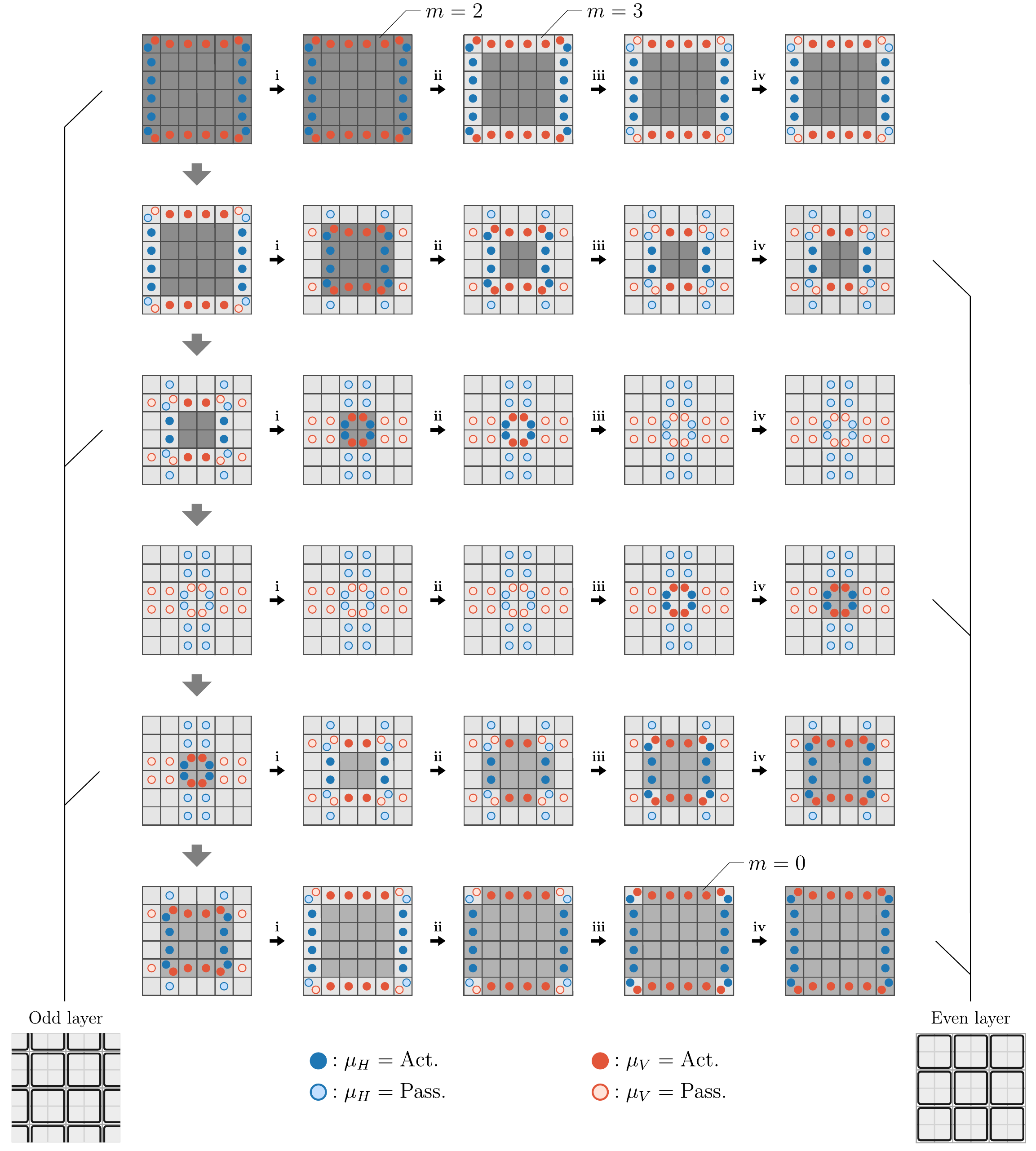}
        \caption{\textbf{Update sequence of the universal QCA.}
        We show the evolution of the mode and mode-change registers under the constructed update channel $\Univ$ for the example of an $M_{\mathrm{sim}}\times M_{\mathrm{sim}}$ block with $M_{\mathrm{sim}}=6$. The full system is initialized by periodically repeating the top left configuration across all blocks. We indicate the local mode $m$ by the cells' shades, and the horizontal and vertical mode-change markers $\mu_H$, $\mu_V$ by onsite blue and red circles, respectively, with filled circles denoting active markers and unfilled circles passive markers. 
        Each row yields one update of $\Univ$, with the relevant substeps shown from left to right. The mode-change flags first evolve in step $4$ of the update where $\mu_H$ and $\mu_V$ propagate horizontally and vertically through brickwork-pattern SWAPs (\textbf{i}). In step $5$, the local mode is then incremented at each site containing at least one active marker (\textbf{ii}), the active/passive status of coincident mode-change markers is flipped (\textbf{iii}), and the mode is incremented whenever this flip left both markers active (\textbf{iv}).
        The resulting front of active mode-change-markers contracts towards the block center and then expands into neighboring blocks, producing exactly the octahedral pattern of constant-mode regions shown in Fig.~\ref{fig:Univ-pattern}.
        }
        \label{fig:Univ-mode-change-sequence}
    \end{figure*}

    To make this more concrete, we label the sites of a $2 \times 2$ plaquette by their positions $(0,0), (0,1), (1,0), (1,1)$ with $(0,0)$ the lower-left corner. The update channel of $\Univ$ first dephases the program registers in the computational basis and then applies the following steps in order (this is the two-dimensional analogue of the scattering unitary in Fig.~12 of Ref.~\cite{arrighi2008}):
    \begin{enumerate}
        \item Diagonal data transport: For each of the two diagonal pairs $\{(0,0), (1,1)\}$ and $\{(0,1), (1,0)\}$, if both sites have mode $m = 0$ and exactly one of the two is occupied ($e + e' = 1$), we swap the two data-track registers together with their occupancy flags. In all other configurations, no action occurs. 
        Thus, data moves only onto empty sites and only inside transport regions designated by the mode.

        \item Flow of the program tape: Swap the $p$-registers between the two sites of each row of the plaquette.

        \item Collision gates: 
        If $p(0,0) = a \neq \emptyset$, $p(1,0) = b \neq \emptyset$ and $p(0,1) = p(1,1) = \emptyset$ (this happens exactly when a right-moving and a left-moving tape symbol are about to cross in the lower row of the plaquette), and if all four sites carry $m \neq 0$ (so that the plaquette lies in a stationary region),
        then apply the channel $G(a,b)$ to the four data-track registers of the plaquette. Here, $G$ is a fixed lookup table assigning to each pair of tape symbols a channel from the addressable channel set $\mathcal{G}$,
        such that every such channel is specified by at least one pair.

        \item Flow of the mode-change flags: Swap the $\mu_H$\nobreakdash-registers horizontally between the two sites of each row of the plaquette. Similarly, swap the $\mu_V$\nobreakdash-registers vertically between the two sites of each column of the plaquette.

        \item Mode changes: On each site, increment the local mode $m \mapsto m + 1 \bmod{4}$ whenever at least one active mode-change marker was placed onto the cell in step $4$. 
        Then, flip the state $\text{Act.} \leftrightarrow \text{Pass.}$ of each mode-change marker on sites which now hold both vertical and horizontal markers ($\mu_H \neq 0$ and ${\mu_V \neq 0}$). On sites where this flip results in $(\mu_H , \mu_V) = (\text{Act.}, \text{Act.})$, also increment the local mode $m \mapsto m + 1 \bmod{4}$.
    \end{enumerate}

    Steps $2,4,5$ act only on the program registers, steps $1$ and $3$ act on the data track controlled by the program, and no step is controlled by the quantum data. The update channel is hence of the form required in (a), i.e. the program track evolves under a classical map $\phi$ (given by composing the occupancy swaps of step $1$ with steps $2,4,5$) and the data track undergoes a channel $\mathcal{F}_{\vec \gamma}$ that is fully determined by the program values in $\vec \gamma$ (i.e. by the $\gamma$-registers in the $2\times 2$ block neighborhood). 

    The addressable set $\mathcal{G}$ may be chosen as an arbitrary finite set of channels on $(\mathbb{C}^{d_\textnormal{data}})^{\otimes 4}$ at the cost of correspondingly enlarging the alphabet $A$. In particular, the addressable gates may include dissipative channels such as local reset operations, so the simulation can be made universal for local quantum channels including non-unitary evolution. In order for $\Univ$ to be exactly implementable from the same channel set $\mathcal{G}$, we require that $\mathcal{G}$ contains the SWAP-gates along horizontal, vertical, and diagonal directions of the $2\times 2$ block.

    Let us examine the signal dynamics under this update channel. Note that the $\Univ$-updates are applied in a two-dimensional brickwork pattern, so the acted-on $2\times 2$ plaquettes are shifted diagonally with each update. Repeating a SWAP-gate in this pattern yields unit-speed motion in a fixed direction, set for each signal by the parity of its initial position (and by the orientation of the SWAP-gate). Thus, moving markers/symbols implicitly remember their direction of movement. For the program symbols and mode-change markers (which are swapped along horizontal rows or vertical columns with each application of $\Univ$), this means that they propagate with unit speed at all times. 
    For data-states this means that they move diagonally under the above updates whenever they are placed in a mode\nobreakdash-$0$ transport region and are facing unoccupied space to move into. Resting data registers start to move at the first update where this transport condition is satisfied, in the direction determined by the parity of that update.
    The routing is engineered entirely through these transport regions via the mode field.

    The ideal pattern of the mode field evolving in space-time is shown in Fig.~\ref{fig:Univ-pattern}. (This replaces the hexagonal routing pattern of the one-dimensional construction from Ref.~\cite{arrighi2008}.) Let us explain how this pattern is enforced. By the update rule above, the mode field is static except at step $5$, in which it is incremented along the paths of active mode-change markers. 
    The initial configuration should ideally be prepared such that these markers form straight horizontal and vertical lines spanning the lattice, placed along the edges of $M_\mathrm{sim} \times M_\mathrm{sim}$ simulation blocks, with horizontal edges carrying north- and south-moving fronts (in $\mu_V$) and vertical block edges carrying east- and west-moving fronts (in $\mu_H$). We choose even $M_\mathrm{sim}$, so that the brickwork parity of these lines/fronts (and thus their direction of movement) alternates.
    Each front sweeps the lattice at unit speed, passing freely through all other signals and through counter-moving fronts. Flips between the active/passive status occur at the crossing points of horizontal and vertical fronts. Within each block, this realizes a contracting-and-expanding pattern of active mode-change markers. 
    So an initial rectangular inward-moving front will first contract towards the block-center, then re-expand and finally exit into the neighboring blocks where it coincides with the fronts entering from there and the sequence repeats, see Fig.~\ref{fig:Univ-mode-change-sequence}. Pictured in space-time this exactly traces out the boundaries at which the mode registers change in Fig.~\ref{fig:Univ-pattern}. 

    \begin{figure}[tb]
        \centering
        \includegraphics[width=1.0\linewidth]{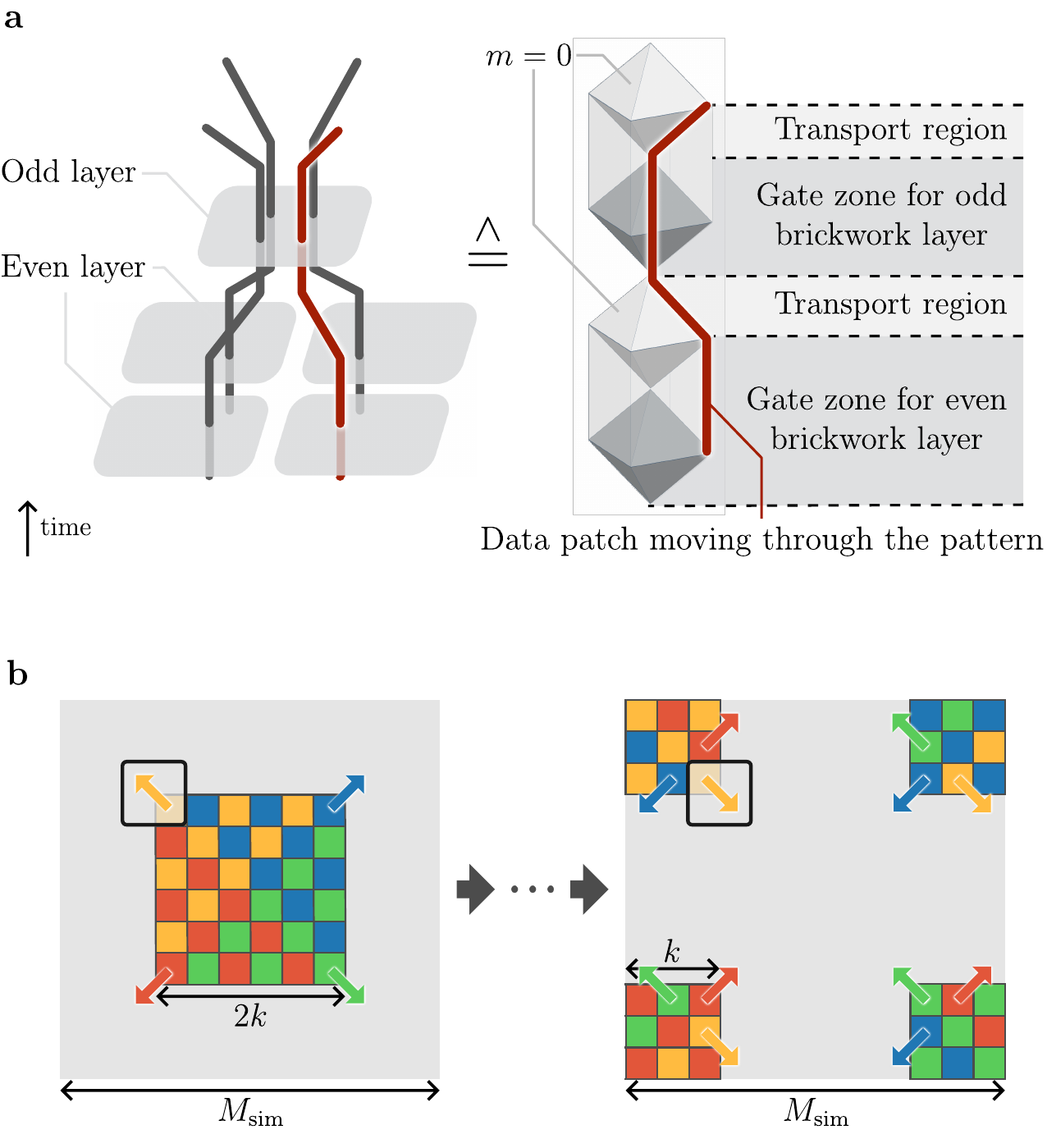}
        \caption{\textbf{Data transport through the mode pattern.}
        \\\textbf{a.}
        In our universal simulator, the data is routed by the mode field to implement alternating brickwork layers of the simulated QCA. The left panel shows two consecutive target brickwork updates, and the right panel shows part of the corresponding space-time pattern on the mode field. The $k^2$ data qudits of each simulated cell are brought together with those of its neighboring cells and subsequently separated (and brought together with the alternating-layer's neighbors), as required. Under the update channel $\Univ$, diagonal transport occurs only in regions with mode $m=0$ and otherwise the data remains stationary. The red line traces the data of one representative cell (specifically the lower-right cell of the simulated odd-layer update) through two successive simulated updates.
        \textbf{b.}
        The gathered data qudits of four simulated cells are interspersed in the $2k \times 2k$ region around the cells' common corner in the depicted pattern. The left panel shows the arrangement during an even layer simulated update, and the right panel during an odd layer simulated update. To alternate between these, data qudits on sites of odd parity (shaded red or blue) move back and forth along the NE-SW diagonal during each movement phase, whereas data qudits on sites of even parity (shaded yellow or green) move back and forth along the NW-SE diagonal. The compiled schedule of the simulated update is written directly for the data qudits in this interspersed arrangement. 
        }
        \label{fig:Univ-data-moving}
    \end{figure}

    The movement of the quantum data is coordinated by said background pattern (see Fig.~\ref{fig:Univ-data-moving}.a). 
    Data registers are swapped only in regions with mode $0$ (light gray in Fig.~\ref{fig:Univ-pattern}), remaining stationary at all other sites. Data registers move only onto unoccupied sites, which ensures that 
    data is kept together
    despite the `tips' of the octahedra in the mode pattern having smaller width.
    Note that, since the diagonal SWAPs of each plaquette act along fixed axes, each data qudit only ever moves along one of the diagonals determined by the parity of its site (with the direction along this axis set by the mode field transport regions and the surrounding occupancies). The sub-patches of data representing distinct simulated cells are thus interspersed in the $2k\times 2k$ data qudit region of the simulating block (see Fig.~\ref{fig:Univ-data-moving}.b).

    To simulate a given QCA $\SimR$ as in (b), one decomposes the desired update of $\SimR$ as an ordered sequence of $\vol(\SimR)$ channels, each supported on a $2\times 2$ plaquette of $d_\text{data}$-dimensional qudits and contained in the addressable set $\mathcal{G}$. The encoding $\Enc$ places the $k^2$ qudits of one cell of $\SimR$ at fixed sites in the interior of the $M_\mathrm{sim} \times M_\mathrm{sim}$ block, %
    and $\Dec$ performs the inverse. One update of $\SimR$ is simulated over a period of $T_\mathrm{sim}$ $\Univ$-updates. These first route the data of the four cells participating in the simulated (either odd or even) plaquette towards their common corner, forming an interspersed block of $4 k^2$ gathered data qudits (as in Fig.~\ref{fig:Univ-data-moving}.b), and then, during the subsequent stationary window, enact the desired compiled update schedule determined by the program collisions. 
    The channels are enacted sequentially, one after the other.
    Note that the compiled update schedule of $\SimR$ must be written directly for the data qudits arranged at the positions of the interspersed arrangement from Fig.~\ref{fig:Univ-data-moving}.b, and that the volume $\vol(\SimR)$ refers to the number of operations with respect to this arrangement. (Alternatively, one could also include a SWAP-network before and after the compiled schedule to bring the gathered qudits into $k \times k$ block-shaped patches.)
    For each channel in the sequence, a pair of counter-moving symbols addressing it is placed on the program tape row of the lower edge of its target plaquette, positioned such that the pair crosses at the intended collision point at the intended time. Consecutive collisions are separated in time and each collision consumes only a constant number of tape cells plus a constant spacing on its row (to ensure that distinct collisions do not interfere, cf. Fig.~13 of Ref.~\cite{arrighi2008}).
    The block size and time-overhead of this simulation are hence linear in the number of operations $\vol(\SimR)$ of the compiled update rule. 
    At worst all channels of the sequence address the same row, so a single tape row must store $\bigO(\vol(\SimR))$ symbols and the stationary window must span $\bigO(\vol(\SimR))$ steps. 
    The tape segments are stored in rows on either side of the gathered data region, and with $M_\mathrm{sim} = \bigO(k + 2\vol(\SimR)) = \bigO(\vol(\SimR))$ the transport phase takes $\bigO(M_\mathrm{sim})$ steps. 

    The ideal program configuration is block-periodic with spatial period $2 M_\mathrm{sim}$ in both directions, because each simulated update spans a $2 \times 2$ plaquette of blocks. 
    Under the ideal evolution, all tape symbols and fronts translate at unit speed. 
    In particular, the boundaries of the mode pattern follow unit-speed diagonals in space-time (Fig.~\ref{fig:Univ-pattern}), so the program (which also moves with unit speed) enters and leaves each stationary-data window at exactly the same position relative to its block. The program symbols emerging from one simulated update are thereby automatically placed correctly for re-use in the next update.
    Thus, the tape and marker configurations recur after any multiple of $2 M_\mathrm{sim}$ steps. Choosing $T_\mathrm{sim} = \bigO(M_\mathrm{sim})$ as a suitable multiple of $2 M_\mathrm{sim}$ and padding the schedule with idle steps as needed, the full program configuration is time-periodic with period $T_\mathrm{sim}$. Iterating the simulation routing above then implements $t$ steps of $\SimR$, as required in (b).
\end{proof}
\emph{Remark.}  An alternative universally programmable QCA %
is given in Ref.~\cite{shepherd2006}, which is based on shifting a program band past a data band with some ancilla and pointer states.
It could also be possible to make that construction suitable to our requirements by inserting ancilla and pointer resets at the beginning of each compiled gate, treating these bands as part of the data track and enforcing the desired simulation schedule regardless of their initial state.

From now on we fix one such $\Univ$. The symbol $\gamma(t,i,j)$ denotes its complete local classical program symbol, including the occupancy and routing registers above. Across $M_\mathrm{sim}\times M_\mathrm{sim}$-sized encoding blocks the layer formed by these $\gamma$-symbols specifies the simulated QCA (with the ideal program configuration repeating with period $2M_\mathrm{sim}$). $\Univ$ is capable of arbitrary computation, and in particular of simulating an arbitrary interacting quantum system (universal simulation, see properties and overhead stated in Lem.~\ref{lem:Univ-existence}).

\subsection{Encoding the universal automaton} \label{sec:combined-encoding}

Recall from Eqn.~\ref{eq:universal-gateset} the set $\mathcal{G}$ of local operations for which there exists a fault-tolerant encoding scheme furnished by Lem.~\ref{lem:concat-code}.
We take an automaton $\SimR$ whose update channel is exactly implementable with channels from $\mathcal{G}$ and map it to another automaton which performs error correction and is also exactly implementable with channels from $\mathcal{G}$ (this will become the error-correcting simulator $\R$ below). We begin by using Lem.~\ref{lem:Univ-existence} to first compile $\SimR$ to the universal automaton $\Univ$.
We choose $\Univ$ so that its addressable channel set consists exactly of the operations from $\mathcal{G}$, as in Lem.~\ref{lem:Univ-existence}, fixing a constant local cell dimension $d_{\Univ} $.
Let us now explain how to apply the fault-tolerant encodings of Sec.~\ref{sec:FT-encodings} on this universal update channel.

We use the encoding scheme from Lem.~\ref{lem:concat-code} for the quantum registers in the data track and combine it with the classical repetition code from Lem.~\ref{lem:rep-code} acting on the program symbols.
Because the program remains in product states throughout the evolution and acts only as transversal control, its repetition-code encoded version is again a computational-basis product state and can be protected by a classical majority correction gadget in parallel with the quantum data-code gadgets.
The combined scheme satisfies the properties listed in Lem.~\ref{lem:concat-code}, yielding a fault-tolerant concatenated correction scheme with measurement-free nearest-neighbor operations, all written in the same primitive $\mathcal{G}$. (In particular, the recursive self-simulation constructed in Sec.~\ref{sec:self-simulation} below does not invoke any gate-set approximation.)

Let us be more specific about the workings of this combined encoding scheme.
The concatenated quantum error-correcting code from Lem.~\ref{lem:concat-code} is used to directly map the data-track qudit at each site of $\Univ$ to an $M_\text{code}\times M_\text{code}$ block of data-track qudits, while the program registers are encoded with the classical repetition code of Lem.~\ref{lem:rep-code}.
This can be written as an encoding map 
$\Penc: (\mathbb{C}^{d_\text{data}} \otimes \mathbb{C}^{d_\text{prog}})^{\otimes L^2} \to  (\mathbb{C}^{d_\text{data}} \otimes \mathbb{C}^{d_\text{prog}})^{\otimes M_\text{code}^2 L^2}$, 
mapping the channel $\Univ_L$ to its encoded variant, such that
$
    \Penc = \Phi_{\text{EC}}^{\otimes L^2} \otimes  \Phi_{\text{rep}}^{\otimes L^2}.
$
Here, $\Phi_\text{EC}$ denotes the error-correcting code from Lem.~\ref{lem:concat-code}, mapping one $d_\text{data}$-level qudit to a grid of $M_\text{code}\times M_\text{code}$ $d_\text{data}$-level qudits.
$\Phi_{\text{rep}} = V_\text{rep}(\cdot) V_\text{rep}^\dagger$ denotes the classical repetition code from Lem.~\ref{lem:rep-code} in the computational basis, with $V_\text{rep}\ket{\gamma} = \ket{\gamma}^{\otimes M_\text{code}^2}$ for $\gamma \in \{1, \dots, d_\text{prog}\}$.

The update channel $\Univ$ can be executed fault-tolerantly with respect to this combined code as follows:
Since $\Univ$ acts on $2 \times 2$ blocks of cells, its encoded variant acts on $2M_\text{code} \times 2M_\text{code}$ blocks, with each site holding a $d_\text{prog}$-level qudit and a $d_\text{data}$-level qudit. 
For each possible logical program configuration $\vec{\gamma} \in \{1, \dots, d_\text{prog}\}^4$ on the $2 \times 2$ block, we compile the corresponding logical update $\mathcal{F}_{\vec \gamma}$ on the data-track qudits (see Lem.~\ref{lem:Univ-existence}) into its encoded version (acting on a $2M_\text{code} \times 2M_\text{code}$ grid of $d_\text{data}$-level qudits) as a circuit of 2D geometrically local channels, $F_{1, \vec \gamma},\ F_{2, \vec \gamma},\ \dots$, which is possible as per Lem.~\ref{lem:concat-code}. 
To simplify our proofs later on, this compiled circuit includes error-correction gadgets of the quantum code from Lem.~\ref{lem:concat-code} on all involved code blocks \textit{both before and after} the encoded operation (so there are redundant error-correction gadgets when executing several encoded $\Univ$ updates in sequence). 
For the combined update we then define 
$
    \widetilde{\mathcal{F}}_i:
        \text{B}((\mathbb{C}^{d_\text{data}})^{\otimes 4} \otimes (\mathbb{C}^{d_\text{prog}})^{\otimes 4}) 
        \to 
        \text{B}((\mathbb{C}^{d_\text{data}})^{\otimes 4} \otimes (\mathbb{C}^{d_\text{prog}})^{\otimes 4})
$ 
such that 
$\widetilde{\mathcal{F}}_i\left(\rho_\text{data} \otimes \ket{\vec \gamma}\!\bra{\vec \gamma}\right) 
    = {F}_{i, \vec \gamma}(\rho_\text{data}) \otimes \ket{\vec \gamma}\!\bra{\vec \gamma}$.
We apply these channels controlled on the program qudits.
For each $i \in \{1, 2, \dots \}$, we:
\begin{enumerate}
    \item Decode the repetition code on the program configuration of the $2 M_\text{code} \times 2 M_\text{code}$ blocks to obtain $\vec{\gamma} \in \{1, 2, \dots, d_\text{prog}\}^4$, overwriting the value from the previous decoding iteration.
    \item Then, apply $\widetilde{\mathcal{F}}_{i}$, conditioned on $\vec{\gamma}$.
    \item Then, apply the error-correction procedure of the classical repetition code on the program blocks.
\end{enumerate}
Finally, after implementing the above for every $i$, we apply the classical map $\phi: \{1, \dots, d_\textnormal{prog}\}^4 \to  \{1, \dots, d_\textnormal{prog}\}^4$ (see Lem.~\ref{lem:Univ-existence}) fault-tolerantly on the program qudits in the repetition code. This is possible as per Lem.~\ref{lem:rep-code}.

Altogether, the above constitutes the schedule of a fault-tolerantly encoded $\Univ$-update on the combined code. 
We consider this to be written as a 2D brickwork circuit with each local channel being applied on $2 \times 2$ sites and define $T_\text{code}$ to be the resulting circuit depth. Hence $T_\text{code}$ depends on the chosen update channel for $\Univ$ and on the used error-correcting code, but otherwise it is independent of all the later parameter choices made below and of the computation being performed.

\begin{lemma}[Fault tolerance of the encoded update]\label{lem:encoded-Univ-FT}
    The encoded implementation of $\Univ$ described above is $\tEC$-fault-tolerant with spread $s_0 \leq 4$ in the sense of Def.~\ref{def:FT-conditions}, with one fault being an incorrectly applied four-site channel of the schedule.
\end{lemma}
\begin{proof}
    The data gadgets (i.e. the encoded circuit for $\mathcal{F}_{\vec \gamma}$ which includes leading and trailing EC gadgets) are $\tEC$-fault-tolerant by Lem.~\ref{lem:concat-code}. Similarly, the program gadgets (i.e. the repetition-code correction and the transversal map $\phi$) are $\tEC$-fault-tolerant by Lem.~\ref{lem:rep-code}.
    For the program-controlled steps, the repetition-code decoding of the program symbols is repeated anew for every $i$, meaning that the physical channels $F_{i, \vec \gamma}$ are each controlled by a separately decoded copy of $\vec \gamma$.
    Thus $\leq \tEC$ input errors per encoding block have no effect on any of the decoded values, and each fault during such a decoding and control-read-out can affect at most the four program symbol copies in its support.
    It remains to account for the spread produced by the difference in fault-counting conventions: In Lems.~\ref{lem:concat-code} and \ref{lem:rep-code} the fault-tolerance refers to the schedule's noisy one- or two-qudit gates, whereas here individual faults may each affect entire $2 \times 2$ plaquettes. Since the plaquettes of the two-dimensional scheme contain at most $4$ qudits each, the resulting spread on the combined encoding is $s_0 \leq 4$.
    Composing gadgets that individually satisfy the correctness and propagation properties hence yields the claimed fault-tolerant schedule.
\end{proof}

\subsection{Structure registers and Toom-correction} \label{sec:structure-registers}

The operation-schedule produced in Secs.~\ref{sec:Univ}--\ref{sec:combined-encoding} is space-time varying within the blocks. We now restore space-time invariance by introducing additional registers at each site. Specifically, the local states in our scheme will carry two kinds of information, stored in separate registers (see Fig.~\ref{fig:state-space}): 
\begin{itemize}
    \item a \textit{structural state}, giving the believed location of the site on the lattice and the current time,
    \item and a \textit{simulation state}, carrying the quantum information represented at this site, including higher-level data, program information and workspace/ancilla registers.
\end{itemize}
For intuition, note that both of these are here considered within a single level of what will become a self-simulating hierarchical organization. For example, the location is considered modulo the block size within a single level. 

We will later fix the block size $M$ and the cycle period $T_0$ to absolute constants, independent of the desired circuit, target accuracy and total lattice size (Secs.~\ref{sec:self-simulation} and \ref{sec:damage-correctness}), but for now we keep these as open parameters. Recall from Lem.~\ref{lem:Univ-existence} that one site of the simulated automaton $\SimR$ is mapped to $M_\mathrm{sim} \times M_\mathrm{sim}$ sites of $\Univ$, and each site of $\Univ$ is mapped to $M_\mathrm{code} \times M_\mathrm{code}$ sites under the error-correcting code. So, up to this point, one site of $\SimR$ is mapped to an $M \times M$ block of sites with $M = M_\mathrm{sim} M_\mathrm{code}$. 
The $2 M_\mathrm{sim}$-periodic program of Lem.~\ref{lem:Univ-existence} becomes $2M$-periodic after this encoding.
The simulator will execute one step of the simulated QCA over several repetitions of a fixed $T_0$-step cycle (defined in Sec.~\ref{sec:EC-simulator} below), and we denote the total number of physical QCA steps required to simulate one step of the encoded higher-level QCA as~$T := T_0 T_\mathrm{sim}$.

\begin{figure}[tb]
    \centering
    \includegraphics[width=0.475\textwidth]{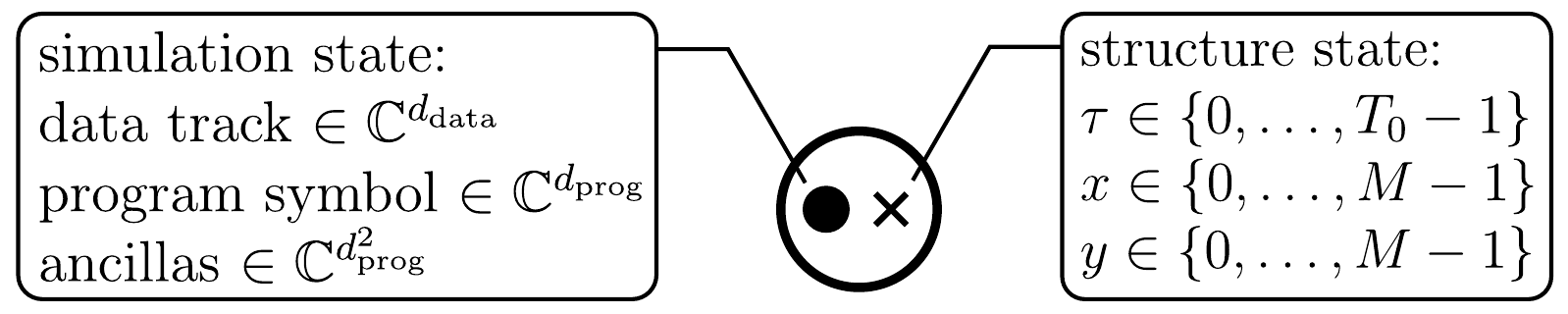}
    \caption{\textbf{Local state-space of the automaton.}
    The quantum cellular automaton acts separately on local simulation data and structure degrees of freedom which store the quantum information of the universal simulation and the believed space-time coordinates, respectively.
    }
    \label{fig:state-space}
\end{figure}

Let $\s$ be a trajectory of the system on the lattice $\Lambda = \mathbb{Z}_L^2$.
By $(t,i,j)$ we refer to the \textit{actual} space-time coordinate, so $t \in \mathbb{Z}$ is the number of updates that have been performed and $(i,j) \in \mathbb{Z}_L^2$ is the real-space lattice site.
The structural state at each site represents the locally \textit{believed} space-time location. 
We partition the entire lattice of sites into $M \times M$ blocks and also divide the time into intervals of $T_0$ steps. At each site we locally store the coordinates
\begin{gather*}
    \tau \in \{0, 1, \dots, T_0 - 1\},\\
    (x, y) \in \{0, 1, \dots, M - 1\}^2,
\end{gather*}
relative to the $M \times M \times T_0$ space-time blocks.
Note that the role of $\tau$ is only to select the next operation within the fixed $T_0$-step cycle, so it does not label the iteration count of that cycle.
In the ideal evolution of the automaton we are about to construct, the structural information in $\tau,x,y$ never entangles with the simulation state, so it can be operated on in an essentially classical manner by measuring the structural layer at each update, without compromising the quantum information stored in other registers. 
Thus, the structural state can indeed be corrected with Toom's rule.
(One should, however, keep in mind that 
the current-level structure registers are auxiliary classical control variables only from that level's perspective, whereas the higher-level simulated cells, including their respective structural states, are placed entirely in the current level's quantum simulation data.)

Let us redefine Toom's rule so that it corrects periodic, instead of just constant, structural information:
\begin{definition}[Toom's rule for structural information] \label{def:Toom-structural}
Trajectories $\s(t,i,j)$ of structural states $\s = (\tau, x, y)$ follow Toom's rule $(\tau,x,y) \mapsto \Toom((\tau,x,y))$ iff
\begin{align*}
    \tau(t+1, i,j) = \Maj\big(\tau(t,i,j),\; &\tau(t,i+1,j),\;  \\
        &\;\tau(t,i,j+1)\big) + 1 \bmod{T_0},\\
    \x(t+1, i,j) = \Maj\big(\x(t,i,j),\; &(\x(t,i+1,j) - 1 \bmod{M}),\; \\
        &\; \x(t,i,j+1)\big), \\
    \y(t+1, i,j) = \Maj\big(\y(t,i,j),\; &\y(t,i+1,j),\; \\
        &\; (\y(t,i,j+1) - 1 \bmod{M}) \big)
\end{align*}
at all points on the trajectory.
\end{definition}

With the structure-registers and their Toom-correction in place, we can now recast the space-time-varying encoded update schedule of Sec.~\ref{sec:combined-encoding} as a translation-invariant transition rule:
An encoded $\Univ$ gate refers to a $\Univ$ update encoded with respect to the combined program-data code from Sec.~\ref{sec:combined-encoding}.
We write the $T_\mathrm{code}$\nobreakdash-length sequence of operations in the resulting schedule of an encoded $\Univ$-gate (including leading and trailing error-correction gadgets that map any input into the codespace) as a time- and spatially \textit{varying} transition rule $\Concat(t,i,j)$.
From this we define the transition function $\Concat' := \Concat(\Toom(\tau,\x,\y))$ in which the space-time variation is supplied by a dependence on the (post-Toom-adjustment) local structure registers, rather than by an external clock.
We fix the convention that every structure-dependent operation
on a $2 \times 2$ plaquette is selected according to the post-Toom structure registers at its southwest corner.
The QCA rule will include the application of $\Concat'$ to the simulation registers.
The description of the logical circuit instance $Q$ is stored only in the initial simulation state (as a spatially distributed tape) and is read when selecting the logical gate to simulate by the encoded $\Univ$-update implemented with $\Concat'$.

\subsection{Definition of the error-correcting simulator} \label{sec:EC-simulator}

We now assemble the components of Secs.~\ref{sec:Univ}-\ref{sec:structure-registers} into the transition rule of the \textit{error-correcting simulator}, a translation-invariant and time-independent QCA $\R$ that fault-tolerantly simulates the given rule $\SimR$.
For each iteration of the simulation procedure, the $M$-sized block represents a single qudit of the level above. In the presence of noise, this qudit experiences an effective noise strength that depends on the physical noise on individual sites of the block.
To simplify our proofs later on, we include here an explicit refresh phase of length $T_{\mathrm{ref}} = \bigO(M)$ before each simulation schedule of $\Concat'$. During this phase the structure layer is Toom-corrected at every step, while the program track undergoes a block-wise Toom-update implemented by a sequence of shift steps across neighboring blocks.

\begin{widetext}
\begin{definition}[Error-correcting simulator] \label{def:error-corrected-simulator}
    Fix a universal QCA $\Univ$ in accordance with Lem.~\ref{lem:Univ-existence}. Fix a quantum error-correcting code satisfying Lem.~\ref{lem:concat-code}. We use the same block size $M_\mathrm{code}$ for the quantum data code and for the classical repetition code protecting the program symbols. Let $T_\mathrm{code}$ denote a common upper bound, after padding by identity locations, on the length of the fault-tolerant gadget schedule implementing one encoded $\Univ$ update (including leading and trailing error-correction gadgets), and of one fault-tolerant correction gadget for the classical repetition code on the program symbols.
    Given the transition rule $\SimR$ of some QCA, let $M_\mathrm{sim} $ %
    be the block size and $T_\mathrm{sim}$ the time-overhead for simulating $\SimR$ with $\Univ$, and let $M = M_\mathrm{code} M_\mathrm{sim}$.

    We define the \emph{error-correcting simulator} as the QCA with transition rule $\R$ whose action on each $M \times M$ block follows a $T_0$-step cycle:
    \begin{enumerate}
        \item In the first $T_{\mathrm{ref}} = T_{\textnormal{str-ref}} + T_{\textnormal{prog-ref}}$ steps (refresh phase):
        First, for $T_{\textnormal{str-ref}}$ steps, apply $\Toom$-updates to the structure registers $(\tau,x,y)$ at all sites while keeping the program registers and the quantum data track idle.
        Then, for $T_{\textnormal{prog-ref}} = T_\mathrm{code} + 1 + 2M + 1 + T_\mathrm{code}$ steps, while continuing the $\Toom$-updates on the structure registers, in parallel refresh the program symbols $\gamma$ using two auxiliary local registers $\gamma_H$, $\gamma_V$ as follows:
        \begin{itemize}[label=-]
            \item Apply the fixed schedule of the correction gadget (majority vote in the $M_\mathrm{code}$-sized blocks) of the classical repetition code to the program symbols $\gamma$ once.
            \item Then, reset the auxiliary registers $\gamma_H$, $\gamma_V$ to a blank state and copy the local $\gamma$-symbol into them at each site.
            \item Next, for $2M$ steps, shift $\gamma_H$ westward and $\gamma_V$ southward site-by-site.
            \item Replace each $\gamma$ by the majority of its current value and the imported symbols in $\gamma_H$, $\gamma_V$ originating from sites at distance $2M$ to the east and north (i.e. from the next-nearest blocks in these directions, with the same within-block coordinates). The shift distance $2M$ matches the periodicity of the ideal program configuration.
            \item Finally, again apply the repetition code's correction gadget to each $M_\mathrm{code}$-sized block of program symbols.
        \end{itemize}
        This sequence implements a single fault-tolerant $\Toom$-update on the program symbols across a distance of two blocks.
        The quantum data track is otherwise idle during this phase.

        \item
        In the next $T_\mathrm{code}$ steps (simulation phase):
        Apply the fixed measurement-free gadget schedule $\Concat'(\tau, x, y)$ that performs one step of $\Univ$ on the logical states encoded in the data-track registers (with the refreshed local program symbols specifying the simulated rule).
        In parallel, continue to apply $\Toom$-correction to the structure registers $(\tau,x,y)$ at all sites.
    \end{enumerate}
    The cycle length is $T_0=T_{\mathrm{ref}}+T_{\mathrm{code}}$. The current point in the cycle is determined from the locally believed time $\tau$.
    Repeating the cycle $T_\mathrm{sim}$ times simulates one step of $\SimR$ on the encoded higher-level cell,
    so one such simulated step (macro-step) has duration $T = T_0 T_\mathrm{sim}$.
\end{definition}
\end{widetext}

Note that the full sequence of $T$ updates implementing one macro-step can be viewed as a gadget in the sense of the extended-rectangle formalism from Sec.~\ref{sec:exRec-review}.

\emph{Remark.} Some steps of Def.~\ref{def:error-corrected-simulator} update every site simultaneously based on its north and east neighbors (namely the $\Toom$-updates of the structure registers and the copy/shift/majority steps of the program refresh), so the cycle above does not literally specify $\R$ in the brickwork format of Def.~\ref{def:brickwork-QCA}. But, since the dephased steps are classical, updates of distinct sites still commute, so it is nevertheless legitimate to think of this as a single `layer' of updates applied at all sites in one step. If needed, every such global step could be compiled into two layers of a brickwork pattern of $2 \times 2$ plaquette channels (e.g. for Toom's rule the first layer could copy the neighbor values of the even-layer plaquette into auxiliary buffer registers, and the second layer could collect the remaining neighbor values and complete the update), so $\R$ can also be written as a brickwork QCA in the sense of Def.~\ref{def:brickwork-QCA}. 
For simplicity, our proofs below treat the classical updates as a single step of the QCA. All statements can be carried over verbatim to the brickwork picture with adjusted constants (one location of the one-step picture corresponds to at most two locations in the brickwork picture, so e.g. the effective noise strength increases by a factor of $2$).

\subsection{Self-simulation} \label{sec:self-simulation}

For any prescribed simulated rule $\SimR$, Def.~\ref{def:error-corrected-simulator} yields a QCA $\R$ whose structure registers are continually Toom-corrected from their surroundings and whose simulation registers undergo one level of fault-tolerant error correction while simulating $\SimR$. However, errors on the simulated state itself are never corrected if the automaton $\SimR$ contains no error-correcting mechanism. To address this, instead of implementing an arbitrary automaton, we let $\R$ simulate itself, i.e. we give it a program that specifies \textit{its own transition rule}. With this, the QCA can simulate itself recursively while correcting the structure states and the encoded data within all levels of the arising hierarchy.

The self-simulation is well-defined if the sizes scale favorably, so that the local cells are capable of storing the required space-time information modulo the automaton's own block size and periodicity needed to simulate a single higher-level step via $\Univ$. 
In Lem.~\ref{lem:self-consistent-parameters} below we analyze the resulting recursive constraints for our final QCA, showing that there indeed exists a self-consistent choice of the parameters $M,\ T$, etc.
We will need the following Lemma, which  characterizes the circuit complexity of our QCA's update channel:
\begin{lemma}[Circuit complexity of the simulator] \label{lem:circuit-complexity}
    The update channel $\R$ of the QCA from Def.~\ref{def:error-corrected-simulator} can be implemented exactly from the channel set $\mathcal{G}$ (Eqn.~\ref{eq:universal-gateset}) with $\vol(\R) = \bigO(\polylog(T_0 M))$ gates in total.
\end{lemma}
\begin{proof}
    We bound the complexity of one step of Def.~\ref{def:error-corrected-simulator} when decomposed into operations from $\mathcal{G}$.
    
    On the structure registers, each step of Def.~\ref{def:error-corrected-simulator} is a composition of dephasing into the computational basis with classical comparison, modular arithmetic and majority selection on the $\bigO(\log T_0 +\log M)$ bits encoding $\tau,x,y$.
    These operations admit exact $\mathcal{G}$-circuits, recalling the gate set $\mathcal{G}_1 \subseteq \mathcal{G}$ from Sec.~\ref{sec:QCA-def}:
    \begin{itemize}[label = -]
        \item shifts of register contents are SWAP-networks,
        \item classical copying of computational-basis states is the sum gate applied to a freshly reset target register,
        \item dephasing into the computational basis is exact by copying to a fresh ancilla and then discarding it,
        \item majority selection, comparison and modular arithmetic are classical functions of the $\bigO(\log T_0 + \log M)$ involved bits and can each be computed exactly by a reversible circuit of increment, sum and generalized Toffoli gates using reset ancillas,
        \item and RESET is exact by discarding and preparing a fresh qudit in the blank state $\ket{0}$.
    \end{itemize}
    The resulting total gate count is $\polylog(T_0 M)$.
    
    To analyze the action on the simulation registers, let $\mathcal{A}$ denote the finite set of channels acting on these registers in Def.~\ref{def:error-corrected-simulator}. Besides idle steps and the classical operations on the refresh registers (dephasing into the computational basis and copy/shift/majority gates on $\gamma_{H}, \gamma_V$), $\mathcal{A}$ contains only the finitely many local channels appearing in the fixed encoding schedule $\Concat'$, which by construction are already written in $\mathcal{G}$.
    The local simulation-register update of $\R$ is obtained by first computing a classical control label $a \in \{1,\hdots,|\mathcal{A}| \}$ from the structure layer and then applying the corresponding channel $A_a \in \mathcal{A}$.
    The ideal action depends only on the current subphase (structure-refresh, program-correction, refresh-copy, refresh-shift, refresh-majority, program-correction, simulation) and, during the program-corrections and the simulation phase, on the residues $(\tau-T_{\textnormal{str-ref}}) \bmod T_{\mathrm{code}}, (\tau-T_{\textnormal{str-ref}}-T_\mathrm{code}-2M-2) \bmod T_{\mathrm{code}}, (\tau-T_{\mathrm{ref}}) \bmod T_{\mathrm{code}},\ x\bmod M_{\mathrm{code}},\ y\bmod M_{\mathrm{code}}$. Since $M_{\mathrm{code}}$ and $T_{\mathrm{code}}$ are constants, the map to the required channel of $\Concat'$ is a constant-size lookup table. Thus the selector of the local channel is computable by a control circuit of size $\polylog(T_0M)$, and one full step of $\R$ admits an exact implementation from $\mathcal{G}$ with $\bigO(\polylog(T_0 M))$ gates, as claimed.
    The same counting bounds the length of the $\Univ$-program implementing $\R$ by $|\gamma| = \bigO(\polylog(T_0 M))$, which is the form used in the self-simulation argument below.
\end{proof}

We now close the loop and determine a choice of parameters with which the self-simulation is well-defined. An overview of the involved parameters can be found in Fig.~\ref{fig:self-simulation}.

\begin{lemma}[Self-simulating fixed point] \label{lem:self-consistent-parameters}
    Fix the universal QCA $\Univ$ (and fix the value of $\tEC$ to be used later in Sec.~\ref{sec:damage-correctness}).
    Then, for every prescribed $T_\textnormal{str-ref} = \bigO(M)$,
    there exists a choice of parameters
    $$M,\ M_\mathrm{code},\ T,\ T_\mathrm{ref},\ T_\mathrm{code},\ T_0,\ T_\mathrm{sim},\ d$$
    with which the QCA $\R$ of Def.~\ref{def:error-corrected-simulator} is well-defined on local dimension $d$, and $\Univ$ can simulate $\R$ on $M/M_\mathrm{code}$-sized blocks with time-overhead $T_\mathrm{sim}$.
    (The specific value of $T_\textnormal{str-ref}$ is fixed later in Sec.~\ref{sec:damage-correctness} according to absolutely-constant requirements on the Toom-mechanism, independent of the other parameter choices.)

    Thereby the recursive constraints defining the self-simulation $\SimR = \R$ have a self-consistent solution.
\end{lemma}
\begin{proof}
    Recall that the simulation state at each site of $\R$ holds one cell of $\Univ$ (including its data track of local dimension $d_\mathrm{data}$ and its local program register $\gamma$ of dimension $d_\mathrm{prog}$) together with the two auxiliary program-refresh registers $\gamma_H$, $\gamma_V$. This fixes the simulation-layer dimension $d_\mathrm{sim} = d_{\Univ}  d_\mathrm{prog}^2 = d_\mathrm{data}  d_\mathrm{prog}^3$.
    We use a concatenated code $\mathcal{C}_\mathrm{data}$ via Lem.~\ref{lem:concat-code} on the $d_\mathrm{data}$-dimensional data-track qudits.
    We will later in Sec.~\ref{sec:damage-correctness} choose an absolute constant $\tEC$ (determined by the locality of the update rule and by the number of faults whose correction we impose per $M \times M \times T$ macro-step of $\R$). The distance of the chosen $\mathcal{C}_\mathrm{data}$ is taken large enough to correct at least $\tEC$ errors.
    This fixes the code's block size $M_\mathrm{code}$ and schedule length $T_\mathrm{code}$ needed to execute one encoded, corrected gate gadget.

    \begin{figure}[tb]
        \centering
        \includegraphics[width=0.49\textwidth]{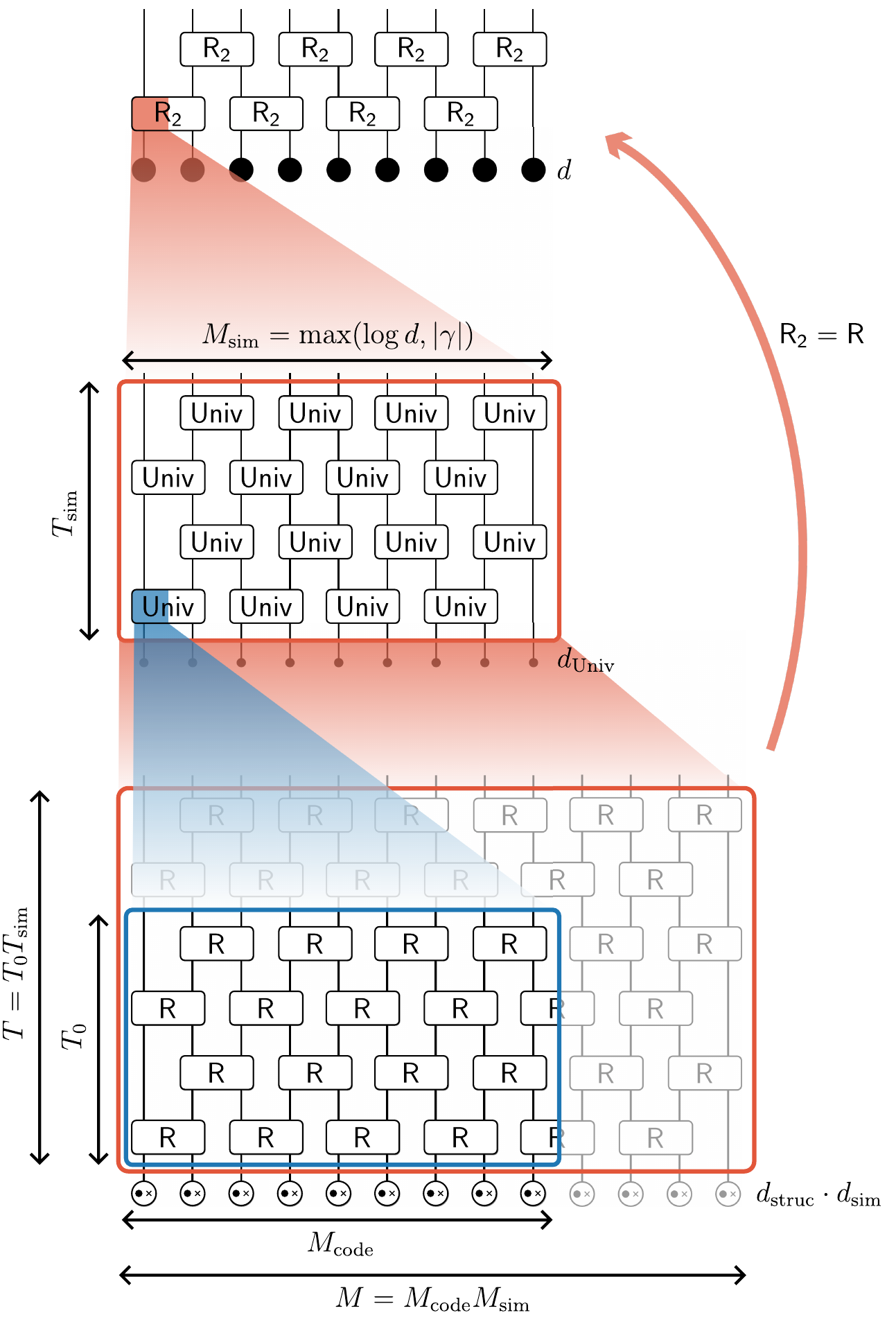}
        \caption{\textbf{Self-simulating setup.}
        We construct a fault-tolerant universal quantum cellular automaton $\R$.
        For a given QCA $\SimR$ which acts on $d$-dimensional qudits, we consider the simulation on $d_{\Univ}$-dimensional qudits with a known, uncorrected universal QCA $\Univ$.
        We compile $\Univ$ with a concatenated encoding scheme and implement the resulting schedule in a translation-invariant, time-independent system using locally stored space-time coordinates $(\tau, x,y)$ which are themselves corrected from a local neighborhood.
        The block-period $M \times M \times T$ that represents one simulated step on the encoded data must be large enough to accommodate the complexity of the simulated automaton $\SimR$.
        To achieve autonomous error correction, we analyze the dependencies between the parameters indicated in this Figure and show that there exists a consistent choice with which the system can self-simulate ($ \SimR = \R$).
        }
        \label{fig:self-simulation}
    \end{figure}

    Now, we examine the simulation with $\Univ$.
    We consider the constraints in terms of the block size $M \geq M_\mathrm{code}$ and period $T = T_0 T_\mathrm{sim} \geq T_\mathrm{code}$ (for simplicity, take integer multiples).
    Each cell of $\R$ contains a simulation register of fixed dimension $d_\mathrm{sim}$. It also contains the structure registers $\tau\in \mathbb{Z}_{T_0}$, $x, y\in \mathbb{Z}_M$ at each site, with dimension $d_\mathrm{struc} = T_0 M^2$. Therefore one cell of $\R$ has local dimension $d = d_\mathrm{sim} d_\mathrm{struc} = \bigO(T_0 M^2)$ and its state can be specified using $\bigO(\log d) = \bigO(\log T_0 + \log M)$ bits.
    For the refresh phase, by definition $T_\mathrm{ref} = T_\textnormal{str-ref} + 2M + 2 T_\mathrm{code} + 2$.
    We keep $T_\textnormal{str-ref}$ arbitrary subject to $T_\textnormal{str-ref} = \bigO(M)$, so this gives $T_\mathrm{ref} = \bigO(M)$. (The self-simulation constraints depend on $T_\textnormal{str-ref}$ only through the constants hidden in $\bigO(M)$, so the later choice in Sec.~\ref{sec:damage-correctness} is admissible.)
    To simulate the QCA $\R$, $\Univ$ spreads out the higher-level local information across several current-level data-track registers and performs a sequence of simulating operations. Specifically, in the universal simulator of Lem.~\ref{lem:Univ-existence}, the overhead from this representation imposes $M \geq \Theta(\max(\log d, |\gamma|))$ and $T_\mathrm{sim} \geq \Theta(M)$ where $|\gamma|$ is the length (i.e. periodicity) of the spatially distributed program. The program length is linear in the circuit volume of decomposing the simulated QCA's update rule into the programmable operations of $\Univ$, that is into $\mathcal{G}$-operations on $\Univ$'s constant-dimensional data track.
    By Lem.~\ref{lem:circuit-complexity}, our specific update rule $\R$ admits such a decomposition with $\vol(\R) = \bigO(\polylog (T_0 M))$, so the corresponding program has length $|\gamma| = \bigO(\polylog (T_0 M))$.

    Let us assemble the dependencies. Note that $T_0 = T_\mathrm{ref} + T_\mathrm{code} = \bigO(M)$, since $T_\mathrm{ref} = \bigO(M)$ and $T_\mathrm{code} = \bigO(1)$.
    The bounds above imply $d = \bigO(d_\mathrm{sim} T_0 M^2)$ and $|\gamma| = \bigO(\polylog(T_0 M))$ so the block size constraint reduces to $M \geq \Theta(\polylog (T_0 M))$. This holds for all sufficiently large $M$.
    After padding by identity program cells if necessary, we choose $M$ divisible by $M_\mathrm{code}$ and large enough to satisfy this condition and the additional constant inequalities imposed in Sec.~\ref{sec:damage-correctness}. The padding may increase $T_\mathrm{sim} = \bigO(M)$, but $T_\mathrm{sim}$ still remains a finite constant and does not enter the constants controlling the $T_0$-scale damage analysis in Sec.~\ref{sec:damage-correctness}. Finally, we set $T = T_0 T_\mathrm{sim}$.
    This gives a finite parameter choice for which $\R$ is capable of simulating itself, so taking $\SimR = \R$ in Def.~\ref{def:error-corrected-simulator} is consistent.
\end{proof}

We fix one parameter choice given by Lem.~\ref{lem:self-consistent-parameters} and set $\SimR = \R$ in Def.~\ref{def:error-corrected-simulator}.
This is the final definition of the error-correcting simulator.
The error correction is automatically scaled with layered self-simulation up to the largest hierarchical level supported by the finite system size.
In the ideal evolution, the circuit schedule is carried along through the recursive pattern of the self-correcting state and the necessary error-correcting operations are performed at each step. Thus, the quantum state that is stored and computed on by the full system effectively undergoes the evolution of the concatenated fault-tolerant circuit $\mathcal{E}^k(Q)$ where $k$ is the highest encoding level of the hierarchy.

We emphasize that the registers $\tau, x, y$ are stored locally in the structure state at each site and corrected with Toom's rule. In contrast, the simulator program is spatially distributed across the local symbols $\gamma$ within the simulation layer, repeating with period $2M$ (i.e. every two blocks)
and is refreshed by the copy/shift/majority Toom-gadget separately for each within-block coordinate. (Since the program encoding in our concatenated code is classical repetition, the site-wise refresh effectively acts as a Toom-update on the encoded program symbols used by the encoded $\Univ$.)
The fault-tolerant encoding schedule $\mathcal{E}(Q)$ of a desired logical circuit $Q$ is distributed over many same-layer blocks within the remaining degrees of freedom (for a larger circuit, this will require more blocks). The corresponding state must be initialized in the simulation layer of the system and is then corrected with the concatenated quantum code.

\emph{Remark.} Our self-simulation mechanism builds on the fact that the one-level encoded $\Univ$ is also a universal simulator, since one may take the program required to simulate a given QCA with $\Univ$, encode it and input it to the encoded $\Univ$.
An alternative approach is to directly use the unencoded $\Univ$ but this would require explicit fault-tolerant error-correction of the simulator, which complicates the analysis (cf. \cite{gacs1989}).
Using a one-level encoded simulator makes the fault-tolerant operation of the QCA relatively transparent and allows for significantly shorter proofs (because we can rely on existing fault-tolerance results from Lem.~\ref{lem:concat-code}, making it for example unnecessary to construct a separate codeword enforcement mechanism).

The next Section of this Supplementary Information is devoted to proving that the QCA $\R$ indeed achieves self-correcting quantum computation.

\section{Threshold theorem} \label{sec:threshold-proof}

We here prove the threshold theorem for our scheme.

\subsection{Overview of the proof} \label{sec:proof-overview}

Following the recursive self-simulation structure of our QCA, we analyze its noisy evolution level by level.
Recall that blocks of physical cells implement encoded cells, which in turn implement cells of the next level, and so on up to the largest scale supported by the system. Our proof will show that each of these simulation levels reliably reduces the effective noise strength. 

We adapt the `extended-rectangle' (exRec) formalism (reviewed in Sec.~\ref{sec:exRec-review}) that was originally developed for fault-tolerance proofs of active concatenated quantum error-correcting circuits~\cite{aliferis2006}. In our setting, space-time blocks implementing a simulated update play the role of the encoded gadgets. The core intuition is that blocks containing sufficiently few faults perform the correct encoded operation (`good exRecs'), while blocks with many faults appear rarely and can be treated as coarse-grained effective faults of the simulated level (`bad exRecs'). The exRec formalism provides a way to keep track of the rescaling of errors in such hierarchical schemes. 
Known sufficient conditions for fault-tolerance (Def.~\ref{def:FT-conditions} in Sec.~\ref{sec:FT-conditions}) ensure that the effective renormalized noise model is indeed local with diminished strength. %

However, we cannot directly invoke the proof strategy of the standard threshold theorem (Lems.~\ref{lem:threshold-thm},~\ref{lem:threshold-thm-w-overlap}) for our purposes. In active schemes the prescribed schedule of operations is hard-coded externally, so the operations after a bad exRec are again performed correctly in subsequent good exRecs (thus, goodness of an exRec immediately implies its correctness, see Lem.~\ref{lem:goodness-correctness}). In our construction, the schedule is instead controlled by the system itself through the locally stored structure and program information. Faults can corrupt this control information and the corrupted region is then repaired \textit{spatially}, by Toom's rule, from its surroundings. This repair takes some time (which scales with the size of the corrupted region) to complete, and in the meantime the affected sites keep performing wrongly scheduled operations, even in the absence of further faults. Goodness of the subsequent exRecs thus does not ensure their correctness and the standard exRec induction does not apply. 

To analyze this delayed repair, we use a bookkeeping device inspired by the classical proofs in Ref.~\cite{gacs1989}. 
For each level $l$ we track the \textit{damage set} $\Dam_l(t)$ which traces
the region at scale $l$ where the decoded correctness may fail. This does not need to be the exact set of erroneous level-$l$ sites, rather it is a controlled overestimate of the sites whose decoded state may have been contaminated by nearby faulty operations. 
In our scheme, this damage set is constructed such that it grows locally in the neighborhood of bad space-time blocks, and it shrinks during noiseless periods according to the north-east geometric erosion of Toom's rule. The proof rests on two properties: 
\begin{itemize}
    \item 
    \textit{Correctness outside damage.} Every simulated operation whose constituents lie entirely outside the damage set implements the ideal encoded update. Shown in Lem.~\ref{lem:correctness-outside-damage} below.
    \item 
    \textit{Sparsity of damage.} Assuming that the system is initially prepared in an appropriate input state, the sum over fault paths in which at least one location of a given set is damaged is bounded linearly in that level's coarse-grained noise strength.
    Shown in Lem.~\ref{lem:sparsity-of-damage} below.
\end{itemize}

Together, these properties substitute for the goodness-implies-correctness step of the exRec proof. Correctness confines all logical deviations to fault paths whose damage reaches the simulated circuit at the top encoding level, and sparsity bounds the weight of such paths by the coarse-grained noise strength. Combining this with the standard level-by-level noise renormalization then yields the threshold theorem for our QCA (see Sec.~\ref{sec:proof-assembly}).

\subsection{Statement of the main theorem} \label{sec:thm-statement}

We need not require the physical initial state to be correct on every level of simulation.
Instead, we allow errors subject to the following condition: 
Starting at the physical level, we recursively mark simulated sites whose constituent encoding blocks contain more than $r$ errors (with the constant $r$ set in Sec.~\ref{sec:damage-correctness} below), counting all sites with structure, program or data errors (compared to the ideal state with program prescribing self-simulation and the desired computation), or inherited marks (i.e. which are damaged at time $t = 0$ in the sense of Def.~\ref{def:damage-set} below). We call the initial state \textit{sufficiently correct} if there exists an encoding level $k_0$ in the system on which no site is marked ($\Dam_{k_0}(0) = \emptyset$). 
This suffices to ensure that the system simulates itself correctly from that level upwards.

\begin{theorem}[formal statement of the main result] \label{thm:self-correction}
    Consider the QCA $\R$ from Def.~\ref{def:error-corrected-simulator} with parameters chosen as in Lem.~\ref{lem:self-consistent-parameters}, and assume local noise of strength $\eta$.
    Assume a sufficiently correct initial state whose level-$k_0$ decoded state is the desired hierarchical input (whose program prescribes self-simulation and the desired computation).

    Then there exists a constant threshold $\eta_{\mathrm{th}}$ (independent of the simulated circuit) such that for all $\eta < \eta_{\mathrm{th}}$, the following holds:
    Let $Q$ be any finite logical computation and compile it into a translation-invariant circuit of depth $D$ on $N$ qubits. Then $Q$ is simulated by $\R$ to within logical accuracy $\delta$ in a finite-sized system with largest encoding level $k = \bigO (\log \log (ND / \delta) )$. Since each simulation level increases spatial size and runtime only by constant factors, the simulation uses $\bigO(N \polylog(ND / \delta))$ physical qudits and total physical runtime $\bigO(D \polylog(ND / \delta))$.
\end{theorem}

The remainder of this Section develops the proof of Thm.~\ref{thm:self-correction}.

\subsection{Coarse-grained locations} \label{sec:proof-notation}

Let us begin by setting up the notation for the hierarchical bookkeeping. 
We define coarse-grained space-time locations at which the proof considers faults and damage. Recall from Sec.~\ref{sec:self-simulation} that $\R$ simulates itself with an $M \times M$ block of cells implementing one cell of the same automaton one level above, and with $T = T_0 T_\mathrm{sim}$ update steps implementing one simulated update step. Iterating this yields a hierarchy of levels ($l = 0, 1, 2, \dots$), which we organize as follows:

\begin{definition}[Macro-locations and exRecs] \label{def:macro-locations}
\phantom{.}
\begin{itemize}
    \item 
    A level\nobreakdash-$0$ \emph{site} is a single physical lattice site. A level\nobreakdash-$l$ \emph{site} $S^{(l)}_{i, j}$, located at position $(i,j)$ of the level-$l$ coarse-grained lattice, refers to the ${M \times M}$ block of level-$(l{-}1)$ sites 
    $$   %
        S^{(l)}_{i, j} = 
        \bigcup_{\substack{iM \leq i'<(i + 1)M \\ jM \leq j' < (j + 1)M}} S^{(l - 1)}_{i', j'} 
        ,
    $$
    comprising $M^l \times M^l$ physical sites in total.
    \item 
    A level\nobreakdash-$0$ \emph{step} is one update step of the physical QCA. A level\nobreakdash-$l$ \emph{step} is one update step of the level-$l$ simulated QCA, implemented by $T$ level-$(l{-}1)$ steps (and hence by $T^l$ physical steps).
    \item
    We describe the evolution at level $l$ in units of level\nobreakdash-$(l{-}1)$ steps and call the times $t \in T_0 \mathbb{Z}$
    at which one $T_0$-cycle of the simulator ends and the next one begins
    the \emph{macro-times} of level $l$.
    \item 
    A level-$l$ \emph{macro-location} is the space-time block consisting of one level-$l$ cell during the $T_0$ update steps of a single cycle of $\R$, i.e. it comprises $M\times M \times T_0$ level\nobreakdash-$(l{-}1)$ sites implementing this cell's part of one encoded $\Univ$-update.
    \item
    A level-$l$ \emph{exRec} is a collection of $T_\mathrm{sim}$ consecutive level-$l$ macro-locations of one level-$l$ cell, i.e. it is the space-time block of this cell during the $T = T_0 T_\mathrm{sim}$ level-$(l{-}1)$ updates which form this cell's part of one simulated update of $\R$.
\end{itemize}
\end{definition}

Note that, since the update rule acts on $2 \times 2$ plaquettes, each simulated level-$l$ operation is implemented by a $2 \times 2$ block of exRecs corresponding to the four involved cells. Moreover, even across the boundaries of these blocks, our QCA update performs non-identity operations, as for example the Toom-updates of the structure registers act between blocks corresponding to several distinct simulated updates. 
The macro-locations and exRecs of Def.~\ref{def:macro-locations} partition space-time cleanly so that every level-$(l{-}1)$ macro-location is contained in exactly one level-$l$ macro-location and in exactly one level-$l$ exRec. To account for the block-crossing operations, we use an \emph{overlapping} convention and say that a level-$(l{-}1)$ macro-location \textit{participates} in a level-$l$ macro-location (or exRec) iff either it is one of the constituents contained in it, or some level-$(l{-}1)$ update acts both on it and on a constituent of that level-$l$ macro-location (exRec). Following the discussion in Sec.~\ref{sec:exRec-noise-renormalization}, we define goodness by counting the number of bad participating lower-level macro-locations under this convention, so a single bad level-$(l{-}1)$ macro-location can contribute to the badness of several level-$l$ macro-locations. 
Since the brickwork updates in our scheme are nearest-neighbor local (i.e. each supported on a square of $2 \times 2$ cells), a level-$(l{-}1)$ macro-location participates in at most $R = 4$ overlapping macro-locations/exRecs, and, as described in Lem.~\ref{lem:threshold-thm-w-overlap}, the renormalization exponent of the coarse-grained effective noise at higher levels is reduced by this overlap factor $R$.

We choose the badness threshold 
$$   \tbad := 2R -1 = 7   $$
so that the renormalization exponent in Lem.~\ref{lem:threshold-thm-w-overlap} takes the value $\alpha = (\tbad + 1) / R = 2$. Let us rephrase the definition of goodness and badness in the present context:
\begin{definition} [Good and bad macro-locations] \label{def:goodness-for-macrolocs}
    In a given fault-path instance, a level-$1$ macro-location is \emph{good} iff at most $\tbad$ of the level-$0$ locations participating in it are faulty; otherwise it is \emph{bad}. At higher levels, a level-$l$ macro-location is \emph{good} iff at most $\tbad$ of the level-$(l{-}1)$ macro-locations participating in it are bad; 
    otherwise it is \emph{bad}.
    We write $\mathcal{N}_l$ for the set of bad level-$l$ macro-locations, and $\eta_l$ for the corresponding local noise strength in the sense of Def.~\ref{def:local-noise}.
\end{definition}

Note that in our setting $\tbad$ differs from the number $\tEC$ of errors which the embedded concatenated code must be able to correct. In Sec.~\ref{sec:damage-correctness} below, we will set $\tEC$ to a constant whose associated fault budget $\tFaults$ (accounting for constant spread $s_0$) is larger than $\tbad$, so that a good higher-level macro-location can correct not only the effective faults directly arising from fresh badness events within it, but also those that stem from a bounded number of contributing lower-level locations which can have incorrect structure/program information and which thus induce additional simulated data-update faults.

To characterize correctness of macro-locations and exRecs, we define
the ideal decoder $\mathcal{D}$. 
This decoder acts on an $M \times M = M_\mathrm{sim} M_\mathrm{code} \times M_\mathrm{sim} M_\mathrm{code}$ block of sites as follows: First, it discards all structure registers and applies the decoders of the quantum error-correcting code from Lem.~\ref{lem:concat-code} and of the repetition code from Lem.~\ref{lem:rep-code} to the contained $M_\mathrm{code} \times M_\mathrm{code}$ sub-blocks of data-track and program qudits, respectively. Then, it applies the decoding map $\Dec$ from Lem.~\ref{lem:Univ-existence} to the resulting decoded $M_\mathrm{sim} \times M_\mathrm{sim}$ block of program and data. The output is the higher-level state of the simulated cell (consisting of the simulated cell's data, program and structure registers, which are all stored within the current level's simulation data).
We denote the decoder $\mathcal{D}$ applied at level $l$ as $\mathcal{D}^{(l)}$.

\begin{definition}[Correctness] \label{def:correct-macrolocs-and-exRecs}
    Recall from Sec.~\ref{sec:combined-encoding} that the encoded $\Univ$ schedule in each cycle begins with a leading error-correction gadget and ends with a trailing error-correction gadget. 

    \begin{itemize}
        \item
        A $2 \times 2$ block of level-$l$ macro-locations implementing one encoded $\Univ$ update is \emph{correct} iff it holds that: 
        \begin{enumerate}
        \item[(i)]
        After decoding with $\mathcal{D}^{(l)}$, the state at the final macro-time of the block is exactly the result of applying one ideal encoded $\Univ$ update to the state decoded immediately after the leading error-correction gadget; and 
        \item[(ii)]
        If the input at the block's initial macro-time has at most $r$ errors per $M_\mathrm{code} \times M_\mathrm{code}$ block (i.e. can be reached by errors on at most $r$ sites per block from the ideal structure and program configurations and an encoded data state, counting all three tracks together)
        then decoding immediately after the leading error-correction gadget gives the same state as decoding at the initial macro-time.
        \end{enumerate}
    
        An individual level-$l$ macro-location is \emph{correct} iff the $2 \times 2$ brickwork block of macro-locations with which it forms an encoded $\Univ$ update is correct.

        \item 
        Similarly, a $2 \times 2$ block of exRecs is \emph{correct} iff the above conditions (i) and (ii) hold with the encoded $\Univ$ update replaced by one ideal noiseless update of $\R$ and with the leading error-correction gadget taken from the block's first cycle. 
        An individual exRec is \emph{correct} iff the $2 \times 2$ block with which it implements a simulated update is correct. At the top level the ideal update of $\R$ includes the evolution under the prescribed circuit $Q$.
    \end{itemize}
\end{definition}

\subsection{Toom geometry} \label{sec:toom-geometry}

Let us now analyze Toom's rule. 
Recall that the structure layer stores local coordinates $(\tau,x,y)$. We say that a site $(i,j)$ is \textit{singular at time $t$} iff its local structure state disagrees with the ideal trajectory, i.e. iff $(\tau(t,i,j),x(t,i,j),y(t,i,j))\neq (t\bmod T_0,\; i\bmod M,\; j\bmod M)$.
We use north-east-oriented triangles
$$\Delta(a,b,c) := \{(i,j)\in \mathbb{Z}^2: i \geq -a, j \geq -b, i+j < c\},$$
where $a,b,c \in \mathbb{Z}$, 
and define the triangle norm $|\Delta(a,b,c)|_{\Delta} := a + b + c$.
For a set $S$ of lattice sites, we set $|S|_\Delta$ to the minimal value attainable by covering $S$ with disjoint triangles and summing over their norms.

\begin{fact}[Toom's rule shrinks triangles of errors] \label{fact:triangle-shrinkage}
    Suppose at time $t$ all singular points in the structure state $\s(t,\cdot)$ lie in $\Delta(a,b,c)$.
    Then after one application of $\Toom$, in the absence of new faults, the set of singular points is contained in $\Delta(a,b,c-1)$.
    The same property holds componentwise for disjoint unions of triangles.
\end{fact}

This is the basic error-correcting property of $\Toom$~\cite{toom1980}.
To prove it for the periodic structure registers, define $\delta \tau(t,i,j):= (\tau(t,i,j) - t) \bmod{T_0}$, $\delta x(t,i,j):= (x(t,i,j) - i) \bmod{M}$, $\delta y(t,i,j):= (y(t,i,j) - j) \bmod{M}$. Def.~\ref{def:Toom-structural} then reduces to standard Toom majority dynamics on $(\delta \tau, \delta x, \delta y)$, so the standard erosion arguments apply unchanged.
By definition, singular points $(i,j)$ can only remain singular if the majority of their neighborhood $\{ (i,j), (i, j+1), (i+1,j) \}$ is singular. The north-east boundary of a triangle surrounding a region of errors thus recedes inwards.

For a triangle $\Delta(a,b,c)$ and an integer $s \geq 0$, we define its deflation $D(\Delta(a,b,c),s):= \Delta(a,b,c-s)$ as the new triangle obtained after $s$ shrinkage steps. (If $c-s \leq -a-b$ then $D(\Delta(a,b,c), s)$ is empty.)
For a family $\mathcal{J}$ of disjoint triangles, we write $D(\mathcal{J},s) := \bigcup_{\Delta \in \mathcal{J}} D(\Delta, s)$.

\subsection{Damage sets} \label{sec:damage-sets}

We are now in a position to define the damage sets announced in Sec.~\ref{sec:proof-overview}. For each level $l$, $\Dam_l(t)$ is determined by that level's initial configuration and the bad macro-locations $\mathcal{N}_l$. 
Damage sets are updated at the macro-times $t = qT_0$ that delimit individual cycles. $\Dam_l(t)$ will be a controlled overestimate of the noise-affected region, not the exact set of erroneous level-$l$ sites. The definition is chosen so that bad macro-locations cause a local contribution to the damaged region, whereas in macro-noise-free space-time regions the damage set shrinks under the Toom geometry described above.

It is convenient to group the cells of each level into larger \textit{clusters}. 
\begin{definition}[Clusters] \label{def:clusters}
A level-$l$ cluster $C^{(l)}_{i, j}$ at $(i, j)$ is the block
$$
    C^{(l)}_{i, j} = 
    \bigcup_{\substack{iT_0 \leq i'<(i + 1)T_0 \\ jT_0 \leq j' < (j + 1)T_0}} S^{(l - 1)}_{i', j'}
$$
of $T_0 \times T_0$ level-$(l{-}1)$ sites. Equivalently, we can also think of this as a block of $(T_0/M) \times (T_0/M)$ level-$l$ sites (because, with the parameters chosen as in Sec.~\ref{sec:damage-correctness} below, $T_0 / M$ is integer).

Given a level-$l$ cluster $C = C^{(l)}_{i, j}$ at $(i, j)$, its cluster neighborhood $C^+$ is the surrounding $3 \times 3$ set of level-$l$ clusters, i.e. the union of the clusters at $(i, j), (i \pm 1, j), (i, j \pm 1), (i \pm 1, j \pm 1), (i \mp 1, j \pm 1)$. This is the relevant spatial region which may affect $C$ during a length-$T_0$ cycle of updates.
\end{definition}

With this, we define:
\begin{definition} [Damage set] \label{def:damage-set}
    The damage set $\Dam_l(t)$ is a set of level-$l$ sites, defined at level-$l$ macro-times $t \in T_0 \mathbb{Z}$.
    All triangle coordinates in this definition are measured in level-$l$ sites, and we note that $T_0 / M$ is integer after padding the cycle (as in Sec.~\ref{sec:damage-correctness}). 
    We consider the initial damage set $\Dam_l(0)$ as given (see Sec.~\ref{sec:damage-correctness} below).
    
    For each subsequent macro-time $t = q T_0$ and each level-$l$ cluster $C$, $\Dam_l((q+1)T_0)\cap C$ is defined as follows:
    \begin{itemize}
        \item If $\mathcal{N}_l \cap ([qT_0, (q+1)T_0)\times C^+) \neq \emptyset ,$ set $${\Dam_l((q+1) T_0) \cap C} := C .$$
        \item Otherwise, let $\mathcal{J}_q(C)$ be a set of disjoint triangles of minimal total triangle norm such that
        $\Dam_l(qT_0) \cap C^+ \subseteq D(\mathcal{J}_q(C),T_0/M) $ and set
        $$\quad \Dam_l((q+1) T_0) \cap C := D(\mathcal{J}_q(C),T_0 / M + 1) \cap C .$$
        (If there are multiple minimal choices, set $\mathcal{J}_q(C)$ by an arbitrary but deterministic rule, e.g. using a fixed ordering of triangles on the lattice. This makes $\Dam_l(t)$ a deterministic function of $\mathcal{N}_l$ and the initial configuration.)
    \end{itemize}
    Thus, in any cluster whose neighborhood is macro-noise-free over one cycle, the triangle cover of the damage set is deflated by at least one level-$l$ site, equivalent to $M$ level-$(l{-}1)$ sites. The $T_0$ steps in the time argument count level-$(l{-}1)$ updates.

    We write $\Dam_l := \{(t,i,j): t \in T_0\mathbb{Z}, (i,j) \in \Dam_l(t) \}$ for the space-time set of damaged sites, and we call a level-$l$ macro-location \emph{damaged} iff its level-$l$ site belongs to $\Dam_l(t)$ at its final macro-time $t$.
\end{definition}

For a set $V$ of level-$l$ macro-locations, we denote by $E_l^\mathrm{dam}(V)$ the sum of fault-path contributions $\Phi^\mathrm{noisy}_F$ over fault paths $F$ for which at least one macro-location in $V$ is damaged.
The next two Subsections establish the two properties (correctness outside $\Dam_l$ and sparsity of $\Dam_l$) required for the fault-tolerance proof.

\subsection{Correctness outside the damage set} \label{sec:damage-correctness}

Before proceeding, we first specify some constants that will appear in the proofs, and we finally fix the parameters that were left open in Lem.~\ref{lem:self-consistent-parameters}:

Let $c$ denote the maximum number of level-$l$ macro-locations which can be newly declared damaged, during one $T_0$-cycle, because of one bad level-$l$ macro-location. Comparing with Def.~\ref{def:damage-set}, $c$ is given by the maximal number of $M\times M \times T_0$-sized macro-locations in a $3 \times 3$ neighborhood of $T_0 \times T_0$-sized clusters during macro-time-duration $T_0$, i.e. $c = (3 \lceil T_0 / M \rceil)^2$. 
With the parameter choice of Lem.~\ref{lem:self-consistent-parameters}, this is an absolute finite constant independent of $M,\ T_\mathrm{sim}$, the simulation level, the system size or the simulated circuit. We fix $c = (3 \cdot 5 )^2$ before choosing the other constants.

Then, we define
$ r := \tbad c $
so that inside any good higher-level macro-location, the lower-level badness events can generate at most $r$ newly damaged sites. Let $\tilde{r} = r + (r-1) + \hdots + 1$ bound the total number of damaged space-time locations which appear when such a set of damaged sites is eroded under Toom's rule.
Let $w$ be the fixed control spread of the implemented gadget schedule, so that one effectively faulty update (which could e.g. be the result of one damaged space-time location with incorrect structure or program information) can change the prescribed channel of at most $w$ locations in the higher-level gadget schedule (including the program-correction gadgets). The constant $w$ depends only on the chosen local interaction geometry, and not on $M,\ T_\mathrm{sim}$, the simulation level, macro-time, system size or the simulated circuit.
For nearest-neighbor local updates in two dimensions we have $w = 4$.
(Note that $w$ counts how many scheduled operations one erroneous control state can corrupt. This is distinct from the fault-tolerance spread $s_0$ of Lem.~\ref{lem:encoded-Univ-FT}, although both constants arise from the same $2 \times 2$ geometry.)

We choose the concatenated code to correct $$\tEC := r + s_0 (w + 1)\tilde{r}$$ errors per block.
Once $\tEC$ is fixed, Lem.~\ref{lem:concat-code} yields $M_\mathrm{code}$ and $T_\mathrm{code}$. We then choose $M$ as a multiple of $M_\mathrm{code}$ large enough for the self-simulation in Lem.~\ref{lem:self-consistent-parameters} and also large enough that
$M \geq 3 r + 3 T_\mathrm{code} + 3.$
Finally, we set $$T_\textnormal{str-ref}:= 2M + 2r + \tbad c + 1 = 2M + 3r + 1 .$$
Inserting into Def.~\ref{def:error-corrected-simulator} gives $T_0 = T_\textnormal{str-ref}+2M+2+3 T_\mathrm{code} = 4M+3r+3 T_\mathrm{code}+3$ and the above lower bound on $M$ implies $\lceil T_0 / M \rceil \leq 5$ (which indeed agrees with the absolutely constant choice $c = (3\cdot 5)^2$). For simplicity, we pad the cycle of Def.~\ref{def:error-corrected-simulator} with idle steps so that $T_0$ becomes a multiple of $M$.

All of these requirements are compatible with the self-consistent parameter choice described in Lem.~\ref{lem:self-consistent-parameters}. 
We summarize the order in which the parameters are fixed as:
\begin{enumerate} [noitemsep]
    \item $R = 4$ and $\tbad = 2R - 1 = 7$ (see Sec.~\ref{sec:proof-notation});
    \item $c = (3 \cdot 5)^2$, $r = \tbad c$, $\tilde{r} = r(r+1)/2$, $w = 4$, and $s_0 = 4$;
    \item $\tEC = r + s_0 (w+1) \tilde{r}$, which via Lem.~\ref{lem:concat-code} also fixes $M_\mathrm{code}$ and $T_\mathrm{code}$;
    \item $M$ (a sufficiently large multiple of $M_\mathrm{code}$, satisfying Lem.~\ref{lem:self-consistent-parameters} and $M \geq 3r + 3T_\mathrm{code} + 3$), then $T_\textnormal{str-ref} = 2M + 3r + 1$ and $T_0 = 4M + 3r + 3T_\mathrm{code} + 3$ (padded to the next multiple of $M$).
\end{enumerate}
This yields a specific constructive scheme in which all our statements can indeed be realized.

Lastly, we now define the initial damage set $\Dam_l(0)$ referenced in Def.~\ref{def:damage-set}.
At the physical level, $\Dam_0(0)$ is the set of all sites whose initial state has an error on the structure or simulation registers. For $l \geq 1$, a level-$l$ site is included in $\Dam_l(0)$ iff it contains more than $r$ level-$(l{-}1)$ sites which are damaged or which have an error on their decoded level-$(l{-}1)$ initial structure or simulation registers.

With this, we prove the correctness property:
\begin{lemma} [Correctness outside the damage set] \label{lem:correctness-outside-damage}
    Consider any fault-path instance and the associated damage sets from Def.~\ref{def:damage-set}.
    For every level $l$ and every simulated level-$l$ operation, if all constituent macro-locations of the $2 \times 2$ block of exRecs implementing that operation 
    are undamaged
    (i.e. have output lying outside $\Dam_l(t)$ when $t \in T_0 \mathbb{Z}$ is the final macro-time of the respective macro-location), then this block of exRecs is correct in the sense of Def.~\ref{def:correct-macrolocs-and-exRecs}. %
\end{lemma}
\begin{proof} 
    We prove the following stronger inductive claim:
    Consider macro-time steps $t = qT_0$.
    Let $X$ be a level-$(l{+}1)$ macro-location on the interval $[qT_0, (q+1) T_0)$. Suppose that the output level-$(l{+}1)$ site of $X$ is outside ${\Dam_{l+1}((q+1)T_0)}$.
    Then:\needspace{3\baselineskip}
    \begin{itemize}[label=-]
        \item during the simulation phase, at times $[qT_0 + T_\mathrm{ref}, (q+1) T_0)$ in $X$, at each time-step there are at most $r$ damaged level-$l$ sites;
        \item at time $(q+1) T_0$, all level-$l$ sites in $X$ outside $\Dam_l((q+1)T_0)$ have the ideal decoded control state, i.e. the ideal structure-layer state and the ideal encoded program symbols;
        \item and, if all four macro-locations of the $2 \times 2$ brickwork block to which $X$ belongs are undamaged,
        then this block of macro-locations is correct in the sense of Def.~\ref{def:correct-macrolocs-and-exRecs}. %
    \end{itemize}
    We prove the claim by induction over $l$, $q$.

    \textit{(i) Base case.}
    At macro-time $t = 0$, the number of lower-level damaged sites which an undamaged macro-location may contain is bounded by the recursive definition of the initial damage sets.
    For every $l$, the definition of $\Dam_l(0)$ already includes all sites whose decoded level-$l$ state differs from the ideal configuration or which contain more than $r$ lower-level damaged constituents, and outside $\Dam_l(0)$ the decoded control and data are correct. So the claim is satisfied at $q = 0$.

    \textit{(ii) Setup of the induction step.}
    Assume now the claim for all levels $\leq l$, and for the level $l{+}1$ up to time $qT_0$. 
    Take a level-$(l{+}1)$ macro-location $X$ %
    and suppose that its output is outside $\Dam_{l+1}((q+1)T_0)$.
    Let $C$ be the level\nobreakdash-$(l{+}1)$ cluster containing $X$'s site (with cluster neighborhood $C^+$, see Def.~\ref{def:clusters}).
    Since $C^+$ must (by Def.~\ref{def:damage-set}) be macro-noise-free, every level-$(l{+}1)$ macro-location in $[qT_0,(q+1)T_0)\times C^+$ has at most $\tbad$ bad participating level-$l$ macro-locations. However, the level-$l$ operations inside $X$ can be incorrect (i.e. can fail to perform the intended simulated operation) not only if they are bad, but also if they are damaged, i.e. if their control information may be erroneous. We now bound the number of such locations.

    \begin{figure*}[tb]
        \centering
        \includegraphics[width=0.95\textwidth]{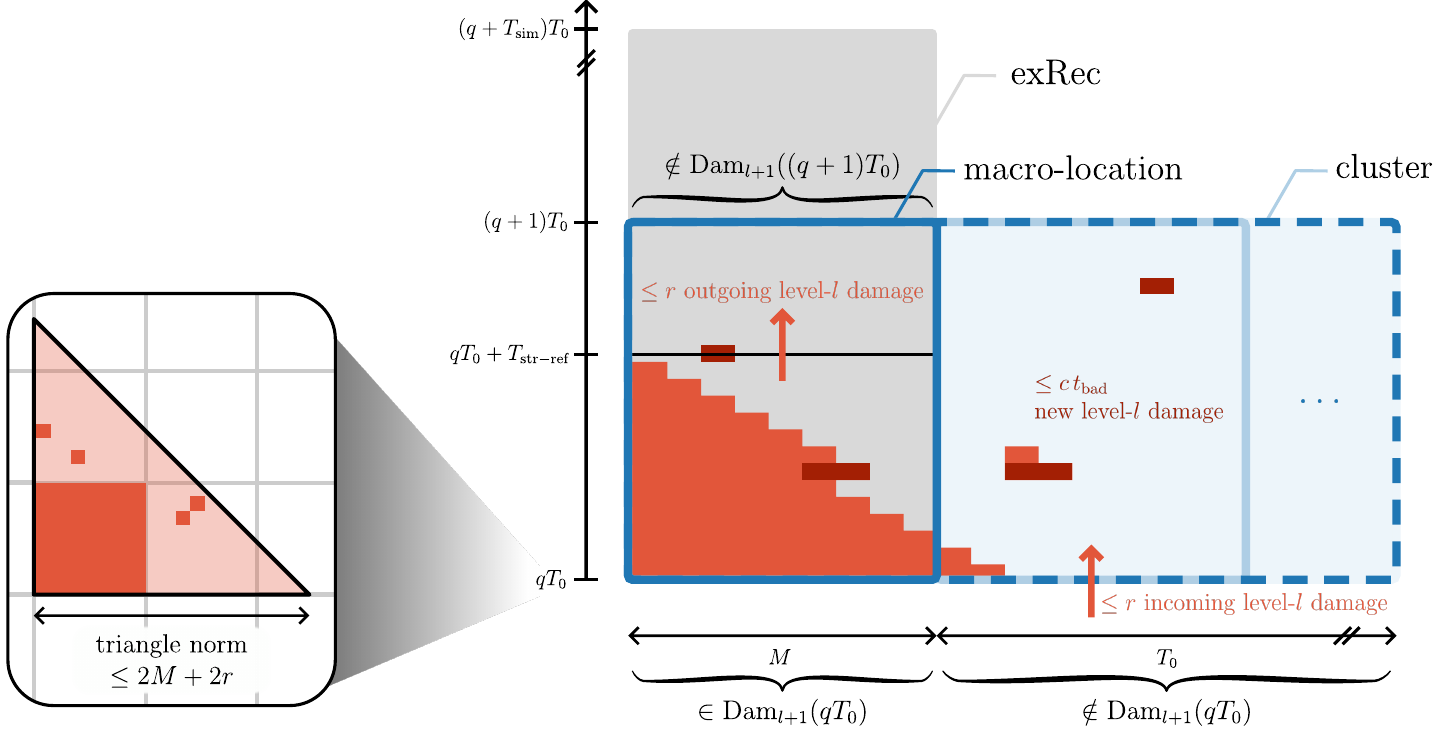}
        \caption{\textbf{Correct operations outside the damage set.}
        We prove that simulated operations implemented by undamaged blocks are correct. The Figure depicts an example configuration where a level-$l{+}1$ macro-location that is damaged at time $qT_0$ becomes undamaged (in the level-$l{+}1$ sense) at time $(q+1) T_0$. Because the surrounding region in north-east-direction of such a macro-location must have been outside $\Dam_{l+1}$ at time $q T_0$, the incoming set of level-$l$ damaged sites is covered by a triangle of bounded norm, as illustrated by the inset.
        The proof analyzes the shrinkage of these lower-level (level-$l$) damaged cells under Toom's rule within a `good' error-correction cycle that contains few coarse-grained faults, i.e. few level-$l$ badness events.
        }
        \label{fig:damage-shrinkage}
    \end{figure*}

    At the beginning of the refresh phase (by the inductive claim from the previous macro-time-step) the prior cycle of each undamaged level-$(l{+}1)$ macro-location %
    finished with at most $r$ outgoing damaged level-$l$ sites. If the macro-location $X$ under consideration is undamaged at time $(q+1) T_0$, then (by Def.~\ref{def:damage-set}) at time $qT_0$ its level-$(l{+}1)$ site either was already outside $\Dam_{l+1}$, or it was at the boundary of the damaged region, i.e. $X$'s site was in $\Dam_{l+1}$ but its north and east neighboring level-$(l{+}1)$ sites were both outside $\Dam_{l+1}$. This means that the triangle norm of the largest triangle in the incoming level-$l$ damage-cover within the cluster around $X$ at time $q T_0$ is bounded by $2M + 2r$ (because, at worst, such a triangle can arise from the level-$(l{+}1)$-damaged site contributing norm $\leq 2M$ and its undamaged neighboring sites contributing norm $\leq r$ each, see Fig.~\ref{fig:damage-shrinkage}). Further, again by the inductive claim from the preceding macro-time-step, the incoming damage region at time $q T_0$ contains all control defects.

    \textit{(iii) Structural refresh.}
    Consider now the structural refresh. During the first $T_\textnormal{str-ref}$ steps of Def.~\ref{def:error-corrected-simulator}, only the structure registers are acted on. 
    By the induction hypothesis from the level below, level-$l$ operations supported on undamaged level-$l$ macro-locations are correct.
    Thus the structure layer follows Toom's rule outside the level-$l$ damaged set (which, by Def.~\ref{def:damage-set}, in turn also shrinks with Toom's geometry).
    In the absence of new bad level-$l$ macro-locations, Fact~\ref{fact:triangle-shrinkage} would shrink the inherited triangle cover with unit speed. The only interruptions are caused by the at most $\tbad$ level-$l$ badness events in $X$. By the definition of the constant $c$, each such bad level-$l$ macro-location yields at most $c$ newly damaged level-$l$ sites.
    Equivalently, the shrinkage budget is reduced by at most $\tbad c = r$. Since $T_\textnormal{str-ref} > 2M + 2r + \tbad c$, the inherited damage and structural defects are completely eroded after $T_\textnormal{str-ref}$ steps. At that time, the only remaining structure defects are hence those belonging to the fresh bursts caused by the new badness events in the current $T_0$-cycle, and there are at most $r$
    such freshly damaged level-$l$ macro-locations. (The same shrinkage holds for the set $\Dam_l$ itself, so at each time in $[qT_0 + T_\textnormal{str-ref}, (q+1) T_0)$ there are at most $r$ damaged level-$l$ locations in $X$. This implies the first part of the inductive claim.)

    \textit{(iv) Program refresh.}
    We now analyze the program track. At the start of the program-refresh-sequence, the structure registers are correct everywhere in $X$ except on the freshly damaged sites just described. The leading program-EC gadget corrects the repetition code block of every program symbol whose input block contains at most $\tEC$ errors. With $\tEC = r + s_0 (w + 1) \tilde{r} > r$, this implies that all program symbols in initially undamaged macro-locations are correct.
    The copy/shift/majority-sequence then enacts a Toom-update on each site using the appropriate distance-$2M$ program registers from the next-nearest blocks to the east and north. In the ideal case this is exactly one block-wise Toom-update on the program layer, and thus removes inherited program defects from macro-locations on the boundary of the damage set. The only faulty operations in the encoded Toom-update are those supported on the newly damaged control sites (stemming from the control registers, e.g. $\tau$, being used to coordinate the program-refresh-sequence) together with their bounded local spread. By the definition of $w$ and $\tilde{r}$, these account for at most $w \tilde{r}$ faulty locations in the program-refresh gadget, leaving at most $s_0 w \tilde{r}$ errors per block. The trailing program-EC gadget corrects these. Thus, after the program refresh, the encoded program state is exactly the ideal Toom-refreshed program state (outside the prescribed damage set), and inside $X$ the only possible remaining program defects are supported on the same set of $\leq r$ freshly damaged level-$l$ locations.

    \textit{(v) Simulation phase.}
    It remains to establish correctness of the encoded $\Univ$ operation in the simulation phase. 
    Suppose that all four macro-locations of the $2 \times 2$ block to which $X$ belongs are undamaged, so that the above considerations apply to each of them.
    At the beginning of the simulation phase, both the structure and program control states are correct except on at most $r$ freshly damaged level-$l$ locations. We treat every data operation supported on one of the damaged sites as faulty. In this picture, structure and program defects cause transient effective data faults. Direct faults arising from the bad level-$l$ locations are included in the same damaged set. Since one damaged control location can alter at most $w$ operations in the gadget schedule, the simulation phase contains at most $w \tilde{r}$ effective faulty gadget locations, in addition to the at most $\tilde{r}$ faulty operations directly supported on the damaged sites. 
    Above, we chose the concatenated code to have sufficiently high distance to tolerate $\tEC$
    errors per block.
    
    The fault-tolerance properties (Def.~\ref{def:FT-conditions}) of $\Concat'$ hence imply that the $2 \times 2$ block including $X$ performs the encoded $\Univ$ gadget correctly in the sense of Def.~\ref{def:correct-macrolocs-and-exRecs}. 
    (By the propagation property of the leading error-correction gadget, the state directly after its completion carries at most $s_0$ errors per block for each effective fault, so the decoder can indeed be commuted through the remaining simulation schedule within the chosen error budget $\tEC = r + s_0 (\tilde{r} + w \tilde{r}) \geq s_0 (\tilde{r} + w \tilde{r})$; that is the decoded output at time $(q+1) T_0$ is exactly the ideal level-$(l{+}1)$ $\Univ$-update applied to the decoded state from directly after the leading error-correction gadget. If the block's input at time $q T_0$ contains at most $r$ errors per block (as in condition (ii) of Def.~\ref{def:correct-macrolocs-and-exRecs}), then the correctness property of the error-correction gadget additionally implies that the leading gadget does not modify the decoded state.
    The same counting bounds the faults affecting the trailing error-correction gadget, so at time $(q +1) T_0$ the state again carries at most $r$ errors per block.)
    Noting that control-state correctness outside the damage set holds at time $(q+1)T_0$ (because the damage set is defined to precisely contain the fresh control defects that may arise from the new badness events after the inherited defects have eroded, as shown above), this proves the inductive claim.

    Finally, every simulated level-$(l{+}1)$ update of $\R$ is implemented by a $2 \times 2$ block of exRecs, which each consist of $T_\mathrm{sim}$ consecutive macro-locations.
    If all of these macro-locations have undamaged output, then, by the inductive claim, every contained $2 \times 2$ block of macro-locations is correct and, as noted above, each cycle yields an output with at most $r$ errors per encoding block. The decoder can hence be passed through the exRecs cycle by cycle, implying correctness of the simulated update of $\R$ in the sense of Def.~\ref{def:correct-macrolocs-and-exRecs}, as claimed.
\end{proof}

\emph{Remark.} The constants in the above proof are not optimized, and the argument may overcount errors in several places. We expect the actual noise threshold to be substantially higher than what results from our proof.

\subsection{Sparsity of the damage set} \label{sec:damage-sparsity}

\begin{lemma} [Sparsity of the damage set] \label{lem:sparsity-of-damage}
    Suppose that $\Dam_l(0) = \emptyset$.
    (Recall from Sec.~\ref{sec:thm-statement} that, for a system that is initially prepared in a sufficiently correct state with the intended input, this holds on all levels $l \geq k_0$.)
    
    Then there exists a constant $B$ (independent of the level $l$ and of the system size) such that for any finite set $V$ of level-$l$ macro-locations, the fault-path sum restricted to trajectories in which at least one element of $V$ is damaged is bounded linearly as
    $$ \| E_l^\mathrm{dam}(V) \|_\diamond \leq B |V| \eta_l .$$
\end{lemma}
\begin{proof}
    Consider a simulation level $l$ and a finite set $V$ of level-$l$ macro-locations. We bound the fault-path contribution for the event that a fixed $v \in V$ lies in $\Dam_l$.
    By Def.~\ref{def:damage-set}, a site is in $\Dam_l(t)$ only if either its cluster was declared damaged due to a bad level-$l$ location occurring in the bounded cluster neighborhood around $v$, or it was already inside the damaged triangle-cover at the preceding macro-time and did not fully deflate away during the error-correction cycle.
    In any macro-noise-free neighborhood, the triangle cover deflates by at least one level-$l$ site per cycle.

    To turn this into a counting argument, we recursively trace the condition $v \in \Dam_l$ backwards in time. At each step, either the damage was freshly created via the first clause of Def.~\ref{def:damage-set} (i.e. via a bad macro-location in the surrounding cluster neighborhood), or the cluster neighborhood is macro-noise-free and $v$ must be \textit{excused} by earlier damage (at least two damaged predecessors in the corresponding backward Toom-neighborhood). Following Ref.~\cite{gacs2021}, we collect all such possible justifications in a graph whose nodes are level-$l$ macro-locations and whose edges connect a macro-location to the macro-locations from the previous time-step that can participate in such an excuse (`arrows' in Ref.~\cite{gacs2021}) and connect same-time macro-locations that can occur together in one excuse (`forks' in Ref.~\cite{gacs2021}). This graph has bounded degree because each cluster neighborhood $C^+$ contains only $\bigO(1)$ coarse-grained sites. Recursively choosing excuses and retaining a minimal connected subgraph yields an \textit{explanation tree} rooted at $v$ whose leaves correspond to bad macro-locations. By the standard combinatorial bound for Toom's rule dynamics~\cite{gacs2021, berman1988}, the number of possible explanation trees originating from $b$ leaves is at most $c_{\mathrm{Toom}}^b$ for some constant $c_{\mathrm{Toom}}$.

    Every fault path for which $v$ is damaged thus admits at least one explanation tree rooted at $v$. Its leaves form a set of witnessing bad macro-locations.
    Using the level-$l$ local noise bound (Def.~\ref{def:local-noise}), the fault-path sum associated with any such fixed set $F$ of $b$ witnessing bad macro-locations has norm at most $\eta_l^b$. 
    Summing over the possible explanation trees for damage at $v$, indexed by their witnessing leaf sets $F$, gives
    $$ \sum_{F: v \in \Dam_l} \| E(F) \|_\diamond \leq \sum_{b \geq 1} c_{\mathrm{Toom}}^b \eta_l^b = \frac{c_{\mathrm{Toom}} \eta_l}{1- c_{\mathrm{Toom}} \eta_l} $$
    which is $\bigO(\eta_l)$ whenever $\eta_l$ is sufficiently small (below threshold this holds uniformly in $l$).
    Finally, by triangle inequality and union bound over $v \in V$,
    \begin{align*}
        \| E_l^\mathrm{dam}(V) \|_\diamond
        &= \left\| \sum_{F: \Dam_l \cap V \neq \emptyset} E(F) \right\|_\diamond
        \\ &\leq \sum_{F: \Dam_l \cap V \neq \emptyset} \left\| E(F) \right\|_\diamond
        \\ &\leq \sum_{v \in V} \sum_{F: v \in \Dam_l} \| E(F) \|_\diamond
        \\ &\leq B |V| \eta_l
    \end{align*}
    for a suitable constant $B$.
\end{proof}

\subsection{Assembling the threshold theorem} \label{sec:proof-assembly}

We now combine the noise renormalization from Lem.~\ref{lem:threshold-thm-w-overlap} with the two damage set properties from Lems.~\ref{lem:correctness-outside-damage} and~\ref{lem:sparsity-of-damage} to prove the main theorem.

\begin{proof}[Proof of Thm.~\ref{thm:self-correction}]
    By assumption, the initial state is sufficiently correct (Sec.~\ref{sec:thm-statement}), so there exists a level $k_0$ with $\Dam_l (t = 0) = \emptyset$ for all $l \geq k_0$. This is used in the damage-sparsity bound.
    Note that the noise renormalization starts at the physical level, including all levels $l < k_0$, because even there this initial-state condition implicitly asserts that the global state has sufficiently few errors for the simulation from level $k_0$ upwards to be unaffected.    

    We check that each level of simulation reliably diminishes the noise strength.
    Let $\eta_0 = \eta$ be the physical noise strength. Consider a level $l$ with coarse-grained noise strength $\eta_l$.
    The time-slices of our QCA's trajectories can be mapped to the hard-coded steps of a gadget's circuitry in active quantum error correction protocols.
    Because each simulation block is a fixed-depth pattern of local operations, we can treat it as a gadget and apply the exRec formalism.
    We use the badness notion from Def.~\ref{def:goodness-for-macrolocs}, with macro-locations taking the role of exRecs in Lem.~\ref{lem:threshold-thm-w-overlap}.
    By that Lemma, the induced level-$(l{+}1)$ coarse-grained noise strengths obey
    $\eta_{l+1} \leq A \eta_l^{\alpha}$ with $\alpha = (\tbad +1)/R = 2$,
    where $A$ is a combinatorial constant determined by the maximal number of lower-level locations participating in one macro-location.
    Hence whenever the physical noise is below the constant threshold
    $$ \eta_\mathrm{th} = A^{- 1/(\alpha -1)}, $$
    iterating this bound yields
    $$ 
        \eta_l \leq \eta_\mathrm{th} (\eta/ \eta_\mathrm{th})^{\alpha^l},
    $$
    i.e. the coarse-grained noise strength decreases doubly-exponentially with the simulation level. The concatenation structure of this recursion builds itself by the self-simulation of Sec.~\ref{sec:self-simulation}, so the ideal decoded evolution on every level is again the evolution under $\R$, and this ideal evolution is realized on all levels outside the sparse damage sets.

    The coarse-grained noise strengths $\eta_l$ are related to logical errors via the damage sets.
    We consider the top-level $l = k > k_0$, at which the simulated circuit is (one iteration of) the desired logical computation $Q$, and let $V_Q$ denote the full set of space-time macro-locations (with $|V_Q| = ND$ being the number of space-time locations in the desired translation-invariant logical circuit). By Lem.~\ref{lem:correctness-outside-damage}, every fault path for which $V_Q \cap \Dam_k$ is empty produces the ideal decoded evolution on $V_Q$. The logical deviation is supported only on fault paths with $V_Q \cap \Dam_k \neq \emptyset$, and therefore
    $$
        \sup_{\|O\|\le 1} \left| \operatorname{Tr}\!\left[ O\left(\rho^{\mathrm{noisy}}_{\mathrm{out}}(Q)-\rho^{\mathrm{ideal}}_{\mathrm{out}}(Q)\right) \right] \right|
        \leq \| E_k^\mathrm{dam}(V_Q) \|_\diamond .
    $$
    By Lem.~\ref{lem:sparsity-of-damage},
    $ \| E_k^\mathrm{dam}(V_Q) \|_\diamond \leq B |V_Q| \eta_k$. Thus, to achieve output logical accuracy $\delta$, it suffices to choose $k$ such that $B \eta_k \leq \delta / N D$. By the doubly-exponential noise reduction above, $k = \bigO(\log \log (ND / \delta) )$ levels suffice.
\end{proof}

For the memory setting, $Q$ is the identity circuit for $D$ macro-steps. The bound above gives failure probability at most $\bigO(N D \eta_k)$. Since $\eta_k$ decreases with the number of encoding levels present in the system, the memory lifetime diverges as the system size increases.

\emph{Remark.} Up to now we have assumed evolution on \textit{periodic boundary conditions}.
The construction can be directly extended to \textit{open boundary conditions} as follows: We store four copies of the periodic-boundary-cells at each site and update each such layer according to the periodic-boundary transition rules, while making the update rule `fold over' at the edges to connect the four copies. This can be achieved in a translation-invariant system by conditioning on the locally detectable open boundary (detectable via absence of neighboring same-layer cells) and re-directing interactions to the appropriate layers as needed. The system then emulates a two-dimensional projection of a torus.

\section{Starting and stopping a self-correcting computation} \label{sec:computation-setup}

We have seen that the self-simulating transition rule $\R$ can fault-tolerantly encode an arbitrary quantum computation. This Section addresses the practical issue of initialization and read-out.

\textit{Initialization.}
Let us analyze the required initial state.
The self-correction of the automaton relies on the hierarchical organization that ensures error correction happens on all levels of the code, and as such it works autonomously only if the evolution starts in a state that is already encoded in the desired code. For the structure layer this demands a product state and for the simulation layer it depends on the specific concatenated code used via Lem.~\ref{lem:concat-code}. Magic-state distillation and verification protocols can be used to prepare these states~\cite{knill2004, bravyi2005} (and, if necessary, it may even be possible to control the preparation protocol within the self-correcting organization by carrying the corresponding instructions and required ancillas along with the circuit-scheme through all levels of the concatenation hierarchy).
Because the one-level encoder of such states has constant depth and can be applied recursively, the parallel depth required to prepare the top-level, $k$-times encoded initial state is $\bigO(\exp(k))$. There are doubly-logarithmically many concatenation levels in total (i.e. $k = \bigO(\log\log(N D / \delta))$ in terms of the target accuracy $\delta$ and the number of locations $N D$ in the desired circuit), so this yields polylogarithmic preparation depth for the fully encoded initial state.

\textit{Computation and read-out.}
Once the desired initial state is prepared, we evolve under the self-correcting transition rule in order to execute the desired computation. After some number of steps, the highest-level encoded circuit will have undergone exactly one round of the simulated logical circuit. At that point, the computation is finished and one should measure the final state and decode classically to obtain the computed \textit{output}. Alternatively, one can switch to different interactions so that the system stores the circuit's final state within its capacity as a quantum memory (see below).

To circumvent the need for externally timing the precise endpoint of the encoded computation, the distributed program of the intended logical circuit can initiate such a switch in the final layer of operations as follows: The simulation of higher levels should be executed regardless of the current-level switch. Note that the logical output is obtained solely from the highest decoding level and thus remains unaffected by faulty interaction switches in the lower simulation levels.
To improve efficiency, one could moreover drop the circuit execution entirely from such levels and keep it only on the highest level of the automaton. That level can be recognized, even in the periodic boundary setting, using a `top-level' marker which is removed whenever the adjacent cells do not carry the same marking, and which, after correct initialization and barring extensive concurrent faults, thus cannot persist anywhere except on the highest level (since that is the only level where cells have no neighbors other than themselves, cf.~\cite{gacs1986}).

\textit{Memory.}
Finally, we note that the same transition rule can also be used to build a self-correcting quantum memory by setting the simulated circuit to consist solely of identity gates.
Stopping the computation is not an issue in that case since there is no disadvantage in executing this `circuit' several times. In the thermodynamic limit this system can remember quantum states indefinitely. On any finite system, the memory lifetime will be determined by the logical noise within the largest level of encoding that still fits inside the boundaries.

\section{Continuous time: Self-correcting Lindbladian} \label{sec:async-scheme}

The protocol above yields a self-correcting QCA under globally synchronized, discrete-time updates.
We now describe an implementation in continuous time whose trajectories realize a version of the same automaton that allows for distant regions of the lattice to advance asynchronously.
The required dynamics are still local, translation-invariant and time-independent. 
We explain why the resulting dynamics remain robust under perturbation by arbitrary local unitary noise jumps of sufficiently small rate. 
In contrast to the discrete-time threshold theorem above, we do not spell out the continuous-time statements and proofs in full formal detail, instead we provide an outline of how the threshold analysis can be carried over to that setting, leaving a more detailed proof for future work.

\subsection{Asynchronous scheme} 
\label{sec:async-scheme-discrete}

We first define an \textit{asynchronous, discrete-time} automaton $\AsyncR$ that implements the same level-to-level self-simulation as the synchronous rule $\R$ from Def.~\ref{def:error-corrected-simulator}. The point is that the updates can be constrained to preserve causality while incurring only a constant slowdown.
Intuitively, synchronized evolution can be emulated within asynchronous settings by restricting the updates according to a `marching soldier' rule~\cite{nehaniv2004, gacs1986}. 
This requires a local mechanism which keeps track of how many updates have been performed at a site compared to its nearest neighbors and which ensures that evolution proceeds only in causality-preserving order, e.g. cells never advance more than one step ahead of their neighbors.

Because the control registers in our scheme (i.e. the structure and program variables) remain classical throughout the ideal evolution, we can dephase into the computational basis and copy these registers locally without affecting the decoded quantum state. 
To implement a marching soldier update mechanism, we enlarge each cell by a buffer copy of the control space. Whenever a local asynchronous update uses the control state on its support, it first copies the current state of these variables into the buffer and then updates according to the synchronous rule. Thus neighboring supports can read either the current or the one-step-old control state, so every update sees a concurrent slice of its structural neighborhood.
For the quantum data track, we note that the operations of each synchronous step of $\R$ can be decomposed into a layered, brickwork-geometry pattern of disjoint local channels (cf. the Remark following Def.~\ref{def:error-corrected-simulator}). Asynchronous reordering within any single temporal brick-layer is irrelevant since the disjoint channels commute. As a result, we do not need to store several copies of the quantum data to support causality-preserving updates.

We synchronize the evolution within each finite, $\bigO(M)$-sized region that simulates an individual higher-level operation. 
For a fixed brickwork decomposition of the updates of $\R$, we introduce at each site a finite clock register $\kappa \in \mathbb Z_{2 M^+}$ that records the site's progress through this decomposition. 
$M^+ = \bigO(M)$ is the diameter of one such \textit{synchronization neighborhood}, consisting of all current-level cells contributing to one simulated update on the level above. Since $M$ is fixed, this enlarges the onsite space (and thus the dependencies in Lem.~\ref{lem:self-consistent-parameters}) by only a constant factor. An attempted asynchronous update is successful exactly if its incident clock remains \textit{valid} (meaning that all sites whose support overlaps in the brickwork decomposition are synchronized up to $\kappa$-difference $0, \pm 1$) after increasing $\kappa$ by one. 
To be specific, when a successful update fires at location $(t,i,j)$, it applies the corresponding brickwork channel to the quantum data on the support of $(t,i,j)$, copies the control states on that support, updates the control registers exactly as in the synchronous rule and increments the local $\kappa$. With this, every completely valid execution of one macro-step induces exactly the same decoded higher-level update as the synchronous circuit.

Faults may push the clock registers outside the valid subspace. In fact, in two dimensions, the naive marching soldier scheme with finite-valued clocks is not robust because locally valid clock values can carry nonzero lag around loops, in which case they have no valid corresponding global update-counter-field and deadlock the dynamics~\cite{cook2008}. Following Ref.~\cite{wang1991}, we avoid this issue by synchronizing only within the finite support neighborhoods of individual simulated updates and concatenating/self-simulating afterwards. Because the synchronization neighborhood diameter is $M^+$, every valid modulo-$2M^+$ clock field in one such synchronization neighborhood has a unique, fully-valid corresponding field of absolute update counters.

To correct the thereby detectable clock errors, we introduce a local repair channel that establishes consistent clock registers throughout the synchronization neighborhood. Ref.~\cite{wang1991} gives an example of such a repair mechanism that purges local regions around clock inconsistencies and subsequently heals them in a reliable manner. For our construction, we use an alternative method that fits rather conveniently into our existing proofs via a consensus-finding variant of Toom's rule applied to the clock registers:
If the local north-east-center neighborhood of a site has invalid clock states, then the clock state at the center is replaced by the maximal value agreeing with the north-east neighbors up to $\pm 1$ (the maximum is taken over the respective least non-negative residues modulo $2M^+$); if no value is compatible with both north-east neighbors the repair channel resets the center clock to a fixed value and the site is treated as damaged. In each update, this step is repeated twice and followed by an inflation step in which clock registers transition to the maximal value of their south-west-neighborhood, unless already within $\pm 1$ of this value. As shown in Ref.~\cite{gacs1989} by noting that non-wrapping/contractible error regions are reliably diminished, the evolution under this rule reaches consensus in $\bigO(M^+)$ steps, meaning that even a completely erroneous clock configuration returns to the valid subspace within finite time.
With this, finite islands of clock-errors shrink under the same north-east-triangle geometry as under standard Toom's rule correction. 

We may therefore absorb erroneous clock registers into the triangle-cover formalism used in Secs.~\ref{sec:toom-geometry},~\ref{sec:damage-sets}. Outside the damaged regions, the asynchronous evolution on each synchronization neighborhood is a valid interleaving of the finite brickwork circuit simulating one macro-step of $\R$. We note that neighboring higher-level cells need not finish these steps simultaneously; instead the simulated higher-level dynamics are again asynchronous, implementing a recursive self-simulation of $\AsyncR$ at the next scale.
The correctness statement of Lem.~\ref{lem:correctness-outside-damage} 
then carries over
after enlarging the damage set by the effective faults (badness events) affecting the clock registers.

We emphasize that the asynchronous trajectory does not generally contain a full simulated time slice at any one physical instant. Instead, the time slices are represented as cuts through the local clock field that are spread over physical time. During the circuit-execution, an instantaneous physical configuration should thus not be interpreted as a state of the simulated synchronous circuit. In our setting, this does not prevent read-out, because we simulate the desired finite-depth circuit and then switch the top-level program to memory dynamics (see Sec.~\ref{sec:computation-setup}). Once all asynchronous parts (synchronization neighborhoods) have advanced beyond the macro-steps corresponding to the final circuit layer, later asynchronous updates only maintain the output state. Subsequent physical read-out followed by top-level decoding therefore gives the desired logical output (except on the threshold-suppressed bad fault paths).

\subsection{Dynamics in continuous time} \label{sec:Lindbladian-correction}

Let us now extend these results to a notion of passive error correction in \textit{continuous time}.

Consider an open quantum system whose evolution $\dot{\rho} = \mathcal{L}(\rho)$ is governed by a Lindbladian 
$   \mathcal{L}(\rho) = 
    \sum_\alpha \lambda_\alpha \left( 
                    \Phi_\alpha(\rho) - \rho
                \right) 
$
built from local, completely positive, trace-preserving update channels $\Phi_\alpha$. 
In our case, these channels have constant rates $\lambda_\alpha = 1$ and are given by the local attempted correction updates that define the asynchronous self-correcting automaton $\AsyncR$ (including local dissipative resets that reinitialize temporary ancillas or empty data sites).
Each $\Phi_\alpha$ acts on $\bigO(1)$ neighboring cells and is a translate of one of finitely many update channels (including the clock-repair mechanism). The resulting Lindbladian is local, time-independent and translation-invariant.
For related sufficient conditions for autonomous protection by engineered dissipation, see Ref.~\cite{lihm2018}. For the present construction, we instead argue for robustness through the recursive self-simulation and fault-path renormalization developed in the discrete-time picture above.
For any given time interval, expectation values of the final state can be computed by averaging over an ensemble of \textit{trajectories} where the probability of each trajectory depends on the update events along the taken path. 
Writing $\Phi_\alpha (\rho) = \sum_s K_{\alpha,s} \rho K^\dagger_{\alpha, s}$ with $\sum_s K^\dagger_{\alpha, s} K_{\alpha,s} = \Id\; \forall \alpha$, this yields independent Poisson processes with constant attempt rates for each jump $\alpha$, independent of the system state.

Let us introduce noise in this setting. The noise is realized by a perturbation $\delta \mathcal{L} $ on the Lindbladian.
We adopt an error model 
$\delta \mathcal{L}(\rho) = \sum_{(i,j)} \sum_\beta \mu_\beta ( N_\beta^{(i,j)} \rho N_\beta^{(i,j) \dagger} - \rho )$ 
with \textit{local unitary noise jumps} $N_\beta^{(i,j)}$ (that may depend on the site $(i,j)$) and with total rate $p = \sum_\beta \mu_\beta$. This encompasses, for example, local Pauli noise. 
(Note that this model perturbs the ideal dynamics by additional jumps whereas imperfections of the engineered jump operators $\Phi_\alpha$ themselves, i.e. non-unitary perturbations of $\mathcal{L}$, are not covered and are left for future work; cf. the Discussion in the main text.)
The trajectories are then described by a stochastic process with independent Poisson processes for attempted correction/repair jumps and for noise-jumps.

To complete the analysis, we bound the waiting time between successive \textit{successful} correction updates.
Recall that an attempted correction update at site $(i,j)$ is accepted if and only if it satisfies the marching soldier constraint introduced in Sec.~\ref{sec:async-scheme-discrete}. Thus, the instantaneous rate of successful correction jumps is $1$ at all local minima of the update surface determined by the clocks $\kappa$, and otherwise it is $0$. (For intuition, we remark that the roughening dynamics of the update surface fall under the Kardar-Parisi-Zhang (KPZ) universality class and are known to have an expected density of local minima that remains bounded away from zero even in the limit of infinite system size~\cite{korniss2001}.)
In our case, the synchronization neighborhoods are of finite diameter $M^+$ and all parameters are fixed constants, so a crude bound follows directly by treating the entire neighborhood as a finite-state continuous-time Markov chain. In every state, completion of the next required brickwork step is reached through a path of at most $\bigO(M^+)$ local updates, which are each attempted at unit rate. Hence there exists a constant $v_* > 0$, determined by $M^+$ and independent of the total system size, such that the next required update is always completed with probability at least $ v_*$ per unit time.
A local noise process with rate $p$ thus produces a fault before the next successful correction step with probability at most $\bigO(p / v_*) = \bigO(p)$.

In the local noise model of Sec.~\ref{sec:noise-model} this corresponds to an effective discrete fault strength $\eta = \bigO(p)$ per successful correction step.
In this way, noisy continuous evolution under the constructed Lindbladian maps to discrete asynchronous steps of $\AsyncR$ with constant local noise strength. 
Clock errors are already included in the damage set, and outside the damage set each synchronization neighborhood implements the correct higher-level update. Therefore the same exRec accounting as in Sec.~\ref{sec:threshold-proof} applies, giving $\eta_{l+1} \leq A_{\text{cont}} \eta_l^{(\tbad + 1) / R}$ with a modified constant $A_{\text{cont}}$. This argument indicates that Thm.~\ref{thm:self-correction} extends to the continuous-time evolution under this Lindbladian, i.e. that below a nonzero noise threshold, the logical errors are suppressed arbitrarily with increasing system size.